\documentclass{article}
\usepackage{arxiv}

\usepackage{natbib}

\usepackage[utf8]{inputenc} % allow utf-8 input
\usepackage[T1]{fontenc}    % use 8-bit T1 fonts
\usepackage{hyperref}       % hyperlinks
\usepackage{url}            % simple URL typesetting
\usepackage{booktabs}       % professional-quality tables
\usepackage{amsfonts}       % blackboard math symbols
\usepackage{nicefrac}       % compact symbols for 1/2, etc.
\usepackage{microtype}      % microtypography
\usepackage{lipsum}
\usepackage{graphicx}
\graphicspath{ {./images/} }
\usepackage{setspace}

\newcommand{\readd}[1]{{#1}}
\newcommand{\est}{\text{CEE}}
\usepackage{bbm}
\newcommand{\readdBiostat}[1]{#1}

\newcommand{\gammaplaceholder}{0.5}

\usepackage{amsmath}

\newtheorem{thm}{Theorem}
\newtheorem{lem}[thm]{Lemma}%\protect\lemmaname
\newtheorem{asu}{Assumption}
\newtheorem{rmk}{Remark}

\newtheorem{proof}{Proof}

\usepackage{tabularx,booktabs}
\usepackage{multirow}
\usepackage{graphicx}
\usepackage{threeparttable}
\usepackage[figuresright]{rotating}
\usepackage{changepage}

\usepackage{savesym}
\savesymbol{thead}
\usepackage{makecell} % for \makecell
\restoresymbol{makecell}{thead}

\usepackage{caption}
\usepackage{subcaption}
\usepackage[symbol]{footmisc}

\usepackage{algorithm}
\usepackage{algpseudocode}
\usepackage{ulem}
\usepackage[figuresright]{rotating}

\begin{document}
\doublespacing

% Title of paper
\title{Estimating Causal Effects for Binary Outcomes Using Per-Decision Inverse Probability Weighting}
% \author{\vspace{-5ex}}
% \date{\vspace{-5ex}}

% List of authors, with corresponding author marked by asterisk
\author{YIHAN BAO \\
\textit{Department of Statistics and Data Science, Yale University, New Haven, Connecticut, USA} \\
\texttt{yihan.bao@yale.edu} \\
\AND
LAUREN BELL \\
\textit{Leeds Institute of Clinical Trials Research, University of Leeds, Leeds, UK} \\
\texttt{lauren.bell@mrc-bsu.cam.ac.uk}\\
\AND
ELIZABETH WILLIAMSON \\
\textit{Department of Medical Statistics, London School of Hygiene and Tropical Medicine, London, United Kingdom} \\
\texttt{elizabeth.williamson@lshtm.ac.uk} \\
\AND
CLAIRE GARNETT\\
\textit{Department of Behavioural Science and Health, University College London, London, United Kingdom} \\
\texttt{c.garnett@ucl.ac.uk} \\
\AND
TIANCHEN QIAN$^*$\\ \textit{Department of Statistics, University of California, Irvine, California, USA}\\
% E-mail address for correspondence
\texttt{t.qian@uci.edu}}

% \author{YIHAN BAO \\
% \textit{Department of Statistics and Data Science, Yale University, New Haven, Connecticut, USA}\\
% LAUREN BELL \\
% \textit{Senior Research Fellow, Leeds Institute of Clinical Trials Research, University of Leeds, Leeds, United Kingdom} \\
% ELIZABETH WILLIAMSON \\
% \textit{Department of Medical Statistics, London School of Hygiene and Tropical Medicine, London, United Kingdom} \\
% CLAIRE GARNETT\\
% \textit{Department of Behavioural Science and Health, University College London, London, United Kingdom} \\
% TIANCHEN QIAN$^\ast$ \\ \textit{Department of Statistics, University of California, Irvine, California, USA}\\[4pt]
% % E-mail address for correspondence
% {t.qian@uci.edu}}

% Running headers of paper:
\markboth%
% First field is the short list of authors
{Y. Bao and others}
% Second field is the short title of the paper
{Estimating Causal Effects for Binary Outcomes Using Per-Decision Inverse Probability Weighting}

\maketitle

% Add a footnote for the corresponding author if one has been
% identified in the author list
\footnotetext{* To whom correspondence should be addressed.}

\begin{abstract}
{Micro-randomized trials are commonly conducted for optimizing mobile health interventions such as push notifications for behavior change. In analyzing such trials, causal excursion effects are often of primary interest, and their estimation typically involves inverse probability weighting (IPW). However, in a micro-randomized trial, additional treatments can often occur during the time window over which an outcome is defined, and this can greatly inflate the variance of the causal effect estimator because IPW would involve a product of numerous weights. To reduce variance and improve estimation efficiency, we propose \readdBiostat{two new estimators} using a modified version of IPW, which we call ``per-decision IPW''. \readdBiostat{The second estimator further improves efficiency using the projection idea from the semiparametric efficiency theory.} These estimators are applicable when the outcome is binary and can be expressed as the maximum of a series of sub-outcomes defined over sub-intervals of time. We establish the estimators' consistency and asymptotic normality. Through simulation studies and real data applications, we demonstrate substantial efficiency improvement of the proposed estimator over existing estimators. The new estimators can be used to improve the precision of primary and secondary analyses for micro-randomized trials with binary outcomes.}
%{Causal excursion effect; Inverse probability weighting; Log relative risk; Micro-randomized trial; Per-decision EMEE; Per-decision importance weighting}
\end{abstract}

\section{Introduction} \label{sec:introduction}

Mobile health interventions that provide individual-level support have been undergoing fast growth \citep{silva2015mobile}. Smartphone- or wearable-based messages and push notifications are used for promoting antiretroviral adherence, smoking cessation, and weight maintenance \citep{free2013effectiveness, battalio2021sense2stop, klasnja2021quality}. A novel experimental design, micro-randomized trial (MRT), was proposed in recent years for developing such interventions and has been increasingly used by domain scientists \citep{dempsey2015, liao2016sample, klasnja2015microrandomized, rabbi2018sara}. ``Micro-randomization'' refers to the repeated randomization among treatment options (e.g., delivering a message or not) in the trial. Such randomization can occur 100s or 1000s of times within each participant. For example, in Sense2Stop, a smoking cessation MRT, randomization occurs every minute for 10 days. In HeartSteps II \citep{spruijt2022advancing}, a physical activity MRT, randomization occurs every 5 minutes for 90 days. We describe a typical MRT design in Section \ref{subsec:data-structure}.

As a result, MRT data consist of longitudinal treatments and outcomes. These outcomes are called ``proximal outcomes'', because they are measured in the near term after each randomization and are often the direct target of the interventions. The primary and secondary analyses of an MRT focus on the causal excursion effect of the treatment on the proximal outcome and the effect moderation by baseline variables and time-varying contexts. Causal excursion effects provide information on whether the intervention is worth further investigation (through assessing the fully marginal effect) and, if so, how one might improve the intervention (through assessing effect moderation by contextual information). We formally introduce the causal excursion effect in Section \ref{subsec:cee-definition}.

Because of the repeated randomization, multiple treatments can occur during the time window over which the proximal outcome is defined. This time window is referred to as the ``proximal outcome window.'' For example, in the Sense2Stop MRT, the proximal outcome is whether the individual experienced a stress episode during the 120 minutes following a randomization, whereas randomizations occur every minute. Thus, in the 120-minute proximal outcome window, there are 119 chances for additional treatments to occur. In the HeartSteps II MRT, the proximal outcome is whether the individual takes enough steps to constitute an activity bout during the 30 minutes following a randomization, whereas randomizations occur every 5 minutes. Existing estimators use inverse probability weighting (IPW) to account for the additional treatments in the proximal outcome window \citep{dempsey2020stratified, qian2021estimating}. These estimators effectively discard the data from the time points where additional treatments do occur during the proximal outcome window, and they upweight the data from the time points where additional treatments do not occur during the proximal outcome window. Such estimators can have high variance because (i) any time point where additional treatments occur in the proximal outcome window will not contribute to the estimation; (ii) the product of the IP weights can be large for the time points where no additional treatments occur in the proximal outcome window. Consequently, a large sample may be necessary to identify promising interventions with statistical confidence.

Our main contribution is \readdBiostat{two more efficient estimators} for such settings---MRTs where additional treatments may occur during the proximal outcome window. We focused on MRTs with binary proximal outcomes, which are common in practice \citep{rabbi2018sara, klasnja2021quality}. For such MRTs, the current state-of-the-art estimator was developed by \citet{qian2021estimating}. We proposed \readdBiostat{two more efficient estimators} for the causal excursion effect and established its consistency and asymptotic normality. The relative efficiency of the proposed estimators over that of \citet{qian2021estimating} in most simulation scenarios ranges between 1.08 and 1.45 or even higher, which roughly translates to savings in sample size between 7\% and 31\% or even larger. \readdBiostat{The efficiency improvement is greater when there is a higher chance of one or more additional treatments occurring in the proximal outcome window; for example, when the proximal outcome window consists of more decision points or when the randomization probability is larger. When the proximal outcome window doesn't contain additional decision points (i.e., when the proximal outcome is measured immediately), the proposed method is equivalent to the EMEE in \citet{qian2021estimating}. }
% The relative efficiency improvement will further increase if the proximal outcome window length is larger, or the baseline success probability (i.e., the success probability of the proximal outcome under no treatment) is larger.

Our estimators are applicable as long as the binary proximal outcome can be expressed as the maximum of a series of sub-outcomes defined over sub-intervals of time. This is often the case because a binary proximal outcome is often defined as whether a specific event occurred during a pre-specified time window. Examples of such events include a stress episode during a 120-minute window \citep{battalio2021sense2stop}, a physical activity bout during a 30-minute window \citep{liao2020personalized}, or an interaction with a smartphone app during a 3-day window \citep{bell2020notifications}.
The key intuition for the efficiency improvement is to realize that the value of a binary proximal outcome is determined as soon as the event of interest occurs, and there is no need to wait till the end of the proximal outcome window to assign this value. Thus, \readdBiostat{when performing estimation,} one does not need to multiply together all the IP weights but only those weights before the occurrence of the event.
Consequently, (i) some data points that would otherwise have been removed (due to being multiplied with a weight 0) are retained; and (ii) for the time points where no additional treatments occur in the proximal outcome window, the product of the IP weights becomes smaller (due to fewer terms in the product). This leads to the efficiency gain.
We refer to this modification of IPW as ``per-decision IPW'', because the inclusion of each weight is determined by whether the event has occurred by each decision point (i.e., a time point where randomization occurs). A related modification idea of the IP weights for efficiency improvement was proposed by \citep{precup2000eligibility} in the reinforcement learning literature for the setting with a single continuous outcome (i.e., the cumulative reward) that is the sum of sub-outcomes. Our method applies to a different setting where the outcome itself is longitudinal binary. \readdBiostat{Furthermore, our second estimator further improves efficiency by leveraging the projection idea from the semiparametric efficiency theory \citep{bickel1993efficient,murphy2001marginal}.}

The rest of the paper is organized as follows. In Section \ref{subsec:drinkless} we introduce the Drink Less MRT, our motivating example. In Section \ref{sec:definition} we define the causal excursion effect using the potential outcomes framework and provide identifiability results. In Section \ref{sec:method} we present the first new estimator, pd-EMEE, which capitalizes on the per-decision IPW idea, and the asymptotic properties of this estimator. \readdBiostat{In Section 4, we present the second new estimator, pd-EMEE2, which further improves efficiency using the projection idea from the semiparametric efficiency theory, and derive the asymptotic properties of this estimator.} In Section \ref{sec:simulation} we present results from simulation studies on the consistency and efficiency of the two proposed estimators. In Section \ref{sec:application} we illustrate the proposed method using the Drink Less MRT \readdBiostat{and HeartSteps I MRT datasets}. Section \ref{sec:discussion} concludes with a discussion.

\subsection{Data Example: Drink Less MRT}
\label{subsec:drinkless}

Drink Less is a smartphone app to help users reduce harmful alcohol consumption \citep{garnett2019development, garnett2021refining}. An MRT was conducted (embedded as one arm in a three-arm RCT) to assess and improve the effectiveness of a push notification intervention component that targets user engagement \citep{bell2020notifications}. Each of the 349 participants in the MRT was randomized every day at 8 pm for 30 days, each time with a probability of 0.6 to receive an engagement intervention in the form of a push notification and a probability of 0.4 to receive nothing. The content of each push notification was randomly selected from a pool of messages that encouraged the user to keep track of their drinking behavior using the app. The researchers have identified a large causal effect on app engagement within 1 hour of the intervention delivery and a smaller effect within 24 hours of the intervention delivery \citep{bell2023notifications}. \readdBiostat{In this illustration of our method, we aim to explore how this effect depends on the choice of the time window, from 1 day to 5 days. We define the binary proximal outcome as whether the user opens the app in the subsequent $\Delta$ days after each randomization at 8 pm. For example, when $\Delta = 3$, during the 3-day proximal outcome window there are two additional occurrences of randomization (at 8 pm in the next two days)}, and the probability of at least one additional treatment being delivered in the proximal outcome window is 0.84 ($=1-0.4^2$). Therefore, on average 84\% of the data points would be discarded if one were to use the estimator by \citet{qian2021estimating}. Our proposed methods will substantially improve the estimation efficiency.

\subsection{Data Example: HeartSteps I MRT}
\label{subsec:heartsteps}

\readdBiostat{HeartSteps I MRT aimed at developing mHealth interventions for increasing physical activity among sedentary adults \citep{klasnja2015microrandomized}. In illustrating our method, we focus on one of the intervention components in the MRT, the activity suggestions. Each of the 37 participants was randomized five times a day for 42 days, each time with a probability of 0.6 to receive a contextually tailored activity suggestion in the form of a push notification, and a probability of 0.4 to receive nothing. The researchers have identified a substantial causal effect on steps taken within 30 minutes of the intervention delivery \citep{klasnja2018efficacy}. We aim to explore whether the activity suggestion has a prolonged effect in the next 24 hours. We define the binary proximal outcome as whether the participant takes more steps than a pre-determined threshold (e.g., 8,000 steps) in the 24 hours following an intervention delivery. During the 24-hour proximal outcome window, there are four additional occurrences of randomization, and thus our proposed method can improve the estimation efficiency.}

\section{Definitions and Assumptions}
\label{sec:definition}

\subsection{Notation and MRT Data Structure}
\label{subsec:data-structure}

Assume there are $n$ individuals in an MRT. Each individual is enrolled for $T$ time points at which treatments can be delivered. We will refer to these time points as decision points. In the presentation below we assume that $T$ is the same for all individuals, but our method generalizes easily to settings where $T$ varies across individuals. We use capital letters without subscript $i$ to denote observations of a generic individual. Denote the treatment assignment at decision point $t$ by $A_t$. For simplicity, we assume that the treatment assignment is binary, with $A_t = 1$ indicating treatment and 0 indicating no treatment. Denote observations collected after decision point $t-1$ and before decision point $t$ by vector $X_t$. Let $X_{T+1}$ denote the information observed after decision point $T$. As will be discussed shortly, the $X_t$'s include both covariates and outcomes. For each individual, the entire trajectory of observed information is $O = (X_{1},A_{1},X_{2},A_{2},\ldots,X_{T},A_{T},X_{T+1})$. We assume that data from different individuals are independently and identically distributed draws from an unknown distribution $P_0$. Unless noted otherwise, all expectations are taken with respect to $P_0$.
% For any generic function $f(\cdot)$ on the generic observation $O$, we use $\mathbbm{P}_n f(O)$ to denote the sample mean $n^{-1}\sum_{i=1}^n f(O_i)$.

When designing the intervention (before the MRT starts), the researchers may deem it unsafe or unethical to deliver a push notification at certain decision points---for example, when the individual is driving. This is captured by the notion of ``being (un)available to be randomized'' \citep{boruvka2018assessing}. Formally, $X_t$ contains an indicator $I_t$, where $I_t = 1$ if the individual is available to be randomized at decision point $t$ and $I_t = 0$ otherwise. If $I_t$ = 0, $A_t = 0$ deterministically. 

We use overbar to denote a sequence of variables: for example, $\bar{A}_{t}=(A_{1},A_{2},\ldots,A_{t})$.
% We use overbar and double subscript to denote a sequence of variables between two decision points: for example, for $t_{2} \geq t_{1}$, $\bar{A}_{t_{1}:t_{2}}=(A_{t_{1}},A_{t_{1}+1},\ldots,A_{t_{2}})$.
Information accrued from an individual up to decision point $t$ is denoted by $H_t = (X_1,A_1,X_2,A_2,\ldots,$ $X_{t-1},A_{t-1},X_t) = (\bar{X}_t,\bar{A}_{t-1})$. At each $t$, $A_t$ is randomized with randomization probability depending on $H_t$, denoted by $p_t(H_t) = P(A_t = 1 \mid H_t)$.

We consider MRTs where the proximal outcome takes a binary value and is defined over a window of length $\Delta$ (of decision points) following each decision point; the same proximal outcome was considered by \citet{qian2021estimating}. Here $\Delta \geq 1$ is a fixed positive integer. This window is the ``proximal outcome window'' referred to in Section \ref{sec:introduction}. Formally, suppose $Y_{t,\Delta}$ is a function of $H_{t + \Delta}$; i.e., the definition of the proximal outcome for decision point $t$ can depend on anything observed up to  \readdBiostat{time $t+\Delta$ and} including $X_{t+\Delta}$. For example, if $\Delta = 1$, then $Y_{t,\Delta}$ is a function of $H_{t + 1}$, i.e., is observed before $A_{t+1}$. If $\Delta = 2$, then $Y_{t,\Delta}$ is a function of $H_{t+2}$, in which case additional treatments (precisely $A_{t+1}$ in this case) can occur before the value of $Y_{t,\Delta}$ is observed. Figure \ref{fig:data-structure} depicts for the case of $\Delta = 3$ a few consecutive decision points and the variables we defined, except for the $R$-variables which will be defined in Section \ref{subsec:maximum-property}. 

\begin{figure}[htbp]
    \begin{center}
        \caption{Illustration of the data structure and the maximum property with $\Delta = 3$.}
    %\centering
        \includegraphics[width=\textwidth]{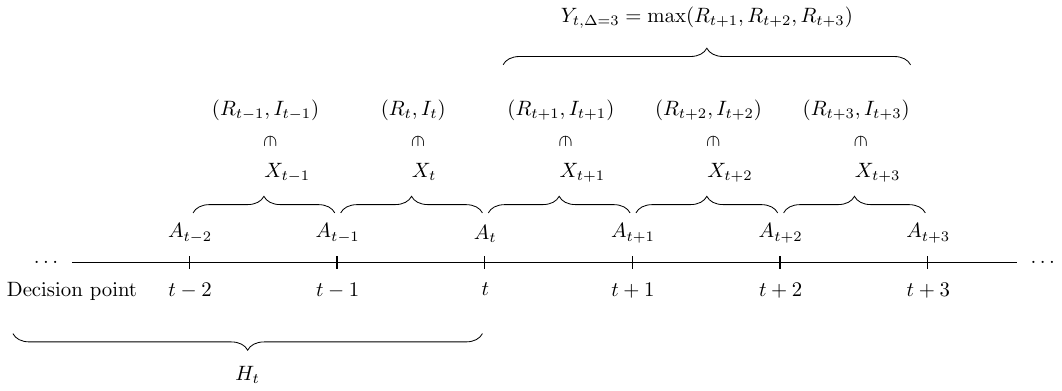}
        \label{fig:data-structure}
    \end{center}
\end{figure}

In the Drink Less MRT, the decision point is every day at 8 pm. Each participant is in the study for $T = 30$ days. $A_t$ is whether a push notification is delivered at 8 pm on day $t$. $X_t$ includes baseline variables and time-varying variables such as the number of days in the study and the user's previous interaction with the app. The researchers deemed that it is always appropriate to deliver push notifications, so $I_t = 1$ for all $t$. The proximal outcome, $Y_{t,\Delta}$, is whether the participant logged in to the app in the subsequent $\Delta$ days, where we consider $\Delta = 1,\ldots, 5$ in separate analyses.

\readdBiostat{In the HeartSteps I MRT, the decision point occurs five times a day during the morning commute, lunchtime, mid-afternoon, evening commute, and after dinner. Each participant is in the study for $T = 210$ decision points (which amounts to 42 days). $A_t$ is whether an activity suggestion is delivered at decision point $t$. $X_t$ includes baseline variables and time-varying variables such as the participant's previous step count. $I_t = 0$ if the participant is driving, already walking, or with no internet connection. The proximal outcome, $Y_{t,\Delta}$, is whether the participant takes more steps than a given threshold in the next 24 hours (so $\Delta = 5$). We consider a few thresholds from 5,000 to 12,000 steps in separate analyses.}

\subsection{Potential Outcomes and Causal Excursion Effects}
\label{subsec:cee-definition}

For a mobile health intervention, causal quantities of interest include whether the intervention has an impact on an individual's behavior and how this impact varies with contextual information such as the number of days in the study. 
% We consider the task of estimating causal excursion effects defined by \citet{qian2021estimating}. Here we review the causal excursion effect defined using the potential outcomes framework \citep{rubin1974estimating, robins1986new}.
To define the causal effects precisely we use the potential outcomes framework \citep{rubin1974estimating, robins1986new}.
Lowercase letters represent instantiations (non-random values) of the corresponding capital letter random variable. For example, $a_t$ is an instantiation of treatment $A_t$. For every individual, denote by $X_t(\bar{a}_{t-1})$ the $X_t$ that would have been observed at decision point $t$ if the individual had been assigned a treatment sequence of $\bar{a}_{t-1}$ prior to $t$. The potential outcome of $H_t$ under $\bar{a}_{t-1}$ is
\begin{equation}
    H_{t}(\bar{a}_{t-1}) = \{X_1, a_1, X_2(a_1), a_2,  X_3(\bar{a}_2), a_3, \ldots, X_{t-1}(\bar{a}_{t-2}), a_{t-1}, X_{t}(\bar{a}_{t-1})\}.
\end{equation}
The potential outcome of $Y_{t,\Delta} = y(H_{t+\Delta})$ is $Y_{t, \Delta}(\bar{a}_{t+\Delta-1})$.
%for any $\bar{a_t} \in \mathcal{A}_t = \{(a_1, \ldots, a_t) \mid a_s \in \{0,1\}, s = 1, \ldots, t \}$.

% For each $t$, the potential outcome for the sub-outcome $R_t$ is denoted by $R_t(\bar{a}_{t-1})$, and the potential outcome for the proximal outcome is $Y_{t, \Delta}(\bar{a}_{t+\Delta-1})$.
% For reason that will be later justified in Section \ref{subsec:identification}, we write $A_2(A_1)$ as $A_2$, $A_3(\bar{A}_2)$ as $A_3$ and so on, and $A_T(\bar{A}_{T-1})$ as $A_T$. 

% Because $Y_{t,\Delta}$ is binary, 
We consider the causal excursion effect (CEE) of $A_t$ on $Y_{t,\Delta}$ on the log relative risk scale \citep{qian2021estimating}, defined as
\begin{equation}
    \est \left\{S_t(\bar{A}_{t-1}) \right\}=\log \frac{E\left\{Y_{t, \Delta}(\bar{A}_{t-1}, 1, \bar{0}_{\Delta - 1}) \mid S_t(\bar{A}_{t-1}), I_t(\bar{A}_{t-1})=1\right\}}{E\left\{Y_{t, \Delta}(\bar{A}_{t-1}, 0, \bar{0}_{\Delta - 1}) \mid S_t(\bar{A}_{t-1}), I_t(\bar{A}_{t-1})=1\right\}}.
    \label{eq:causal-effect-def}
\end{equation}
Here $\bar{0}_{\Delta - 1}$ is a vector of 0's of length $\Delta-1$; $(\bar{A}_{t-1}, 1, \bar{0}_{\Delta - 1})$ and $(\bar{A}_{t-1}, 0, \bar{0}_{\Delta - 1})$ are sequences of treatment assignments of length $t+\Delta-1$. $S_t(\bar{A}_{t-1})$ is a length-$p$ vector of summary variables from $H_t(\bar{A}_{t-1})$, usually chosen by the domain scientist based on the scientific question regarding effect modification. In Drink Less MRT, we will first set $S_t(\bar{A}_{t-1}) = 1$ to obtain the fully marginal effect (i.e., averaged over all participants and all days in the study) of a push notification on the app use in the subsequent $\Delta$ days. We will then set $S_t(\bar{A}_{t-1}) = (1, \text{Day}_t)$ to assess how the effect is moderated by days in the study.

The causal excursion effect \eqref{eq:causal-effect-def} is commonly used in the primary and secondary analyses for MRTs \citep{nahum2021translating, klasnja2021quality}. It differs from most causal effects under time-varying treatments such as marginal structural models and structural nested mean models \citep{robins1994correcting, robins2000marginal}. Instead of contrasting fixed treatment trajectories, \eqref{eq:causal-effect-def} contrasts between two stochastic treatment trajectories, $(\bar{A}_{t-1}, 1, \bar{0}_{\Delta - 1})$ and $(\bar{A}_{t-1}, 0, \bar{0}_{\Delta - 1})$. That is, \eqref{eq:causal-effect-def} is a causal contrast between two excursions from the treatment policy: following the treatment policy up to time $t - 1$ and then deviating from the policy to assign $A_t = 1$ and no treatment for the next $\Delta-1$ decision points, and following the treatment policy up to time $t-1$ and then deviating to assign $A_t = 0$ and no treatment for the next $\Delta-1$ decision points. \eqref{eq:causal-effect-def} is conditional on $I_t(\bar{A}_{t-1})=1$ and $S_t(\bar{A}_{t-1})$, meaning that it concerns the subpopulation who, after following the treatment policy up to time $t - 1$, are available to be randomized at $t$ and within a stratum defined by $S_t(\bar{A}_{t-1})$. In practice we often impose a model on how $\est \left\{S_t(\bar{A}_{t-1}) \right\}$ depends on $S_t(\bar{A}_{t-1})$ to pool across $S_t$-strata.

$\est \left\{S_t(\bar{A}_{t-1}) \right\}$ depends on the treatment policy in the MRT, because the variables in $H_t(\bar{A}_{t-1}) \setminus S_t(\bar{A}_{t-1})$ are marginalized over in the conditional expectations in \eqref{eq:causal-effect-def}. Such dependence is scientifically desirable for two reasons. First, the causal excursion effect approximates the treatment effect in a real-world implementation. A good MRT policy would already have incorporated implementation considerations such as burden and feasibility through the choice of randomization probability and availability criteria, so a feasible policy will not deviate too far from the MRT policy. Second, the causal excursion effect indicates the effective deviations from the MRT policy and how the policy might be improved. A fully marginal effect (by setting $S_t(\bar{A}_{t-1}) = 1$ in \eqref{eq:causal-effect-def}) indicates whether a treatment is worth further investigation, and an effect modification analysis (by setting $S_t(\bar{A}_{t-1})$ to be certain time-varying covariates in \eqref{eq:causal-effect-def}) indicates whether the MRT policy should be modified to depend on time-varying covariates. In addition, the dependence of $\est \left\{S_t(\bar{A}_{t-1}) \right\}$ on the MRT policy resembles the primary analysis in factorial designs and allows one to design trials with a higher power to detect meaningful effects. Related marginalization ideas were considered by \citep{robins2004optimal, neugebauer2007causal}. See \citet{boruvka2018assessing}, \citet{qian2021estimating}, and \citet{qian2021micro} for more discussion on causal excursion effect and its use in MRT. See \citet{guo2021discussion} for a comprehensive comparison of various causal estimands under time-varying treatments.

% For example, in Sense2Stop MRT with $\Delta = 120$, $Y_{t, \Delta}(\bar{A}_{t-1}, 1, \bar{0}_{\Delta - 1})$ represents whether the individual would have experienced stress in the 120 minutes following time $t$ if a treatment were to be sent at $t$ and no other treatment were to be sent for the subsequent 119 minutes; similarly for $Y_{t, \Delta}(\bar{A}_{t-1}, 0, \bar{0}_{\Delta - 1})$.  \tq{Note-to-self: Also use HeartSteps II to illustrate. Make sure we convince the reader that this estimand is meaningful and of scientific interest.}

\subsection{Maximum Property of the Binary Proximal Outcome}
\label{subsec:maximum-property}

We consider settings where the binary proximal outcome $Y_{t,\Delta}$ can be expressed as the maximum of a series of sub-outcomes defined over sub-intervals of time. Formally, let $R_{t+1}$ denote a sub-outcome---that is, a binary response observable after decision point $t$ and before $t+1$ and thus is a function of $X_{t+1}$. We assume that $Y_{t,\Delta}$ satisfies the \textit{maximum property}:
\begin{align}
    Y_{t,\Delta} = \max(R_{t+1}, R_{t+2}, \ldots, R_{t+\Delta}). \label{eq:maximum-property}_{}
\end{align}
That is, $Y_{t, \Delta} = 1$ if and only if any of the sub-outcomes between $t$ and $t+\Delta$ is 1. Figure \ref{fig:data-structure} shows how $Y_{t,\Delta}$ is related to the sub-outcomes when $\Delta = 3$. We assume that $R_{t+\Delta}$ and $Y_{t, \Delta}$ are well-defined even when $t+\Delta > T+1$; this is satisfied when the researchers continue to follow up trial participants for $\Delta$ time points after decision point $T$.

Equation \eqref{eq:maximum-property} holds if the binary outcome $Y_{t,\Delta}$ is whether an event of interest occurs between $t$ and $t+\Delta$ and such events are ``instantaneous''. In the Drink Less MRT example with $\Delta = 3$, the event is a user logging into the app. Logging in between $t$ and $t+3$ is equivalent to logging in between $t$ and $t+1$, or between $t+1$ and $t+2$, or between $t+2$ and $t+3$. Thus, in this example \eqref{eq:maximum-property} is satisfied with $R_{t+1}$ denoting whether the individual logging into the app between 8 pm of day $t$ and 7:59 pm of day $t+1$, in which case we have $Y_{t,\Delta = 3} = \max(R_{t+1}, R_{t+2}, R_{t+3})$.

\readdBiostat{If, however, $Y_{t,\Delta}$ is based on a cumulative event, then \eqref{eq:maximum-property} may not hold. For example, in the HeartSteps I MRT, the binary proximal outcome is whether a participant takes more than, say, 10,000 steps in the next 24 hours (a proximal outcome window spanning over $\Delta = 5$ decision points). In this case, no sub-outcomes $R_{t+1}$ can satisfy \eqref{eq:maximum-property} because the total step count is cumulative instead of instantaneous. We generalize the maximum property to accommodate such settings in Supplementary Material \ref{appen:sec:generalized-max-property} by introducing a double-subscript version of the sub-outcome $R$. However, this generalization would make the notation too cumbersome without adding substantial novelty, so we focus on the current version of the maximum property \eqref{eq:maximum-property} in the sequel when we introduce the methods.}

\subsection{Identification of Parameters}
\label{subsec:identification}

%\textit{provide identifiability result}
%\textit{refer to the appendix for proof}
To express the causal excursion effect using observed data, we assume the following.

\begin{asu}[Stable unit treatment value assumption (SUTVA)]
  \label{assumption1}
  The data observed are equal to the set of potential outcomes corresponding to the treatment assignment actually received, and one's potential outcomes are not affected by others' treatment assignments. That is, $X_t = X_t(\bar{A}_{t-1})$ for any $t \leq T$. This implies $R_t(\bar{A}_{t-1}) = R_t$ and $Y_{t,\Delta}(\bar{A}_{t+\Delta-1})=Y_{t,\Delta}$.
\end{asu}

%\textit{Assumption 2 (sequential ignorability)}: 

\begin{asu}[Sequential ignorability]
  \label{assumption2}
  For $1 \leq t \leq T$,
  the potential outcomes $\{R_j(\bar{a}_{j-1}):\bar{a}_{j-1}\in\{0,1\}^{\otimes j-1}, (t+1) \leq j\leq T\}$ are independent of $A_t$ given $H_t,I_t=1$. This implies that $\{Y_{j,\Delta}(\bar{a}_{j+\Delta-1}):\bar{a}_{j+\Delta-1}\in\{0,1\}^{\otimes j+\Delta-1}, t \leq j\leq T\}$ are independent of $A_t$ given $H_t,I_t=1$.
\end{asu}

%\textit{Assumption 3 (positivity)}:
\begin{asu}[Positivity]
  \label{assumption3}
  $P(A_t=1\mid H_t,I_t=1)$ is bounded away from 0 and 1.
\end{asu}

\begin{asu}[Maximum property]
  \label{assumption4}
  $Y_{t,\Delta}(\bar{a}_{t+\Delta-1}) = \max(R_{t+1}(\bar{a}_t), R_{t+2}(\bar{a}_{t+1}), \ldots, R_{t+\Delta}(\bar{a}_{t+\Delta-1}))$ for $1 \leq t \leq T$.
\end{asu}

Assumption \ref{assumption1} is standard for causal inference settings without interference. It may be violated if the mobile health intervention incorporates community features or the participants are clustered, in which case a causal inference framework that allows interference should be used instead \citep{hudgens2008toward,shi2023assessing}. (There is no community feature in the Drink Less app, our motivating example.) Assumption \ref{assumption2} is always valid for MRTs because the randomization probability of $A_t$ depends (at most) on $H_t$ by study design. Assumption \ref{assumption3} is valid as long as the researchers do not set the randomization probability to exactly 0 or 1 when $I_t=1$. Assumption \ref{assumption4} is the potential outcomes version of the maximum property \eqref{eq:maximum-property}.

% key assumption enabling our estimator to achieve improved efficiency. It is valid for MRTs where $Y_{t,\Delta} = \max(R_{t+1}, R_{t+2}, \ldots, R_{t+\Delta})$, i.e., when the maximum property \eqref{eq:maximum-property} holds. Assumption 4 appears complicated but the intuition behind it is straightforward: if the first sub-outcome that equals to 1 after decision point $t$ is $R_{t+u}$ (for some $1 \leq u \leq \Delta-1$), then the treatments after $u$ will not alter the value of $Y_{t,\Delta}$ (because $Y_{t,\Delta}$ has to be 1 by the maximum property). Assumption 4 formalizes this intuition using the potential outcomes notation. The treatment trajectories in the potential outcome of $R_{t+u}$ in Assumption 4 end with $\bar{0}_{u-1}$, because the treatment trajectories in the potential outcome of $Y_{t,\Delta}$ in the causal excursion effect definition \eqref{eq:causal-effect-def} end with $\bar{0}_{\Delta - 1}$.

Let $\mathbbm{1}(\cdot)$ denote the indicator function, and define $0^0 = 1$. In Supplementary Material \ref{appen:sec:causal-assumption-and-identifiability-proof}, we prove that under the above 4 assumptions the casual excursion effect in (\ref{eq:causal-effect-def}) can be expressed using observed data distribution:
\begin{align}
    \hspace{-2em}
    \est\left\{S_t(\bar{A}_{t-1})\right\} 
    = \log \frac{E \bigg[ E\left\{\prod_{j=t+1}^{t+\Delta-1}\left[\frac{\mathbbm{1}(A_{j}=0)}{1-p_{j}\left(H_{j}\right)}\right]^{\mathbbm{1}(\max_{1\leq s\leq j-t}R_{t+s}=0)} Y_{t,\Delta}\mid A_{t}=1,H_{t},I_{t}=1\right\}  \mid S_t, I_t = 1 \bigg]} {E \bigg[ E\left\{ \prod_{j=t+1}^{t+\Delta-1}\left[\frac{\mathbbm{1}(A_{j}=0)}{1-p_{j}\left(H_{j}\right)}\right]^{\mathbbm{1}(\max_{1\leq s\leq j-t}R_{t+s}=0)} Y_{t,\Delta}\mid A_{t}=0,H_{t},I_{t}=1\right\}  \mid S_t, I_t = 1 \bigg]}.
    \label{eq:identification}
\end{align}
% The indicator function in the power sets the corresponding weight to 1 once the event has occurred (i.e., once some $R_{t+s} = 1$).
We define $\prod_{j=t+1}^{t+\Delta-1}\left[\frac{\mathbbm{1}(A_{j}=0)}{1-p_{j}\left(H_{j}\right)}\right]^{\mathbbm{1}(\max_{1\leq s\leq j-t}R_{t+s}=0)} = 1$ if $\Delta = 1$. Equation \eqref{eq:identification} will be used to establish the asymptotic property of our estimator. We will slightly abuse the notation and denote the right-hand side of (\ref{eq:identification}) by $\est(S_{t})$.

\section{pd-EMEE: An Estimator Using Per-Decision Inverse Probability Weights}
\label{sec:method}

Consider a linear model for the causal excursion effect in \eqref{eq:identification}:
\begin{equation}
    \est(S_{t}) = S_t^T \beta \quad \text{for all } 1 \leq t \leq T,
    \label{eq:linear-model}
\end{equation}
where $\beta$ is $p$-dimensional. Assume the first entry of $S_t$ is 1 so that $\beta$ contains the intercept. For example, when $S_t = 1$, $\est(S_t) = \beta_0$ is the fully marginal effect. When $S_t = (1, Z_t)$ for some time-varying covariate $Z_t$, $\est(S_t) = \beta_0 + \beta_1 Z_t$, where $\beta_0$ denotes the effect when $Z_t = 0$ and $\beta_1$ denotes the increase in the effect when $Z_t$ increases by 1. Model \eqref{eq:linear-model} also accommodates flexible time-varying effects when $S_t$ contains basis functions. We assume that \eqref{eq:linear-model} is a correct model and our goal is to estimate the unknown $p$-dimensional vector $\beta$.

\readdBiostat{Our first estimator for $\beta$ is denoted by $\hat\beta$ and referred to as pd-EMEE (a name we will explain later)}. It requires a working model for $E\{Y_{t, \Delta}(\bar{A}_{t-1}, 0, \bar{0}_{\Delta - 1}) \mid H_{t}, I_{t}=1 \}$, which characterizes the expected outcome under no treatment between $t$ and $t + \Delta$ given the history $H_t$ and $I_t = 1$. We use $\exp \{g(H_t)^T \alpha \}$ as the working model, where $g(H_t)$ is a vector of features constructed from  $H_t$ and $\alpha$ is a $q$-dimensional nuisance parameter. As will be established in Theorem \ref{thm:CAN}, $\hat\beta$ is consistent and asymptotically normal regardless of whether the working model $\exp \{g(H_t)^T \alpha \}$ is a correct model for $E\{Y_{t, \Delta}(\bar{A}_{t-1}, 0, \bar{0}_{\Delta - 1}) \mid H_{t}, I_{t}=1 \}$. 

The estimator $\hat\beta$, along with the nuisance estimator $\hat\alpha$, is obtained by solving an estimating equation $\frac{1}{n}\sum_{i=1}^n U_i(\alpha,\beta) = 0$, with the estimating function $U_i(\alpha,\beta)$ defined as 
\begin{equation}
    U_i(\alpha,\beta) = \sum_{t=1}^T I_{it} e^{- A_{it} S_{it}^T \beta } \epsilon_{it} (\alpha,\beta) M_{it} W_{it} \begin{bmatrix}
        g(H_{it}) \\           
        \{A_{it} - \tilde{p}_t(S_{it})\}S_{it} 
    \end{bmatrix}.
    \label{eq:ee}
\end{equation}
The estimating function is multiplied by $I_{it}$ because the causal excursion effect only concerns the decision points where individuals are available to be randomized.
The factor $e^{- A_{it} S_{it}^T\beta}$ is analogous to the blipping-down term in multiplicative structural nested mean models \citep{robins1994correcting}.
The error term $\epsilon_{it}(\alpha,\beta)$ is defined as $\epsilon_{it}(\alpha,\beta) = Y_{it, \Delta} - e^{g(H_{it})^T \alpha + A_{it} S_{it}^T \beta}$, where $Y_{it,\Delta}$ is the proximal outcome $Y_{t,\Delta}$ for the $i$-th individual.
The $M_{it}$ is defined as 
$M_{it} = \left\{\frac{\tilde{p}_t(S_{it})}{p_t(H_{it})}\right\}^{A_{it}} \left\{\frac{1 - \tilde{p}_t(S_{it})}{1 - p_t(H_{it})}\right\}^{1 - A_{it}}$,
% $M_{it} = \{\tilde{p}_t(S_{it}) / p_t(H_{it})\}^{A_{it}} [\{1 - \tilde{p}_t(S_{it})\} / \{1 - p_t(H_{it})\}]^{1 - A_{it}}$.
% \begin{equation}
%     M_{it} = \bigg[\frac{\tilde{p}_t(S_{it})}{p_t(H_{it})}\bigg]^{\mathbbm{1}(A_{it}=1)} \bigg[ \frac{1 - \tilde{p}_t(S_{it})}{1- p_t(H_{it})} \bigg]^{\mathbbm{1}(A_{it}=0)}. \nonumber
% \end{equation}
which is similar to the stabilized IPW \citep{robins2000marginal} and is employed by other estimators for causal excursion effects \citep{boruvka2018assessing,qian2021estimating}. Here $\tilde{p}_t(S_{it}) \in (0,1)$ is similar to the numerator of a stabilized weight and can be chosen to improve efficiency by making $\tilde{p}_t(S_{it})$ as close to $p_t(H_{it})$ as possible. For example, if $p_t(H_{it})$ is a constant or depends only on $S_{it}$, one can simply set $\tilde{p}_t(S_{it}) = p_t(H_{it})$ which makes $M_{it} \equiv 1$; if $p_t(H_{it})$ depends on variables in $H_{it}\setminus S_{it}$, one should not set $\tilde{p}_t(S_{it}) = p_t(H_{it})$ but can set $\tilde{p}_t(S_{it})$ to be the prediction from a logistic regression with $A_{it}$ being the response and $S_{it}$ being the predictor.
The $W_{it}$ is defined as
$W_{it} = \prod_{j=t+1}^{t+\Delta-1} \left\{\frac{\mathbbm{1}(A_{ij}=0)}{1-p_{j}(H_{ij})}\right\}^{\mathbbm{1}(\max_{1\leq s\leq j-t}R_{i,t+s}=0)}$
% $W_{it} = \prod_{j=t+1}^{t+\Delta-1}\left[\mathbbm{1}(A_{ij}=0) / \{1-p_{j}(H_{ij})\}\right]^{\mathbbm{1}(\max_{1\leq s\leq j-t}R_{i,t+s}=0)}$
% \begin{equation}
%     W_{it} = \prod_{j=t+1}^{t+\Delta-1}\left[\frac{\mathbbm{1}(A_{ij}=0)}{1-p_{j}\left(H_{ij}\right)}\right]^{\mathbbm{1}(\max_{1\leq s\leq j-t}R_{i,t+s}=0)}, \nonumber
% \end{equation}
which is the per-decision inverse probability weight that we will elaborate on later.

We established the consistency and asymptotic normality of $\hat\beta$ in Theorem \ref{thm:CAN}, the proof of which is in Supplementary Material \ref{appen:sec:theorem1-proof}.

\begin{thm}
  \label{thm:CAN}
  Assume the causal effect model \eqref{eq:linear-model} and Assumptions 1-4 in Section \ref{subsec:identification} hold. Let $\beta^*$ denote the true value of $\beta$ corresponding to the data generating distribution $P_0$. Define $\partial_\beta U_i(\alpha,\beta) = \partial U_i(\alpha,\beta) / \partial (\alpha^T, \beta^T)$. Under regularity conditions, $\hat\beta$ is consistent for $\beta^*$ and $\sqrt{n}(\hat\beta-\beta^*)$ is asymptotically normal with mean zero and variance-covariance matrix $\Sigma$, and $\Sigma$ can be consistently estimated by the lower block diagonal $(p \times p)$ entry of the matrix 
  \begin{align*}
      \bigg\{\frac{1}{n}\sum_{i=1}^n \partial_\beta U_i(\hat\alpha,\hat\beta) \bigg\}^{-1} \bigg\{\frac{1}{n}\sum_{i=1}^n U_i(\hat\alpha, \hat\beta) U_i(\hat\alpha, \hat\beta)^T\bigg\} \bigg\{\big[\frac{1}{n}\sum_{i=1}^n \partial_\beta U_i(\hat\alpha, \hat\beta)\big]^{-1}\bigg\}^T.
  \end{align*}
\end{thm}

For the estimand \readdBiostat{$\beta^*$}, \citet{qian2021estimating} proposed an estimator called the estimator for marginal excursion effect (EMEE). The EMEE is the same as our estimator except that $W_{it}$ is replaced by $W_{it}' = \prod_{j=t+1}^{t+\Delta-1}\left\{\frac{\mathbbm{1}(A_{ij}=0)}{1-p_{j}\left(H_{ij}\right)}\right\}$. The $W_{it}'$ is a product of $(\Delta-1)$ inverse probability weights (IPW) to account for the fact that the causal excursion effect \eqref{eq:causal-effect-def} concerns excursions under treatment trajectory $A_{i,t+1}=\cdots=A_{i,t+\Delta-1} = 0$. Let us examine how the EMEE weighs the $(i,t)$-decision point by $W_{it}'$ and later contrast with our estimator. Case 1: suppose for some $t+1 \leq j \leq t+\Delta-1$, $A_{ij} = 1$ (Table \ref{tab:weight-comparison}, Examples 1a and 1b). Because this is different from the treatment trajectory for the excursion, the $(i,t)$-decision point is assigned weight $W_{it}' = 0$ by EMEE. Case 2: suppose $A_{i,t+1}=\cdots=A_{i,t+\Delta-1} = 0$ (Table \ref{tab:weight-comparison}, Examples 2a and 2b). In this case, $W_{it}' > 1$ and the $(i,t)$-decision point is upweighted to account for itself and similar decision points that fall into Case 1 and thus have weight 0.

The key innovation in our proposed estimator is $W_{it}$, which improves efficiency over EMEE. We call $W_{it}$ the \textit{per-decision inverse probability weight}, \readdBiostat{and $\hat\beta$ the pd-EMEE method,} because the inclusion of each of the $(\Delta-1)$ IPW terms is determined by each decision point: for each $j$ from $t+1$ to \readdBiostat{$t+\Delta-1$}, if $\max_{1\leq s\leq j-t}R_{i,t+s}=0$ then $\frac{\mathbbm{1}(A_{ij}=0)}{1-p_{j}\left(H_{ij}\right)}$ is included in $W_{it}$; otherwise it is not included. Therefore, $W_{it}$ includes in its product the IPW terms from $j=t+1$ to right before the first occurrence of some $R_{j^*} = 1$. By not including the IPW terms after this $R_{j^*} = 1$, the per-decision IPW does not account for the treatment assignments after $j^*$. This is appropriate because the maximum property \eqref{eq:maximum-property} implies that the treatment assignments after $j^*$ don't affect the value of $Y_{t,\Delta}$ (which has to equal 1). To see how our estimator improves efficiency over the EMEE, let us now examine how our estimator weighs the $(i,t)$-decision point by $W_{it}$. For Case 1 described above, if some $R_j = 1$ occurs before the first $A_k = 1$ (Table \ref{tab:weight-comparison}, Example 1b), the $(i,t)$-decision point will have a positive weight. For Case 2 described above, if the first $R_{j^*} = 1$ occurs for $j^* < t+\Delta-1$ (Table \ref{tab:weight-comparison}, Example 2b), then $W_{it}$ will include fewer terms than $W_{it}'$ and thus be smaller, because $ W_{it} = \prod_{j=t+1}^{j^*}\frac{1}{1-p_{j}\left(H_{ij}\right)} < \prod_{j=t+1}^{t+\Delta-1}\frac{1}{1-p_{j}\left(H_{ij}\right)} = W_{it}'$.

To summarize, by capitalizing on the maximum property \eqref{eq:maximum-property}, the per-decision IP weight $W_{it}$ makes use of some decision points that would be otherwise discarded because $W_{it}' = 0$ (such as Example 1b) and reduces the weight of some other decision points (such as Example 2b), and this leads to the efficiency improvement. The efficiency improvement of $\hat\beta$ over EMEE is larger when a larger proportion of the decision points belong to the above two types, which happens when the proximal outcome window length $\Delta$ is large, when the randomization probability is large, or when the baseline success probability (i.e., the success probability of the proximal outcome under no treatment) is large. When the proximal outcome window is shorter than the distance between two decision points so that $\Delta = 1$, $\hat\beta$ is equivalent to EMEE. \readdBiostat{We will use ``relative efficiency'' to quantify the efficiency improvement, with relative efficiency between two estimators $\hat\theta_1$ and $\hat\theta_2$ defined as $\text{Var}(\hat\theta_2) / \text{Var}(\hat\theta_1)$.} In Supplementary Material \ref{appen:subsec:re} we provided a theoretical analysis for the relative efficiency between the proposed estimator and the EMEE under simplifying assumptions.

\begin{table}[htbp]
  \caption{Comparison between per-decision IP weight $W_{it}$ in our estimator and standard IP weight $W_{it}'$ in EMEE in a few examples for the case of $\Delta = 3$. Subscript $i$ for all variables are omitted.} \label{tab:weight-comparison}
  \centering
  \small
  \begin{tabular}{cccccccccc}
  %\begin{tabular}{@{}cccccccccc@{}}
  \toprule
    & $R_{t}$ & $A_{t+1}$ & $R_{t+1}$ & $A_{t+2}$ & $R_{t+2}$ & $Y_{t,\Delta}$ & & $W_t$ & $W_t'$ \\
  \midrule
  Example 1a & 0 & 0 & 0 & 1 & 0 & 0 & & 0 & 0 \\
  Example 1b & 0 & 0 & 1 & 1 & 0 & 1 & & $\frac{1}{1-p_{t+1}\left(H_{t+1}\right)}$ & 0 \\
  Example 2a & 0 & 0 & 0 & 0 & 0 & 0 & & $\frac{1}{1-p_{t+1}\left(H_{t+1}\right)} \cdot \frac{1}{1-p_{t+2}\left(H_{t+2}\right)}$ & $\frac{1}{1-p_{t+1}\left(H_{t+1}\right)} \cdot \frac{1}{1-p_{t+2}\left(H_{t+2}\right)}$ \\
  Example 2b & 0 & 0 & 1 & 0 & 0 & 1 & & $\frac{1}{1-p_{t+1}\left(H_{t+1}\right)}$ & $\frac{1}{1-p_{t+1}\left(H_{t+1}\right)} \cdot \frac{1}{1-p_{t+2}\left(H_{t+2}\right)}$ \\
  %\lastline
  \bottomrule
  \end{tabular}

  % \begin{tabular}{lcccccc}
  % \toprule
  % & & $j = t+1$ & $j = t+2$ & & $W_t$ & $W_t'$ \\
  % \midrule
  % & $A_{ij}$ & 0 & 1 \\
  % \multirow{-2}{*}{\centering Case 1(a)} & $R_{ij}$ & 0 & 0 & & \multirow{-2}{*}{\centering 0} & \multirow{-2}{*}{\centering 0} \\
  % \\
  % & $A_{ij}$ & 0 & 1 \\
  % \multirow{-2}{*}{\centering Case 1(b)} & $R_{ij}$ & 1 & 0 & & \multirow{-2}{*}{\centering 0} & \multirow{-2}{*}{\centering 0} \\
  % \\
  % & $A_{ij}$ & 0 & 0 \\
  % \multirow{-2}{*}{\centering Case 2(a)} & $R_{ij}$ & 0 & 0 & & \multirow{-2}{*}{\centering 0} & \multirow{-2}{*}{\centering 0} \\
  % \\
  % & $A_{ij}$ & 0 & 0 \\
  % \multirow{-2}{*}{\centering Case 2(b)} & $R_{ij}$ & 1 & 0 & & \multirow{-2}{*}{\centering 0} & \multirow{-2}{*}{\centering 0} \\
  % \bottomrule
  % \end{tabular}
\end{table}

% When comparing our estimator with EMEE, there is a trade-off between the estimation efficiency and the range of applicable settings. Our estimator capitalizes on the maximum property \eqref{eq:maximum-property} to achieve higher efficiency. Therefore, when applicable our estimator should always be preferred over EMEE. However, for MRTs where the maximum property does not hold (and nor does the generalized maximum property in Supplementary Material  \ref{appen:sec:generalized-max-property}), EMEE should be used because one cannot define the $R$-variables in \eqref{eq:maximum-property} that enables our estimator. This is the case if the value of the proximal outcome cannot possibly be determined until the end of the proximal outcome window. 

Related techniques of reducing IP weights to improve efficiency were developed for other settings. In off-policy reinforcement learning, the per-decision importance sampling technique was proposed to construct IP weights related to ours by leveraging the fact that earlier outcomes cannot depend on later treatments \citep{precup2000eligibility}. In Cox marginal structural model, to determine a subject's contribution to a risk-set at time $t$, a stabilized weight was proposed which involves IP weights not through the end of the follow-up but only up till time $t$ \citep{robins2000marginal}. Those settings are analogous to a special case in our setting ($T=1$), i.e., when focusing on a single proximal outcome that decomposes into a sequence of sub-outcomes. Our contribution is to extend the per-decision IPW technique to estimate causal excursion effects on longitudinal proximal outcomes. In Supplementary Material  \ref{appen:sec:literature}, we provide a detailed explication on how our method connects to and differs from the literature.

Additional remarks about the property of the proposed estimator are as follows.

\begin{rmk} \label{rmk:robustness}
    \textit{(Robustness against a misspecified working model.)}
    Consistency of $\hat\beta$ only requires that $S_t^T \beta$ be a correct model for the causal excursion effect $\est(S_{t})$. It does not require a correct working model $g(H_t)^T\alpha$; i.e., it does not require $E\{Y_{t, \Delta}(\bar{A}_{t-1}, 0, \bar{0}_{\Delta - 1}) \mid H_{t}, I_{t}=1 \} = \exp \{g(H_t)^T \alpha\}$ to hold for some $\alpha$. This robustness property is desired because it is virtually impossible to correctly specify a model for $E\{Y_{t, \Delta}(\bar{A}_{t-1}, 0, \bar{0}_{\Delta - 1}) \mid H_{t}, I_{t}=1 \}$ when the number of decision points for an MRT is moderate to large. Different choices of $g(H_t)$ affect the variance of $\hat\beta$. \readdBiostat{See Section \ref{sec:projection-estimator} and Remark \ref{rmk:efficiency-improvement} on how the projection idea leads to a more efficient choice of $g(H_t)^T\alpha$.}
    % The role of $g(H_t)$ is similar to the control variables in the randomized trials literature, and we will refer to $g(H_t)$ as the control variables.
\end{rmk}

\begin{rmk} \label{rmk:small-ss}
    \textit{(Small-sample correction for inference.)}
    The sandwich estimator for the variance in Theorem \ref{thm:CAN} can be anti-conservative with a small sample size. To correct this, we adopt two small-sample correction techniques: (i) Similar to the correction to the sandwich variance estimator by \citet{mancl2001covariance}, we pre-multiply the vector of each individual’s residual, $\epsilon_t(\hat\alpha, \hat\beta)$, by the inverse of the identity matrix minus the leverage for that individual; (ii) Similar to the correction to the critical value for inference by \citet{liao2016sample}, we use as the critical value the quantile of a $t$-distribution, $t^{-1}_{n-p-q}(1-\eta/2)$, where $\eta$ is the significance level and $p$ and $q$ are the dimensions of $\beta$ and $\alpha$, respectively.
\end{rmk}

\section{\readdBiostat{pd-EMEE2: Further Improving Efficiency Using Projection}}

\label{sec:projection-estimator}

% added in revision for Biostatistics
\readdBiostat{
Following suggestions by a reviewer, we developed a variation of the pd-EMEE $\hat\beta$ using the projection idea from semiparametric efficiency theory. We refer to this variation as pd-EMEE2 and denote it $\tilde\beta$. In particular, let $\xi_i(\beta) = \sum_{t=1}^T I_{it} e^{- A_{it} S_{it}^T \beta} Y_{it, \Delta} M_{it} W_{it} \{A_{it} - \tilde{p}_t(S_{it}) \} S_{it}$ denote the bottom portion of $U_i(\alpha,\beta)$ in \eqref{eq:ee} with $\epsilon_{it}(\alpha,\beta)$ replaced by $Y_{it,\Delta}$. Equation \eqref{eq:identification} implies the moment condition $E\{\xi_i(\beta^*)\} = 0$ (proven in Lemma \ref{lem:consistency} in the Supplementary Material). The efficiency of the estimating function $\xi_i(\beta)$ can be improved by subtracting from $\xi_i(\beta)$ its projection on the score functions for the treatment assignment probability \citep{bickel1993efficient,murphy2001marginal,shi2023meta}. The projection is given by $\sum_{u=1}^T [E\{\xi_i(\beta) \mid H_{iu}, A_{iu}\} - E\{\xi_i(\beta) \mid H_{iu}\}]$ \citep{robins1999testing}. We show in Supplementary Material \ref{appen:sec:projection} that under an additional working assumption (Assumption \ref{appen:est2_asu} in Supplementary Material \ref{appen:sec:projection}), we have
\begin{align}
  & ~~~~ \xi_i(\beta) - \sum_{u=1}^T [E\{\xi_i(\beta) \mid H_{iu}, A_{iu}\} - E\{\xi_i(\beta) \mid H_{iu}\}] \nonumber \\
  & = \sum_{t=1}^T I_{it} e^{- A_{it} S_{it}^T \beta}  M_{it}\{A_{it} - \tilde{p}_t(S_{it})\} S_{it} \bigg[Y_{it, \Delta} W_{it} - E(Y_{it, \Delta}W_{it}|H_{it}, A_{it}) \nonumber \\
  & ~~~~ - \sum_{u = t+1}^{t+\Delta-1} \left\{ E(Y_{it, \Delta}W_{it}|H_{iu}, A_{iu}) - E(Y_{it, \Delta}W_{it}|H_{iu}) \right\} \bigg] \nonumber \\
    &+\sum_{t=1}^T I_{it} \tilde{p}_t(S_{it}) \{1 - \tilde{p}_t(S_{it})\} S_{it} \left\{e^{- S_{it}^T \beta} E(Y_{it, \Delta}W_{it}|H_{it}, A_{it} = 1) - E(Y_{it, \Delta}W_{it}|H_{it}, A_{it} = 0) \right\}. \label{eq:eqs2_star}
\end{align}
This motivates the following estimating function
\begin{align}
  \tilde{U}_i(\beta, \mu) & = \sum_{t=1}^T I_{it} e^{- A_{it} S_{it}^T \beta}  M_{it}\{A_{it} - \tilde{p}_t(S_{it})\} S_{it} \bigg[Y_{it, \Delta} W_{it} - \mu_0(H_{it}, A_{it}) \nonumber \\
  & ~~~~- \sum_{u = t+1}^{t+\Delta-1} \left\{ \mu_{u-t}(H_{iu}, A_{iu}) - p_u(H_{iu})\mu_{u-t}(H_{iu}, 1) - (1-p_u(H_{iu}))\mu_{u-t}(H_{iu}, 0) \right\} \bigg] \nonumber \\
    &+\sum_{t=1}^T I_{it} \tilde{p}_t(S_{it}) \{1 - \tilde{p}_t(S_{it})\} S_{it} \left\{e^{- S_{it}^T \beta} \mu_0(H_{it}, 1) - \mu_0(H_{it}, 0) \right\}, \label{eq:ee-2}
\end{align}
where $\mu = (\mu_0, \mu_1,\ldots, \mu_{\Delta-1})$ is a collection of nuisance parameters with $\mu_s(H_{iu}, A_{iu})$ denoting a working model for $E(Y_{i,u-s, \Delta}W_{i,u-s}|H_{iu}, A_{iu})$. We propose to compute the estimator $\tilde\beta$ based on the estimating function $\tilde{U}$ using a two-step approach depicted in Algorithm \ref{alg:alg1}. The asymptotic property of $\tilde\beta$ is described in Theorem \ref{thm:CAN2}, with proof in Supplementary Material \ref{subsec:appen:proof-ee2}.
}

\begin{algorithm}[htbp]
\doublespacing
\readdBiostat{
  \caption{The pd-EMEE2 estimator $\tilde\beta$} \label{alg:alg1}
  \textbf{Step 1}: For each $s = 0,1,\ldots, \Delta-1$, fit $E(Y_{i,u-s, \Delta}W_{i,u-s}|H_{iu}, A_{iu})$ by pooling across all $1 \leq u \leq T$ and denote the fitted models by $\hat\mu_s(H_{iu},A_{iu})$. Denote by $\hat\mu$ the collection of $\hat\mu_s$ for $0 \leq s \leq \Delta-1$. \\
  \textbf{Step 2}: Obtain $\tilde\beta$ by solving $\frac{1}{n} \sum_{i = 1}^n \tilde{U}_i(\beta; \hat\mu) = 0$ with $\tilde{U}_i(\beta, \mu)$ defined in \eqref{eq:ee-2}.
  }
\end{algorithm}

\readdBiostat{
\begin{thm}
  \label{thm:CAN2}
  Assume the causal effect model \eqref{eq:linear-model} and Assumptions 1-4 in Section \ref{subsec:identification}. Let $\beta^*$ denote the true value of $\beta$ corresponding to the data generating distribution $P_0$. Define $\partial_\beta \tilde{U}_i(\beta; \hat\mu) = \partial \tilde{U}_i(\beta; \hat\mu) / \partial \beta^T$. Under regularity conditions, $\tilde\beta$ is consistent for $\beta^*$ and $\sqrt{n}(\tilde\beta-\beta^*)$ is asymptotically normal with mean zero and variance-covariance matrix $\Sigma$, and $\Sigma$ can be consistently estimated by the lower block diagonal $(p \times p)$ entry of the matrix 
  \begin{align*}
      \bigg\{\frac{1}{n}\sum_{i=1}^n \partial_\beta \tilde{U}_i(\tilde\beta; \hat\mu) \bigg\}^{-1} \bigg\{\frac{1}{n}\sum_{i=1}^n \tilde{U}_i(\tilde\beta; \hat\mu) \tilde{U}_i(\tilde\beta; \hat\mu)^T \bigg\} \bigg\{\big[\frac{1}{n}\sum_{i=1}^n \partial_\beta \tilde{U}_i(\tilde\beta; \hat\mu)\big]^{-1}\bigg\}^{T}.
  \end{align*}
\end{thm}
}

\readdBiostat{
\begin{rmk}
    \label{rmk:efficiency-improvement}
    \textit{(Robustness and efficiency improvement.)}
    Similar to Remark \ref{rmk:robustness}, the consistency of $\tilde\beta$ is robust against misspecified regression models $\hat\mu$. Equation \eqref{eq:eqs2_star} is rooted in the semiparametric efficiency theory, and one can view \eqref{eq:eqs2_star} as identifying a collection of working models that can potentially lead to more efficient $\beta$ estimators than the working model $g(H_{it})^T\alpha$ in \eqref{eq:ee}. In other words, the pd-EMEE2 $\tilde\beta$ can be more efficient than the pd-EMEE $\hat\beta$. Note that Assumption \ref{appen:est2_asu} was made to facilitate a closed-form \eqref{eq:eqs2_star}, and this working assumption is not required by Theorem \ref{thm:CAN2}. Furthermore, the efficiency gain by $\tilde\beta$ does not require this working assumption, as is supported by numerical evidence (Sections \ref{sec:simulation} and \ref{sec:application}).
\end{rmk}
}

\readdBiostat{
\begin{rmk} 
    \textit{(Flexible modeling choice for $\mu$.)}
    The choice of models for the conditional expectations $E(Y_{i,u-s, \Delta}W_{i,u-s}|H_{iu}, A_{iu})$ in Algorithm \ref{alg:alg1} can be very flexible, and one can take advantage of the maximum property to further improve efficiency. For example, one can use Poisson regression or logistic regression with rescaled $Y_{i,u-s, \Delta}W_{i,u-s}$ that are within $[0,1]$ (and scale them back after the logistic regression fit). One can also include functions and spline basis of the decision point index, $u$, in the model. In addition, the $(H_{iu}, A_{iu})$-measurable part of $W_{i,u-s}$, i.e., $\prod_{j=u-s+1}^u \left\{\frac{\mathbbm{1}(A_{ij}=0)}{1-p_{j}(H_{ij})}\right\}^{\mathbbm{1}(\max_{u-s+1\leq v\leq j}R_{i,v}=0)}$, can be taken outside of the conditional expectation $E(Y_{i,u-s, \Delta}W_{i,u-s}|H_{iu}, A_{iu})$, before a model is fitted to the conditional expectation, and doing this has the potential to improve efficiency. Furthermore, due to the maximum property, $Y_{i,u-s, \Delta}$ will be deterministically 1 if one of its sub-outcomes $R$ is included in $H_{iu}$ and equals 1, and taking advantage of this fact can further improve efficiency. Finally, due to the estimating function $\tilde{U}_i(\beta, \mu)$ being globally robust (i.e., it has expectation 0 at the true $\beta^*$ for any $\mu$; see Supplementary Material \ref{subsec:appen:proof-ee2}), the nuisance parameter estimators $\hat\mu$ can converge at an arbitrarily slow rate and the asymptotic variance of $\tilde\beta$ does not need to account for the fact that $\mu$ is estimated \citep{cheng2023efficient}. This means that a wide range of machine learning and nonparametric models can be used to fit $\mu$. 
\end{rmk}
}

\section{Simulation}
\label{sec:simulation}

% \subsection{Overview}
% We studied the finite-sample performance of our proposed estimator and three competitors: the EMEE by \citep{qian2021estimating}, and two generalized estimating equations \citep{liang1986longitudinal}. We will refer to our estimator as ``pd-EMEE'', where ``pd'' stands for per-decision.

\subsection{Generative Model}
\label{subsec:simulation-gm}

We set $I_t \equiv 1$ and $T = 100$. The generative model depends on two parameters, $\Delta$ (length of the proximal outcome window) and $p_a$ (constant randomization probability). $A_t \in \{0,1\}$ was generated from $\text{Bernoulli}(p_a)$. A scalar time-varying covariate, $Z_t \in \{0, 1, 2\}$, was generated independently \readdBiostat{of} all previous variables (i.e., $Z_t \perp \{Z_s,A_s,R_{s+1}: 1 \leq s \leq t-1\}$) with
\begin{align}
    P(Z_t = 0) = \gammaplaceholder^{-1/2\Delta}/C, \quad
    P(Z_t = 1) = 1/C, \quad
    P(Z_t = 2) = \gammaplaceholder^{1/2\Delta}/C, \label{eq:dgm-Zt}
\end{align}
where $C = \gammaplaceholder^{-1/2\Delta} + \gammaplaceholder^{1/2\Delta} + 1$.
The sub-outcome $R_{t+1}$ given $A_t$ and $H_t$ was generated as
\begin{align}
    P(R_{t+1} = 0 \mid A_t = 0, H_t) & = \gammaplaceholder^{(1.5-0.5Z_t)/\Delta}, \label{eq:dgm-R-A0}\\
    P(R_{t+1} = 0 \mid A_t = 1, H_t) & = \frac{1- \{1-\gammaplaceholder^{(1.5-0.5Z_t)/\Delta}\cdot (3/C\cdot \gammaplaceholder^{1/\Delta})^{\Delta-1}\} \cdot e^{0.1+0.2Z_t}}{(3/C \cdot \gammaplaceholder^{1/\Delta})^{\Delta-1}}. \label{eq:dgm-R-A1}
\end{align}
We set $R_{t+1} = 0$ for $ T < t \leq T + \Delta$.
The proximal outcome was generated as $Y_{t,\Delta} = \max (R_{t+1}, \ldots, R_{t+\Delta})$. 

%appen:sec:dgm
In Supplementary Material  \ref{appen:sec:dgm} we showed that for $1 \leq t \leq T - \Delta$, 
\begin{align}
    & E\{Y_{t,\Delta}(\bar{A}_{t-1}, a_t, \bar{0}_{\Delta-1}) \mid H_t(\bar{A}_{t-1}), I_t(\bar{A}_{t-1}) = 1\} = \{1 - \gammaplaceholder^{(\Delta+0.5-0.5Z_{t})/\Delta}\cdot(3/C)^{\Delta-1}\} \cdot e^{a_t(0.1+0.2Z_t)}. \label{eq:dgm-Y}
\end{align}
Therefore, $\est(Z_t) = 0.1 + 0.2 Z_t$ for $1 \leq t \leq T - \Delta$ under the generative model, and this holds for all $\Delta \geq 1$ and $p_a \in (0,1)$. 
%\tq{7/12: Yihan, please check the correctness of my following statement.}\yb{should be correct!}
For $T - \Delta < t < T$, $\est(Z_t)$ is slightly different but its influence on the simulation results is negligible because in the simulations $\Delta \ll T$.

There is substantial variability in $R_{t+1}$ and $Y_{t,\Delta}$. When we vary $\Delta$ from 2 to 10 and vary $Z_t \in \{0,1,2\}$,
% and set $\gammaplaceholder = 0.5$,
\eqref{eq:dgm-R-A0} varies from 0.59 to 0.97, \eqref{eq:dgm-R-A1} varies from 0.38 to 0.80, and \eqref{eq:dgm-Y} varies from 0.41 to 0.58 when setting $a_t = 0$ and from 0.57 to 0.80 when setting $a_t = 1$. 

We constructed this generative model based on three considerations.
First, in order to fairly assess how the efficiency improvement is affected by $\Delta$ and $p_a$, we needed $\est(Z_t)$ to not depend on $\Delta$ or $p_a$. 
%As discussed in Section \ref{sec:method}, a larger $\Delta$ or a smaller $p_a$ would lead to a larger efficiency improvement. To assess this relationship in the simulations, we need a parsimonious $\est(Z_t)$ that does not depend on $\Delta$ or $p_a$. This way, when we vary $\Delta$ or $p_a$, the resulting change in the efficiency improvement is not confounded by the potential change in $\est(Z_t)$. 
Thus, the distribution of the sub-outcome $R_t$ would depend on $\Delta$ because there are $\Delta$ sub-outcomes in $Y_{t,\Delta}$.
Second, we wanted to construct a complicated $E\{Y_{t,\Delta}(\bar{A}_{t-1}, 0, \bar{0}_{\Delta-1}) \mid H_t, I_t = 1\}$ to demonstrate the robustness of the estimator against misspecified working models.
Third, we wanted $Z_t$ to depend on $\Delta$ as little as possible. In a real application, $Z_t$ (a covariate) would not depend on $\Delta$ (a choice in the analysis). However, we found it difficult to achieve while keeping the form of $\est(Z_t)$ simple. The current generating distribution of $Z_t$ \eqref{eq:dgm-Zt} is the least dependent on $\Delta$ that we came up with, and the three probabilities in \eqref{eq:dgm-Zt} are all very close to 1/3 for $\Delta \geq 2$.

\subsection{Simulation Result} 
\label{subsec:simulation-result}

% We set $\gammaplaceholder = 0.5$ throughout until the last paragraph, where we vary the success probability for the $R$-variables under no treatment through varying $\gammaplaceholder$.

\readdBiostat{We studied the finite-sample performance of the proposed pd-EMEE, pd-EMEE2, and three competitors: the EMEE by \citet{qian2021estimating}, the log-linear generalized estimating equations with independent correlation structure (GEE.ind), and the log-linear generalized estimating equations with exchangeable correlation structure (GEE.exch) \citep{liang1986longitudinal}.}

We considered two estimands with different $S_t$: the first estimand is $\est(S_t) = \beta_0$ with $S_t = 1$, and the second estimand is $\est(S_t) = \beta_1 + \beta_2 Z_t$ with $S_t = (1, Z_t)$. They represent common primary and secondary analyses of MRTs: $\beta_0$ is the fully marginal causal excursion effect, and $\beta_1 + \beta_2 Z_t$ is the effect moderated by $Z_t$. For pd-EMEE and EMEE, we used the correct model of $\est$ for both estimands and the incorrect working model $\exp\{g(H_t)^T \alpha\} = \exp(\alpha_0 + \alpha_1 Z_t)$ for $E(Y_{t,\Delta}(\bar{A}_{t-1}, 0, \bar{0}_{\Delta-1}) \mid H_t, I_t = 1)$. \readdBiostat{For pd-EMEE2, we used the correct model of $\est$ and included only $(Z_{iu},A_{iu})$ in a linear regression fit for $\hat\mu_s(H_{iu},A_{iu})$ and thus the latter is misspecified.} For GEE.ind and GEE.exch, we use the marginal mean model $\exp(\alpha_0 + \alpha_1 Z_t + \beta_0 A_t)$ for the first estimand and $\exp\{\alpha_0 + \alpha_1 Z_t + A_t (\beta_1 + \beta_2 Z_t)\}$ for the second estimand. The marginal mean models used by the GEEs are misspecified because the part involving the $\alpha$'s is incorrect. See Supplementary Material \ref{appen:sec:gee} for the exact form of the log-linear GEEs used.

We first conducted simulations with two different $\Delta$: $\Delta = 3$ and $\Delta = 10$. We used a constant randomization probability $p_a = 0.2$. Results for $\Delta = 3$ are shown in top halves of Tables \ref{tab:delta3and10-1} and \ref{tab:delta3and10-2}, where the true parameter values are $\beta_0^* = 0.283$ (obtained numerically by plugging \eqref{eq:dgm-Y} into \eqref{eq:causal-effect-def}), $\beta_1^* = 0.1$, and $\beta_2^* = 0.2$. pd-EMEE, \readdBiostat{pd-EMEE2,} and EMEE consistently estimate all parameters of interest with close to nominal confidence interval coverage. The confidence interval coverage is improved by the small sample correction (Remark \ref{rmk:small-ss}). pd-EMEE is more efficient than EMEE: for example, when $n=100$, the relative efficiency between pd-EMEE and EMEE for estimating $\beta_0$, $\beta_1$, and $\beta_2$ is $0.026^2/0.025^2 = 1.08$, $0.037^2/0.035^2 = 1.12$, and $0.029^2 / 0.027^2 = 1.15$, respectively. These values roughly translate to 7\% ($ = 1 - 1/1.08$), 11\%, and 13\% savings in sample size if one were to power an MRT using pd-EMEE instead of EMEE. \readdBiostat{Comparing pd-EMEE and pd-EMEE2, pd-EMEE2 is as or more efficient across $n = 30, 50 \text{ or } 100$, with the highest relative efficiency (against pd-EMEE) being $1.08$.} Both GEE.ind and GEE.exch are inconsistent because the marginal mean model is misspecified. The performance of $\hat\alpha_0$ and $\hat\alpha_1$ is not presented because $\alpha_0$ and $\alpha_1$ are nuisance parameters and thus not of interest.

Results for $\Delta = 10$ are shown in bottom halves of Tables \ref{tab:delta3and10-1} and \ref{tab:delta3and10-2}, where the true parameter values are $\beta_0^* = 0.304$, $\beta_1^* = 0.1$, and $\beta_2^* = 0.2$. The conclusions regarding consistency and confidence interval coverage are the same as \readdBiostat{those} when $\Delta = 3$. Here pd-EMEE is substantially more efficient than EMEE: for example, when $n=100$, the relative efficiency between pd-EMEE and EMEE for estimating $\beta_0$, $\beta_1$, and $\beta_2$ is $0.065^2/0.054^2 = 1.45$, $0.092^2/0.078^2 = 1.39$, and $0.064^2 / 0.054^2 = 1.40$, respectively, which roughly translates to 31\%, 28\%, and 29\% savings in sample size. \readdBiostat{When $\Delta = 10$, pd-EMEE2 further substantially improves the efficiency over pd-EMEE. For example, when $n=100$, the relative efficiency between pd-EMEE2 and pd-EMEE for estimating $\beta_0$, $\beta_1$, and $\beta_2$ is $0.054^2/0.053^2 = 1.04$, $0.078^2/0.073^2 = 1.14$, and $0.054^2 / 0.050^2 = 1.17$, respectively, which roughly translates to another 3\%, 12\%, and 15\% savings in sample size on top of the savings of pd-EMEE.} The efficiency improvement when $\Delta = 10$ is greater compared to when $\Delta=3$.

\begin{table}[htbp]
\caption{Simulation result: Performance of pd-EMEE2, pd-EMEE, EMEE, GEE.ind and GEE.exch when estimating $\beta_0$ ($\Delta$ = 3 and 10). Based on 1,000 simulation replications.}
\centering
\begin{threeparttable}
\begin{tabular}[t]{ccccccc}
%\begin{tabular}{@{}ccccccc@{}}
\toprule
Estimator & $n$ & Bias & SD & RMSE & CP.unadj & CP.adj \\
\midrule
 \multicolumn{7}{c}{\textit{Setting:} $\Delta = 3$} \vspace{0.5em} \\
& 30 & 0.006 & 0.044 & 0.044 & 0.94 & 0.96\\

& 50 & 0.003 & 0.036 & 0.036 & 0.94 & 0.95\\

\multirow{-3}{*}{\centering\arraybackslash pd-EMEE2} & 100 & 0.005 & 0.025 & 0.026 & 0.94 & 0.94\\
\cmidrule{1-7}

 & 30 & 0.006 & 0.045 & 0.045 & 0.94 & 0.96\\

 & 50 & 0.004 & 0.037 & 0.037 & 0.93 & 0.94\\

\multirow{-3}{*}{\centering\arraybackslash pd-EMEE} & 100 & 0.005 & 0.025 & 0.026 & 0.94 & 0.94\\
\cmidrule{1-7}
 & 30 & 0.005 & 0.047 & 0.048 & 0.95 & 0.96\\

 & 50 & 0.004 & 0.039 & 0.039 & 0.93 & 0.95\\

\multirow{-3}{*}{\centering\arraybackslash EMEE} & 100 & 0.005 & 0.026 & 0.027 & 0.95 & 0.96\\
\cmidrule{1-7}
 & 30 & -0.061 & 0.032 & 0.069 & 0.51 & 0.54\\

 & 50 & -0.062 & 0.026 & 0.068 & 0.29 & 0.31\\

\multirow{-3}{*}{\centering\arraybackslash GEE.ind} & 100 & -0.062 & 0.018 & 0.064 & 0.05 & 0.06\\
\cmidrule{1-7}
 & 30 & -0.063 & 0.032 & 0.071 & 0.48 & 0.52\\

 & 50 & -0.065 & 0.026 & 0.069 & 0.26 & 0.28\\

\multirow{-3}{*}{\centering\arraybackslash GEE.exch} & 100 & -0.064 & 0.018 & 0.066 & 0.04 & 0.04\\

\\
\multicolumn{7}{c}{\textit{Setting:} $\Delta = 10$} \vspace{0.5em}\\

 & 30 & 0.026 & 0.095 & 0.099 & 0.93 & 0.96\\

 & 50 & 0.021 & 0.078 & 0.080 & 0.93 & 0.94\\

\multirow{-3}{*}{\centering\arraybackslash pd-EMEE2} & 100 & 0.021 & 0.053 & 0.057 & 0.93 & 0.94\\
\cmidrule{1-7}

 & 30 & 0.026 & 0.103 & 0.106 & 0.94 & 0.96\\

 & 50 & 0.022 & 0.084 & 0.087 & 0.94 & 0.95\\

\multirow{-3}{*}{\centering\arraybackslash pd-EMEE} & 100 & 0.022 & 0.054 & 0.058 & 0.95 & 0.96\\
\cmidrule{1-7}
 & 30 & 0.032 & 0.127 & 0.130 & 0.94 & 0.96\\

 & 50 & 0.026 & 0.099 & 0.103 & 0.94 & 0.96\\

\multirow{-3}{*}{\centering\arraybackslash EMEE} & 100 & 0.023 & 0.065 & 0.069 & 0.95 & 0.96\\
\cmidrule{1-7}
 & 30 & -0.168 & 0.025 & 0.169 & 0.00 & 0.00\\

 & 50 & -0.169 & 0.019 & 0.170 & 0.00 & 0.00\\

\multirow{-3}{*}{\centering\arraybackslash GEE.ind} & 100 & -0.169 & 0.014 & 0.169 & 0.00 & 0.00\\
\cmidrule{1-7}
 & 30 & -0.175 & 0.024 & 0.177 & 0.00 & 0.00\\

 & 50 & -0.176 & 0.019 & 0.177 & 0.00 & 0.00\\

\multirow{-3}{*}{\centering\arraybackslash GEE.exch} & 100 & -0.176 & 0.013 & 0.177 & 0.00 & 0.00 \\
\bottomrule
\end{tabular}
\begin{tablenotes}
\item pd-EMEE, the estimator $\hat\beta$ proposed in Section \ref{sec:method}; \readdBiostat{pd-EMEE2, the estimator $\tilde\beta$ proposed in Section \ref{sec:projection-estimator};} EMEE, the estimator in \citet{qian2021estimating}; GEE.ind, generalized estimating equations with independent working correlation structure; GEE.exch, generalized estimating equations with exchangeable working correlation structure; SD, standard deviation; RMSE, root mean squared error; CP, $95\%$ confidence interval coverage probability before (unadj) and after (adj) small-sample correction.
\end{tablenotes}
\end{threeparttable}
\label{tab:delta3and10-1}
\end{table}

\begin{sidewaystable}[htbp]
\begin{center}
\caption{Simulation result: Performance of pd-EMEE2, pd-EMEE, EMEE, GEE.ind and GEE.exch when estimating $(\beta_1,\beta_2)$ ($\Delta$ = 3 and 10). Based on 1,000 simulation replications.}
\small
%\resizebox{\columnwidth}{!}{%
\begin{threeparttable}
\begin{tabular}[t]{cccccccccccc}
\hline
\multicolumn{2}{c}{ } & \multicolumn{5}{c}{$\beta_1$} & \multicolumn{5}{c}{$\beta_2$}\\
\toprule
Estimator & Sample Size & Bias & SD & RMSE & CP(unadj) & CP(adj) & Bias & SD & RMSE & CP(unadj) & CP(adj)\\
\midrule
\multicolumn{12}{c}{\textit{Setting:} $\Delta = 3$} \vspace{0.5em} \\
 & 30 & 0.004 & 0.064 & 0.064 & 0.95 & 0.97 & 0.002 & 0.052 & 0.053 & 0.94 & 0.95\\

 & 50 & -0.001 & 0.051 & 0.051 & 0.94 & 0.96 & 0.004 & 0.041 & 0.041 & 0.94 & 0.95\\

\multirow{-3}{*}{\centering\arraybackslash pd-EMEE2} & 100 & 0.002 & 0.035 & 0.035 & 0.95 & 0.96 & 0.003 & 0.027 & 0.027 & 0.96 & 0.96\\
\cmidrule{1-12}

 & 30 & 0.004 & 0.066 & 0.066 & 0.94 & 0.96 & 0.002 & 0.053 & 0.053 & 0.94 & 0.96\\

 & 50 & 0.000 & 0.053 & 0.053 & 0.94 & 0.95 & 0.004 & 0.041 & 0.042 & 0.94 & 0.94\\

\multirow{-3}{*}{\centering\arraybackslash pd-EMEE} & 100 & 0.002 & 0.035 & 0.035 & 0.96 & 0.96 & 0.003 & 0.027 & 0.027 & 0.96 & 0.96\\
\cmidrule{1-12}
 & 30 & 0.003 & 0.070 & 0.070 & 0.94 & 0.96 & 0.002 & 0.055 & 0.055 & 0.94 & 0.96\\

 & 50 & 0.000 & 0.057 & 0.057 & 0.94 & 0.95 & 0.003 & 0.044 & 0.044 & 0.93 & 0.95\\

\multirow{-3}{*}{\centering\arraybackslash EMEE} & 100 & 0.002 & 0.037 & 0.037 & 0.96 & 0.97 & 0.003 & 0.029 & 0.029 & 0.95 & 0.96\\
\cmidrule{1-12}
 & 30 & -0.018 & 0.050 & 0.054 & 0.92 & 0.93 & -0.038 & 0.040 & 0.056 & 0.81 & 0.84\\

 & 50 & -0.020 & 0.041 & 0.045 & 0.90 & 0.91 & -0.039 & 0.032 & 0.050 & 0.74 & 0.75\\

\multirow{-3}{*}{\centering\arraybackslash GEE.ind} & 100 & -0.018 & 0.026 & 0.032 & 0.91 & 0.91 & -0.040 & 0.020 & 0.044 & 0.53 & 0.54\\
\cmidrule{1-12}
 & 30 & -0.019 & 0.050 & 0.054 & 0.92 & 0.93 & -0.040 & 0.040 & 0.057 & 0.79 & 0.82\\

 & 50 & -0.020 & 0.041 & 0.045 & 0.90 & 0.91 & -0.040 & 0.031 & 0.051 & 0.72 & 0.74\\

\multirow{-3}{*}{\centering\arraybackslash GEE.exch} & 100 & -0.019 & 0.026 & 0.032 & 0.90 & 0.91 & -0.041 & 0.020 & 0.046 & 0.50 & 0.51\\

\\
\multicolumn{12}{c}{\textit{Setting:} $\Delta = 10$} \vspace{0.5em} \\
& 30 & 0.011 & 0.134 & 0.134 & 0.92 & 0.95 & 0.015 & 0.096 & 0.097 & 0.91 & 0.94\\

 & 50 & 0.003 & 0.108 & 0.108 & 0.92 & 0.93 & 0.016 & 0.072 & 0.074 & 0.93 & 0.95\\

\multirow{-3}{*}{\centering\arraybackslash pd-EMEE2} & 100 & 0.005 & 0.073 & 0.073 & 0.95 & 0.95 & 0.015 & 0.050 & 0.052 & 0.93 & 0.94\\
\cmidrule{1-12}
 & 30 & 0.013 & 0.150 & 0.151 & 0.93 & 0.96 & 0.014 & 0.105 & 0.106 & 0.93 & 0.95\\

 & 50 & 0.003 & 0.120 & 0.120 & 0.92 & 0.93 & 0.018 & 0.082 & 0.084 & 0.93 & 0.95\\

\multirow{-3}{*}{\centering\arraybackslash pd-EMEE} & 100 & 0.007 & 0.078 & 0.078 & 0.95 & 0.96 & 0.014 & 0.054 & 0.056 & 0.94 & 0.95\\
\cmidrule{1-12}
 & 30 & 0.008 & 0.188 & 0.188 & 0.93 & 0.96 & 0.021 & 0.130 & 0.131 & 0.93 & 0.96\\

 & 50 & 0.001 & 0.141 & 0.141 & 0.95 & 0.97 & 0.022 & 0.095 & 0.097 & 0.95 & 0.96\\

\multirow{-3}{*}{\centering\arraybackslash EMEE} & 100 & 0.006 & 0.092 & 0.092 & 0.97 & 0.97 & 0.015 & 0.064 & 0.066 & 0.96 & 0.96\\
\cmidrule{1-12}
 & 30 & -0.058 & 0.040 & 0.070 & 0.68 & 0.71 & -0.109 & 0.029 & 0.112 & 0.05 & 0.06\\

 & 50 & -0.061 & 0.030 & 0.068 & 0.48 & 0.50 & -0.107 & 0.022 & 0.110 & 0.00 & 0.00\\

\multirow{-3}{*}{\centering\arraybackslash GEE.ind} & 100 & -0.060 & 0.021 & 0.063 & 0.21 & 0.22 & -0.108 & 0.015 & 0.109 & 0.00 & 0.00\\
\cmidrule{1-12}
 & 30 & -0.061 & 0.039 & 0.073 & 0.63 & 0.68 & -0.112 & 0.029 & 0.116 & 0.04 & 0.04\\

 & 50 & -0.063 & 0.030 & 0.070 & 0.42 & 0.44 & -0.111 & 0.022 & 0.113 & 0.00 & 0.00\\

\multirow{-3}{*}{\centering\arraybackslash GEE.exch} & 100 & -0.063 & 0.021 & 0.066 & 0.16 & 0.16 & -0.112 & 0.015 & 0.113 & 0.00 & 0.00\\
\bottomrule
\end{tabular}

\begin{tablenotes}
\item pd-EMEE, the estimator $\hat\beta$ proposed in Section \ref{sec:method}; \readdBiostat{pd-EMEE2, the estimator $\tilde\beta$ proposed in Section \ref{sec:projection-estimator};} EMEE, the estimator in \citet{qian2021estimating}; GEE.ind, generalized estimating equations with independent working correlation structure; GEE.exch, generalized estimating equations with exchangeable working correlation structure; SD, standard deviation; RMSE, root mean squared error; CP, $95\%$ confidence interval coverage probability before (unadj) and after (adj) small-sample correction.
\end{tablenotes}
\end{threeparttable}
%}
\label{tab:delta3and10-2}
\end{center}
\end{sidewaystable}

\readdBiostat{
We then investigated how the relative efficiency of pd-EMEE or pd-EMEE2 against EMEE depends on $\Delta$ (proximal outcome window length) and $p_a$ (randomization probability) by varying $\Delta$ and $p_a$ separately. When $\Delta$ varies from 1 to 10 with $p_a = 0.2$, the relative efficiency increases from 1 to 1.5 for pd-EMEE and 1 to 1.75 for pd-EMEE2 (left panel of Figure \ref{fig:simulation_efficiency}). When $p_a$ varies from 0.1 to 0.8 with $\Delta = 3$, the relative efficiency increases from 1.05 to 1.75 for pd-EMEE and from 1.06 to 2.00 for pd-EMEE2  (right panel of Figure \ref{fig:simulation_efficiency}). Therefore, a larger $\Delta$ or a larger $p_a$ makes the proposed pd-EMEE and pd-EMEE2 even more efficient than the EMEE method. This phenomenon was expected and explained in Section \ref{sec:method}.
}

\begin{figure}[htbp]
\begin{center}
    \caption{Simulation result: The relative efficiency (RE) of pd-EMEE (red) or pd-EMEE2 (blue) over EMEE and how it varies with $\Delta$ (proximal outcome window length) and $p_a$ (randomization probability). The variances of EMEE, pd-EMEE, and pd-EMEE2 are illustrated by various line types in black.}
  \centering
  \includegraphics[width = 1\textwidth]{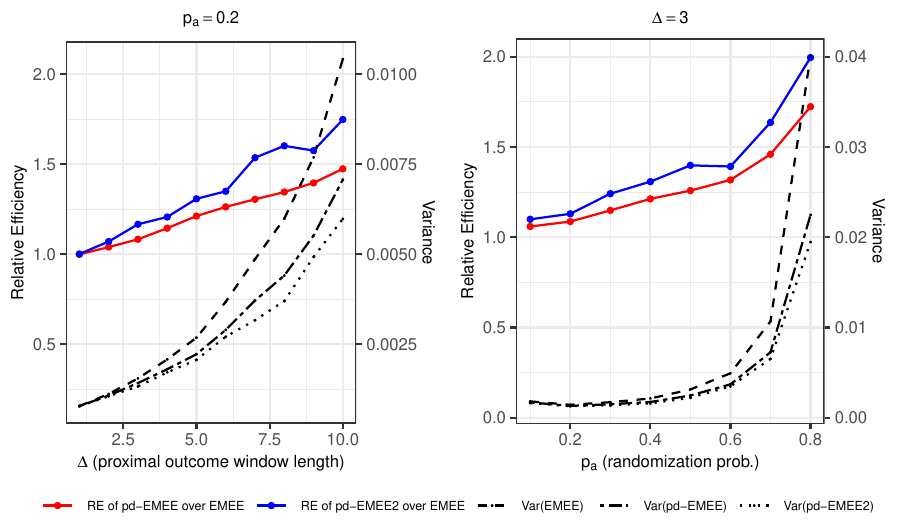}
  %[width=0.6\linewidth,height = 0.35\paperheight]  
  \label{fig:simulation_efficiency}
\end{center}
\end{figure}

\section{Application}
\label{sec:application}

\subsection{Drink Less MRT}
\label{subsec:drinkless-application}

\readdBiostat{We applied the proposed pd-EMEE and pd-EMEE2 methods to analyze the Drink Less MRT data.} Recall that there are $n=349$ participants in this MRT, each in the study for $T=30$ days. Each participant was randomized once every day at 8 pm with 0.6 probability of receiving a push notification encouraging using the Drink Less app ($A_t = 1$) and with 0.4 probability of receiving nothing ($A_t = 0$). Individuals were available to be randomized at all decision points ($I_t \equiv 1$). The proximal outcome, $Y_{t,\Delta}$, is whether the participant logged in to the app in the subsequent $\Delta$ days. The binary sub-outcome, $R_{t+1}$, is whether the participant logged in to the app between 8 pm of day $t$ and 7:59 pm of day $t+1$, and $Y_{t,\Delta} = \max(R_{t+1}, \ldots, R_{t+\Delta})$. \readdBiostat{We conducted separate analyses for $\Delta = 1,2,\ldots,5.$}

% The covariates $X_t$ include age, gender, employment type, the number of days since completing the on-boarding process (1 to 30), and time-varying variables such as ``whether the user uses the app before 8 PM on day $t$" and ``whether the user uses the app after 9 PM before day $t$". 

\readdBiostat{We considered four estimands: the fully marginal effect (moderator $S_t = \emptyset$), the effect moderated by the number of days in the study ($S_t = \text{decision point index } t$), the effect moderated by past treatment ($S_t = \text{receiving treatment the day before}$), and the effect moderated by past app use ($S_t = \text{using the app between 8 pm and 9 pm the day before}$). In all analyses, the $g(H_t)$ in the working model for $E\{Y_{t, \Delta}(\bar{A}_{t-1}, 0, \bar{0}_{\Delta - 1}) \mid H_{t}, I_{t}=1 \}$ includes the decision point index, gender, age, employment type, the AUDIT score (a measure of alcohol misuse), and two time-varying variables about app usage: whether the participant used the app before 8 pm that day, and whether they used the app after 9 pm the previous day. For pd-EMEE2, the same control variables corresponding to decision point $u$ is used to fit $\mu_s(H_{iu}, A_{iu})$.}

% number of days in the study. $\text{Day}_t$ takes values $0,1, \ldots, 29$. The fully marginal excursion effect, $\beta_0$, is defined by setting $S_t = 1$ in \eqref{eq:linear-model}:
% \begin{align*}
%     \beta_0 = \log \frac{E\left\{ Y_{t, \Delta}(\bar{A}_{t-1}, 1, \bar{0}_{\Delta - 1}) \right\}}{E\left\{ Y_{t, \Delta}(\bar{A}_{t-1}, 0, \bar{0}_{\Delta - 1}) \right\}}.
% \end{align*}
% The effect moderated by $\text{Day}_t$, parameterized by $(\beta_1, \beta_2)$, is defined by setting $S_t = (1, \text{Day}_t)$ in \eqref{eq:linear-model}:
% \begin{align*}
%     \beta_1 + \beta_2 \text{Day}_t = \log \frac{E\left\{ Y_{t, \Delta}(\bar{A}_{t-1}, 1, \bar{0}_{\Delta - 1}) \mid \text{Day}_t \right\}}{E\left\{ Y_{t, \Delta}(\bar{A}_{t-1}, 0, \bar{0}_{\Delta - 1}) \mid \text{Day}_t \right\}}.
% \end{align*}
% $\beta_1$ is the causal excursion effect on the first day, and $\beta_2$ is the change in the effect with each additional day.

\begin{figure}[htbp]
    \begin{center}
    \caption{Estimated marginal excursion effect (Moderator: None) and effect moderation by decision point index, past treatment, or past app usage of Drink Less push notification on whether the user opens the app in different subsequent times windows from $\Delta = 1$ day to $\Delta = 5$ days. For each moderator and estimand, the relative efficiency of pd-EMEE or pd-EMEE2 against EMEE with various $\Delta$ is shown in the corresponding dot figure below.}
    \includegraphics[width = 1\textwidth]{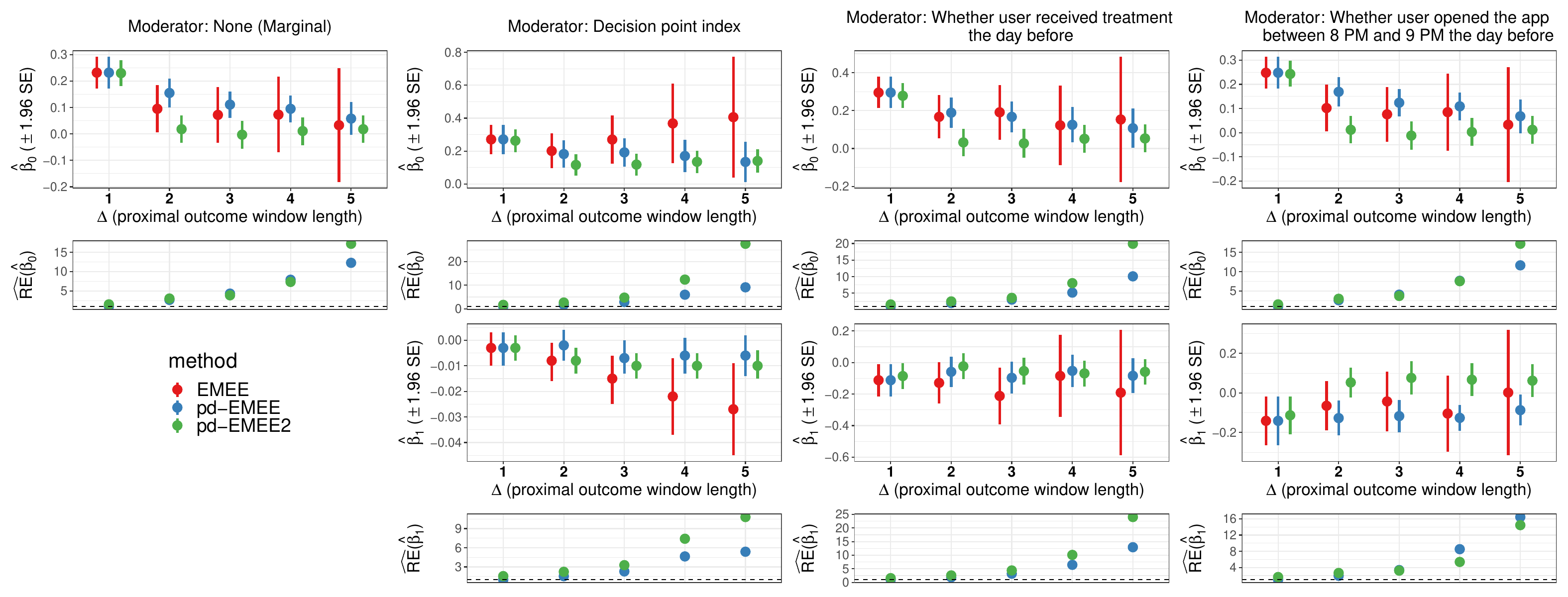}
    \label{fig:Drinkless}
    \end{center}
\end{figure}

\readdBiostat{Figure \ref{fig:Drinkless} shows the analysis results, where for each pair of plots the top one shows the point estimate and 95\% confidence interval (CI) for the three methods (EMEE, pd-EMEE, pd-EMEE2) and the bottom one shows the estimated relative efficiency ($\widehat{\text{RE}}$) of pd-EMEE or pd-EMEE2 against EMEE. $\widehat{\text{RE}}$ between two estimators is computed as the ratio between the estimated standard error of the estimators. We see that for each estimand and each $\Delta$, EMEE, pd-EMEE, and pd-EMEE2 yield estimates in the same direction with CIs mostly overlapping. Both pd-EMEE and pd-EMEE2 yield substantially narrower CIs than EMEE for $\Delta \geq 2$ and their $\widehat{\text{RE}}$ against EMEE increases with larger $\Delta$ as expected. For example, with $\Delta = 3$, the CI for EMEE is almost twice as wide as the CIs for pd-EMEE and pd-EMEE2. pd-EMEE2 is slightly more efficient than pd-EMEE in most cases, especially for the larger $\Delta$ ($\Delta = 4 \text{ or } 5$).}

\subsection{HeartSteps I MRT}
\label{subsec:heartsteps-application}

\readdBiostat{We applied the proposed pd-EMEE and pd-EMEE2 methods to analyze the HeartSteps I MRT data. Recall that each participant was in the study for 42 days and was randomized five times a day, each time with 0.6 probability of receiving an activity suggestion ($A_t = 1$) and with 0.4 probability of receiving nothing ($A_t = 0$). The data set included 37 participants with a total of 7629 decision points, and seven participants had less than 210 decision points of data (with reasons detailed in \citep{klasnja2018efficacy}). Individuals were available to be randomized ($I_t = 1$) for about 80\% of the decision points, and the reason for unavailability ($I_t = 0$) includes the person currently driving, having just had an activity bout, or having connection issues. In our analysis, $\Delta = 5$ corresponds to a 24-hour proximal outcome window. We considered five different proximal outcomes corresponding to five different step count thresholds. Given a total step count threshold $C$, the proximal outcome, $Y_{t,\Delta}$, is whether the participant takes at least $C$ total steps in the 24-hour window subsequent to decision point $t$ (we omit the notational dependence of $Y_{t,\Delta}$ on $C$). The binary sub-outcome $R_{t,t+s}$ for $1 \leq s \leq \Delta$, which is the double-subscript version explicated in the Supplementary Material \ref{appen:sec:generalized-max-property}, is whether the individual takes at least $C$ total steps between decision points $t$ and $t+s$, and $Y_{t,\Delta} = \max(R_{t, t+1}, \ldots, R_{t, t+\Delta})$. We conducted separate analyses for $C = 6000, 7000, \ldots, 10000$.}

\readdBiostat{We considered four estimands: the fully marginal effect (moderator $S_t = \emptyset$), the effect moderated by the decision point index $t$ ($S_t = t$), the effect moderated by location ($S_t = \text{whether the current location is home/work (vs. other location)}$), and the effect moderated by weekday ($S_t = \text{whether the day is a weekday}$). In all analyses, the $g(H_t)$ in the working model for $E\{Y_{t, \Delta}(\bar{A}_{t-1}, 0, \bar{0}_{\Delta - 1}) \mid H_{t}, I_{t}=1 \}$ includes the decision point index, the total step count in the 30 minutes prior to the decision point, the current location (home/work vs. others), and whether the day is a weekday. For pd-EMEE2, the same control variables corresponding to decision point $u$ is used to fit $\mu_s(H_{iu}, A_{iu})$.}

\begin{figure}[htbp]
    \begin{center}
    \caption{Estimated marginal excursion effect (Moderator: None) and effect moderation by decision point index, whether the participant is at home/work, whether the day is weekday of HeartSteps activity suggestion on whether the user takes more than a given threshold of steps in the next 24 hours. For each moderator and estimand, the relative efficiency of pd-EMEE or pd-EMEE2 against EMEE with various step count threshold is shown in the corresponding dot figure below.}
    \includegraphics[width = 1\textwidth]{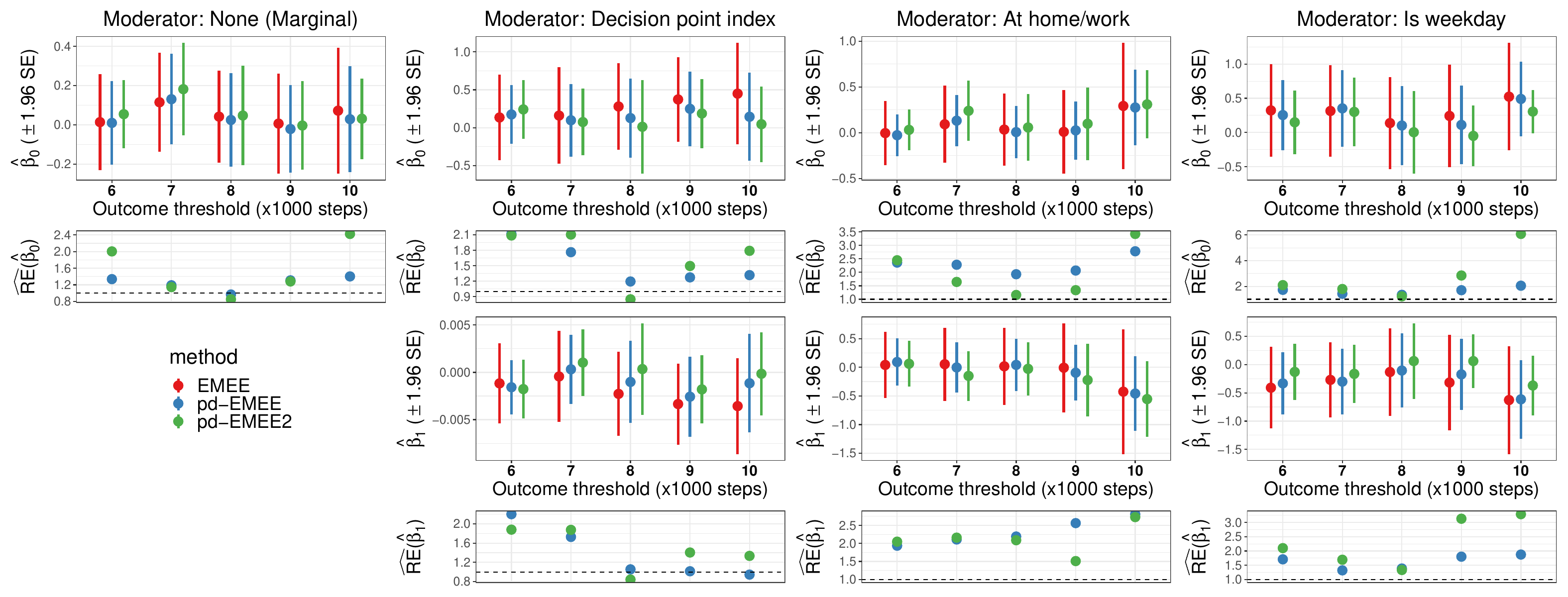}
    \label{fig:HeartSteps}
    \end{center}
\end{figure}

\readdBiostat{Figure \ref{fig:HeartSteps} shows the analysis results, and the format of the figure is the same as Figure \ref{fig:Drinkless}. We see that for each estimand and for each threshold $C$, EMEE, pd-EMEE, and pd-EMEE2 yield estimates in the same direction with CIs mostly overlapping. Both pd-EMEE and pd-EMEE2 yield substantially narrower CIs than EMEE for most threshold $C$ values. The efficiency gain of pd-EMEE and pd-EMEE2 against EMEE is not as pronounced as in the Drink Less application but is still substantial, with $\widehat{\text{RE}}$ mostly in the range of 1.2--3.}

\section{Discussion}
\label{sec:discussion}

We proposed the per-decision estimator of marginal excursion effect (pd-EMEE) and \readdBiostat{its projection-based variant (pd-EMEE2)} for the causal excursion effect of a time-varying treatment on a binary longitudinal outcome for micro-randomized trials. The proposed estimators are applicable as long as the binary outcome encodes the occurrence of an event of interest in a pre-specified time window. When this time window spans multiple decision points, the proposed estimator improves efficiency compared to the currently widely-used EMEE estimator \citep{qian2021estimating}, by using a modified form of inverse probability weighting (IPW) called per-decision IPW. \readdBiostat{The pd-EMEE2, motivated by results from the semiparametric efficiency theory, further improves efficiency. This efficiency improvement can lead to a more precise estimation of the causal effect or a smaller sample size needed when planning an MRT.} The efficiency improvement is greater when there is a higher chance of one or more additional treatments occurring in the proximal outcome window; for example, when the proximal outcome window consists of more decision points or when the randomization probability is larger.
% When the proximal outcome window doesn't contain additional decision points (i.e., when the proximal outcome is measured immediately), the proposed pd-EMEE is equivalent to the EMEE in \citep{qian2021estimating}.

The methods require that the binary proximal outcome satisfies a maximum property, i.e., the outcome can be expressed as the maximum of a series of sub-outcomes defined over sub-intervals of time. This holds when the outcome indicates the occurrence of an instantaneous event. 
For non-instantaneous events, if it is still possible to determine the first decision point, within the proximal outcome window, at which the binary outcome has already occurred, then one can use a similar estimator for efficiency improvement by assuming a generalized maximum property. The generalization is detailed in Supplementary Material \ref{appen:sec:generalized-max-property}, and estimators based on this generalization are illustrated in Section \ref{subsec:heartsteps-application}.

The estimator can also be generalized to settings where the number of decision points within a proximal outcome window is not fixed. Such a setting may arise if the decision points are irregularly spaced and the proximal outcome window is based on calendar time. In this setting, as long as the occurrence of the event may be determined early on in the proximal outcome window, one can still apply the per-decision IPW idea and thus discard the IPW terms after the event occurs to improve efficiency, with the necessary notational generalization.

% Equation \eqref{eq:linear-model} assumes that the causal excursion effect is linear in $S_t$. This can be immediately generalized to transformations of $S_t$ by replacing $S_t$ with $f_t(S_t)$ for arbitrary fixed function $f_t$.

% Our estimator can also be generalized to observational studies where the time-varying treatments are not sequentially randomized. One can replace $p_t(H_t)$ in the estimating equation by a predicted propensity score from, say, a logistic regression fit. In this case, the validity of the estimator would require a correct propensity score model. Alternatively, one could use the projection idea in \citep{guo2021discussion} to construct a doubly-robust version of the estimator for observational studies.

% The idea of only multiplying with weights at the relevant decision points also appeared in the reinforcement learning literature. Even though we independently came up with the per-decision IPW idea, we borrowed the terminology ``per-decision'' from \citep{precup2000eligibility}.

The maximum property may resemble the structure of a discrete-time survival outcome, and may raise the question of whether the proximal outcome should be treated as a survival outcome instead of a binary one. We believe the choice of the model and outcome type should depend on the scientific question. For some interventions, the goal is so that a desirable event, such as self-monitoring in Drink Less and activity bout in HeartStep, occurs within a reasonable time window but not necessarily as soon as possible. For such settings, the causal effect on the binary outcome will contain relevant information for assessing and improving the intervention. For other settings, the goal might be to reduce the risk of a dangerous or harmful behavior as soon as possible, in which case the causal effect on the survival outcome might be more relevant. An example of the latter is problem anger prevention \citep{metcalf2022anger}, where the intervention is to reduce the risk of aggressive behavior once anger is detected.

We considered the same causal excursion effect as in \citet{qian2021estimating}, where the future $\Delta-1$ treatment assignments ($A_{t+1},\ldots,A_{t+\Delta-1}$) are all set to 0. That is, we consider excursions starting from decision point $t$ that receives treatment (or no treatment) at $t$, and then receive no treatment for the next $\Delta-1$ decision points. The choice of this reference treatment regime depends on the scientific question, and in some applications other reference regimes may be of interest. In Supplementary Material \ref{appen:sec:reference-regime} we provide an extension of the proposed method to a class of alternative reference regimes and evaluate the efficiency gain under those settings.

% \section{Software}
% \label{sec:software}
% R (R Development Core Team \citepalp{Rlanguage}) code to reproduce the simulations and the real data analysis can be downloaded at \url{https://github.com/tqian/EMEE_binary_modified_weight}. The Drink Less MRT data is publicly available at \url{https://osf.io/mtcfa}. \readdBiostat{The HeartSteps I MRT data is publicly available at \url{https://github.com/klasnja/HeartStepsV1}.}

% \section{Supplementary Material}
% \label{sec6}
%\input{appendix_new}
% Supplementary material is available online at
% % \href{http://biostatistics.oxfordjournals.org}%
% % {http://biostatistics.oxfordjournals.org}.
% \url{http://biostatistics.oxfordjournals.org}.

% \section*{Acknowledgments}
% Lauren Bell is supported by a PhD studentship funded by the MRC Network of Hubs for Trials Methodology Research (MR/L004933/2-R18).

%\bibliographystyle{biorefs}
\newpage
\bibliographystyle{biorefs} 
\bibliography{main}

\newpage
\appendix
\section*{Supplementary Materials}

\numberwithin{equation}{section}

\appendix
\renewcommand{\thesubsection}{\Alph{section}.\arabic{subsection}}
\setcounter{section}{0}

\section{Interpretation of the Causal Excursion Effect and Its Use in MRT}
\label{appen:sec:CEE-MRT}

The causal excursion effect (\ref{subsec:cee-definition}) differs from most causal effects under time-varying treatments such as marginal structural models and structural nested mean models \citep{robins1994correcting, robins2000marginal}. Instead of contrasting fixed treatment trajectories, \eqref{eq:causal-effect-def} contrasts between two stochastic treatment trajectories, $(\bar{A}_{t-1}, 1, \bar{0}_{\Delta - 1})$ and $(\bar{A}_{t-1}, 0, \bar{0}_{\Delta - 1})$. That is, \eqref{eq:causal-effect-def} is a causal contrast between two excursions from the treatment policy: following the treatment policy up to time $t - 1$ and then deviating from the policy to assign $A_t = 1$ and no treatment for the next $\Delta-1$ decision points, and following the treatment policy up to time $t-1$ and then deviating to assign $A_t = 0$ and no treatment for the next $\Delta-1$ decision points. \eqref{eq:causal-effect-def} is conditional on $I_t(\bar{A}_{t-1})=1$ and $S_t(\bar{A}_{t-1})$, meaning that it concerns the subpopulation who, after following the treatment policy up to time $t - 1$, are available to be randomized at $t$ and within a stratum defined by $S_t(\bar{A}_{t-1})$. In practice we often impose a model on how $\est \left\{S_t(\bar{A}_{t-1}) \right\}$ depends on $S_t(\bar{A}_{t-1})$ to pool across $S_t$-strata.

$\est \left\{S_t(\bar{A}_{t-1}) \right\}$ depends on the treatment policy in the MRT, because the variables in $H_t(\bar{A}_{t-1}) \setminus S_t(\bar{A}_{t-1})$ are marginalized over in the conditional expectations in \eqref{eq:causal-effect-def}. Such dependence is scientifically desirable for two reasons. First, the causal excursion effect approximates the treatment effect in a real-world implementation. A good MRT policy would already have incorporated implementation considerations such as burden and feasibility through the choice of randomization probability and availability criteria, so a feasible policy will not deviate too far from the MRT policy. Second, the causal excursion effect indicates the effective deviations from the MRT policy and how the policy might be improved. A fully marginal effect (by setting $S_t(\bar{A}_{t-1}) = 1$ in \eqref{eq:causal-effect-def}) indicates whether a treatment is worth further investigation, and an effect modification analysis (by setting $S_t(\bar{A}_{t-1})$ to be certain time-varying covariates in \eqref{eq:causal-effect-def}) indicates whether the MRT policy should be modified to depend on time-varying covariates. In addition, the dependence of $\est \left\{S_t(\bar{A}_{t-1}) \right\}$ on the MRT policy resembles the primary analysis in factorial designs and allows one to design trials with a higher power to detect meaningful effects. Related marginalization ideas were considered by \citet{robins2004optimal, neugebauer2007causal}. See \citet{boruvka2018assessing}, \citet{qian2021estimating}, and \citet{qian2021micro} for more discussion on causal excursion effect and its use in MRT. See \citet{guo2021discussion} for a comprehensive comparison of various causal estimands under time-varying treatments.

\section{Generalized Version of Maximum Property}
\label{appen:sec:generalized-max-property}

We present a generalized version of the maximum property and a double-subscript version of the variable $R$. Using this formulation, we can apply methodology developed in the paper to a broader range of problems as long as it is possible to determine the first decision point, within the proximal outcome window, at which the binary outcome has already occurred.

Define $R_{t, t+s}$ to be the indicator of whether the event of interest occurred between decision points $t$ and $t+s$. Under this definition, we have the following generalized maximum property:
\begin{align}
    Y_{t,\Delta} = \max(R_{t,t+1}, R_{t,t+2}, \ldots, R_{t,t+\Delta}). \label{eq:appen-generalized-maximum-property}
\end{align}
The definition of $R_{t,t+s}$ and its relationship with $Y_{t,\Delta}$ is illustrated in Web Figure \ref{fig:generalized-max-property} for the case of $\Delta = 3$. The estimator and the proofs work for this more general definition with minimal modification: One simply needs to replace $R_{t+s}$ by $R_{t,t+s}$ throughout.

For example, in the HeartSteps II MRT \citep{liao2020personalized}, the proximal outcome was whether there is an activity bout in the 30 minutes following a decision point. There was a decision point every 5 minutes, so we denote by $Y_{t,\Delta = 6}$ the proximal outcome using our notation. Here $R_{t, t+1}$ is whether an activity bout occurred between $t$ and $t+1$, $R_{t, t+2}$ is whether an activity bout occurred between $t$ and $t+2$, and so on. Here \eqref{eq:appen-generalized-maximum-property} holds with $\Delta = 6$.

For the rest of the supplementary material, we assume that the original maximum property \eqref{eq:maximum-property} in the main paper holds.

\begin{figure}[htbp]
    \begin{center}
    \caption{Illustration of $R_{t,t+s}$ and its relationship with $Y_{t,\Delta}$ for the case of $\Delta = 3$.}
    \includegraphics[width = \textwidth]{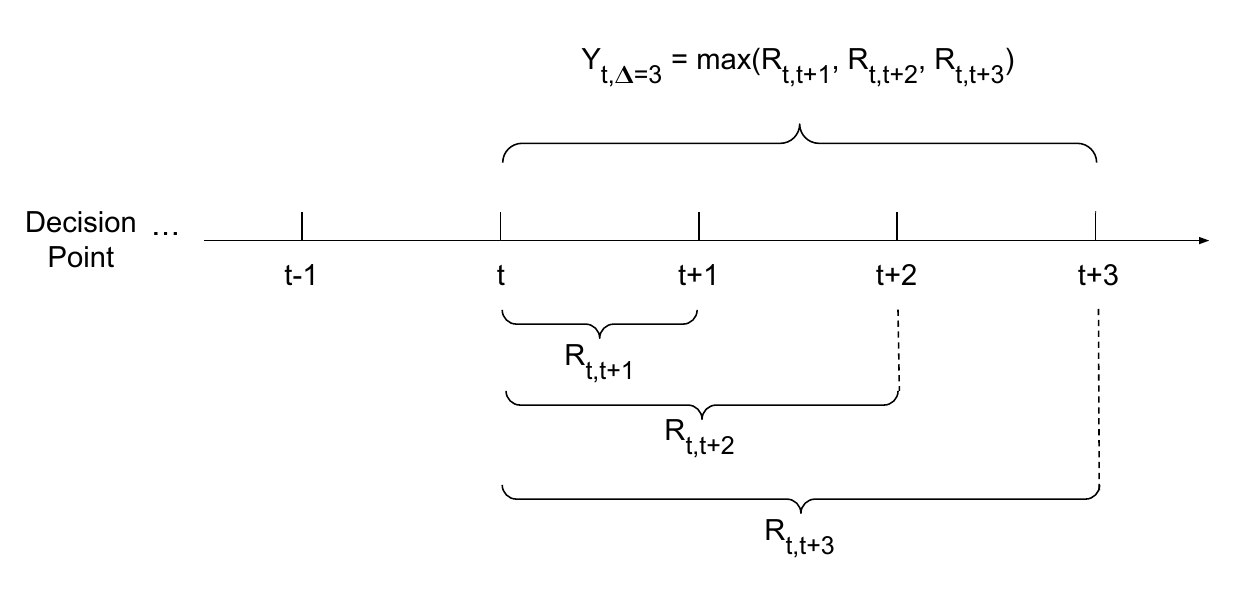}
    \label{fig:generalized-max-property}
    \end{center}
\end{figure}

\section{Causal Assumptions and Parameter Identification}
\label{appen:sec:causal-assumption-and-identifiability-proof}

We use overbar and double subscript to denote the sequence of treatments
between two time points: for example, for $t_2\geq t_1$, $\bar{A}_{t_1:t_2}:=(A_{t_1},A_{t_1+1},\ldots,A_{t_2})$. When there is no confusion, we will omit the subscript on $\bar{0}$: for example, $Y_{t,\Delta}(\bar{A}_{t-1}, a, \bar{0}) = Y_{t,\Delta}(\bar{A}_{t-1}, a, \bar{0}_{\Delta-1})$.

Under Assumptions \ref{assumption1}-\ref{assumption4}, we have 
\begin{align}
     & E\left\{ Y_{t,\Delta}(\bar{A}_{t-1},a,\bar{0})\mid A_t=a,H_t,I_t=1\right\} \\
    = & E\left\{ \prod_{j=t+1}^{t+\Delta-1}\left[\frac{1(A_j=0)}{P(A_j=0\mid H_j)}\right]^{1(\max_{1\leq s\leq j-t}R_{t+s}=0)}Y_{t,\Delta}\mid A_t=a,H_t,I_t=1\right\}, \label{eq:id1}
\end{align}
which implies the identification result \eqref{eq:identification} in the main paper. Equation \eqref{eq:id1} follows immediately from the statements of Lemma \ref{lem:id3} and Lemma \ref{lem:id4} below. We first define a set of new quantities used in the proof, then state and prove Lemmas \ref{lem:id1}-\ref{lem:id4}.

Consider any fixed $t$ with $1 \leq t \leq m$. For each $1\leq u\leq\Delta-1$,
let $C_{tu}(\bar{A}_{t-1},a,\bar{0}_{u-1})$ denote the event that the first
nonzero $R_{t+s}(\bar{A}_{t-1},a,\bar{0})$ after time $t$ occurs
at $s=u$; for notation simplicity we write $C_{tu}(\bar{A}_{t-1},a,\bar{0}_{u-1})$ as $C_{tu}$. That is, for $1\leq u\leq\Delta-1$
\begin{align*}
    C_{tu}:=\left\{ \max_{1\leq s\leq u-1}R_{t+s}(\bar{A}_{t-1},a,\bar{0})=0\text{ and }R_{t+u}(\bar{A}_{t-1},a,\bar{0})=1\right\} .
\end{align*}
Let $C_{t\Delta} := C_{t\Delta}(\bar{A}_{t-1},a,\bar{0}_{\Delta-1})$ denote
the event $\left\{ \max_{1\leq s\leq\Delta-1}R_{t+s}(\bar{A}_{t-1},a,\bar{0})=0\right\}$.
(The definition of $C_{t\Delta}$ is different from the rest $C_{tu}$. This is intentional because, as we will see shortly, there is no need to distinguish between $R_{t+\Delta}(\bar{A}_{t-1},a,\bar{0})=1$ and $R_{t+\Delta}(\bar{A}_{t-1},a,\bar{0})=0$ for the purpose of proving identification. Also, omitting the potential treatment trajectory in $C_{tu}$ and $C_{t\Delta}$ will not cause confusion because we always condider the treatment trajectory $\bar{A}_{t-1}, a, \bar{0}_{\Delta-1}$ in this section.)

\begin{lem}\label{lem:id1}
    If $\Delta \geq 2$, then for all $1\leq t\leq m$, for all $1\leq u\leq\Delta-1$, for all $a_t \in \{0,1\}$ and for all $\bar{a}_{t+u:t+\Delta-1}\in\{0,1\}^{\otimes(\Delta-u)}$,
    \begin{align}
         Y_{t,\Delta}(\bar{A}_{t-1},a_t,\bar{0}_{\Delta-1})\times \mathbbm{1}(C_{tu}) =  Y_{t,\Delta}(\bar{A}_{t-1},a_t,\bar{0}_{u-1},\bar{a}_{t+u:t+\Delta-1})\times \mathbbm{1}(C_{tu}). \label{eq:lem-id0}
    \end{align}
\end{lem}

\begin{proof}[Proof of Lemma~\ref{lem:id1}]
    When $\mathbbm{1}(C_{tu}) = 0$, \eqref{eq:lem-id0} holds trivially. When $\mathbbm{1}(C_{tu}) = 1$, it must be the case that $R_{t+u}(\bar{A}_{t-1},a_t,\bar{0}_{u-1})=1$, and Assumption \ref{assumption4} implies that $Y_{t,\Delta}(\bar{A}_{t-1},a_t,\bar{0}_{u-1},\bar{a}_{t+u:t+\Delta-1})$ and $Y_{t,\Delta}(\bar{A}_{t-1},a_t,\bar{0}_{\Delta-1})$ both equal to 1. Thus \eqref{eq:lem-id0} holds when the indicator function equals 1. This completes the proof.
    % Assumption \ref{assumption4} implies 
    % \begin{align*}
    %     Y_{t,\Delta}(\bar{A}_{t-1},a_t,\bar{0}_{u-1},\bar{a}_{t+u:t+\Delta-1}) = & \max\{ R_{t+1}(\bar{A}_{t-1}, a_t), R_{t+2}(\bar{A}_{t-1}, a_t, 0), \ldots, \\
    %     & R_{t+u-1}(\bar{A}_{t-1}, a_t, \bar{0}_{u-2}), R_{t+u}(\bar{A}_{t-1}, a_t, \bar{0}_{u-1}), R_{t+u+1}(\bar{A}_{t-1}, a_t, \bar{0}_{u-1},a_{t+u}), \ldots, \\
    %     & R_{t+\Delta-1}(\bar{A}_{t-1},a_t,\bar{0}_{u-1},\bar{a}_{t+u:t+\Delta-2}), R_{t+\Delta}(\bar{A}_{t-1},a_t,\bar{0}_{u-1},\bar{a}_{t+u:t+\Delta-1}) \}.
    % \end{align*}
\end{proof}

\begin{lem}\label{lem:id2}
    Suppose Assumptions \ref{assumption1}-\ref{assumption4} hold. Given fixed $1\leq t\leq m$ and $1\leq u\leq\Delta$, for all $1\leq k\leq u$
    we have
    \begin{align*}
         & E\left\{ Y_{t,\Delta}(\bar{A}_{t-1},a,\bar{0})\mid A_t=a,H_t,I_t=1,C_{tu}\right\} \\
        = & E\left\{ \prod_{j=t+1}^{t+k-1}\frac{1(A_j=0)}{P(A_j=0\mid H_j)}\times Y_{t,\Delta}(\bar{A}_{t-1},a,\bar{0})\mid A_t=a,H_t,I_t=1,C_{tu}\right\}.
    \end{align*}
    Note that for $k=1$, we define $\prod_{j=t+1}^{t+k-1}\frac{1(A_j=0)}{p(A_j=0\mid H_j)}=1$.
\end{lem}

\begin{proof}[Proof of Lemma~\ref{lem:id2}]
    We prove by induction. $k=1$ holds trivially. Now suppose $k=k_0$
    for some $1\leq k_0\leq u-1$ holds and we are to show $k=k_0+1$
    holds. We have
    \begin{align}
         & E\left\{ Y_{t,\Delta}(\bar{A}_{t-1},a,\bar{0})\mid A_t=a,H_t,I_t=1,C_{tu}\right\} \nonumber \\
        = & E\left\{ \prod_{j=t+1}^{t+k_0-1}\frac{\mathbbm{1}(A_j=0)}{p(A_j=0\mid H_j)}\times Y_{t,\Delta}(\bar{A}_{t-1},a,\bar{0})\mid A_t=a,H_t,I_t=1,C_{tu}\right\} \label{eq:proof5}\\
        = & E\left[E\left\{ \prod_{j=t+1}^{t+k_0-1}\frac{\mathbbm{1}(A_j=0)}{p(A_j=0\mid H_j)}\times Y_{t,\Delta}(\bar{A}_{t-1},a,\bar{0}) \mid A_t=a, H_{t+k_0}, I_t=1, C_{tu}\right\} \mid A_t=a,H_t,I_t=1,C_{tu}\right]\label{eq:proof6}\\
        = & E\left[E\left\{ \prod_{j=t+1}^{t+k_0-1}\frac{\mathbbm{1}(A_j=0)}{p(A_j=0\mid H_j)}\times Y_{t,\Delta}(\bar{A}_{t-1},a,\bar{0})\mid A_t=a, H_{t+k_0}, I_t=1, C_{tu}\right\} \right.\nonumber\\
         & \left.\times E\left\{ \frac{\mathbbm{1}(A_{t+k_0}=0)}{P(A_{t+k_0}=0\mid A_t=a, H_{t+k_0}, I_t=1, C_{tu})}\mid A_t=a, H_{t+k_0}, I_t=1, C_{tu}\right\} \mid A_t=a,H_t,I_t=1,C_{tu}\right] \label{eq:proof7} \\
        = & E\left[E\left\{ \prod_{j=t+1}^{t+k_0}\frac{\mathbbm{1}(A_j=0)}{p(A_j=0\mid H_j)}\times Y_{t,\Delta}(\bar{A}_{t-1},a,\bar{0})\mid A_t=a, H_{t+k_0}, I_t=1, C_{tu}\right\} \mid A_t=a,H_t,I_t=1,C_{tu}\right]\label{eq:proof8}\\
        = & E\left\{ \prod_{j=t+1}^{t+k_0}\frac{\mathbbm{1}(A_j=0)}{p(A_j=0\mid H_j)}\times Y_{t,\Delta}(\bar{A}_{t-1},a,\bar{0})\mid A_t=a,H_t,I_t=1,C_{tu}\right\} .\label{eq:proof9}
    \end{align}
    Here, \eqref{eq:proof5} follows from the induction assumption, \eqref{eq:proof6}
    and \eqref{eq:proof9} follow from the law of iterated expectation,
    \eqref{eq:proof7} follows from the fact that $E\left\{ \frac{\mathbbm{1}(A_{t+k_0}=0)}{P(A_{t+k_0}=0\mid A_t=a, H_{t+k_0}, I_t=1, C_{tu})}\mid A_t=a, H_{t+k_0}, I_t=1, C_{tu}\right\} =1$
    (valid under Assumptions \ref{assumption2} and \ref{assumption3}), and \eqref{eq:proof8} follows
    from Assumption \ref{assumption2} (in particular, conditional
    on $H_{t+k_0}$, $A_{t+k_0}$ is independent of all potential
    outcomes including $Y_{t,\Delta}$ and $C_{tu}$). This completes the proof.
\end{proof}

\begin{lem}\label{lem:id3}
    We have
    \begin{align}
         & E\left\{ Y_{t,\Delta}(\bar{A}_{t-1},a,\bar{0})\mid A_t=a,H_t,I_t=1\right\} \nonumber \\
        = & E\left\{ \prod_{j=t+1}^{t+\Delta-1}\left[\frac{1(A_j=0)}{P(A_j=0\mid H_j)}\right]^{1(\max_{1\leq s\leq j-t}R_{t+s}(\bar{A}_{t-1},a,\bar{0})=0)}Y_{t,\Delta}(\bar{A}_{t-1},a,\bar{0})\mid A_t=a,H_t,I_t=1\right\} \label{eq:proof2}
    \end{align}
    holds for all $t=1,\ldots,m$.
\end{lem}

\begin{proof}[Proof of Lemma~\ref{lem:id3}]
    Because $\{C_{tu}:1\leq u\leq\Delta\}$ is a partition of the whole sample space, for any random variable $Z$ defined on the same sample space we have
    \begin{align*}
        E\left\{ Z\mid A_t=a,H_t,I_t=1\right\}
        = \sum_{u=1}^{\Delta}E\left\{ Z\mid A_t=a,H_t,I_t=1,C_{tu}\right\} P\left\{ C_{tu}\mid A_t=a,H_t,I_t=1\right\} .
    \end{align*}
    Therefore, to prove \eqref{eq:proof2} it suffices to show that for any
    $1\leq u\leq\Delta$,
    \begin{align}
         & E\left\{ Y_{t,\Delta}(\bar{A}_{t-1},a,\bar{0})\mid A_t=a,H_t,I_t=1,C_{tu}\right\} \nonumber \\
        = & E\left\{ \prod_{j=t+1}^{t+\Delta-1}\left[\frac{1(A_j=0)}{P(A_j=0\mid H_j)}\right]^{1(\max_{1\leq s\leq j-t}R_{t+s}(\bar{A}_{t-1},a,\bar{0})=0)}Y_{t,\Delta}(\bar{A}_{t-1},a,\bar{0})\mid A_t=a,H_t,I_t=1,C_{tu}\right\} .\label{eq:proof3}
    \end{align}

    Conditional on $C_{tu}$ (for a fixed $1\leq u\leq\Delta$), 
    \begin{align*}
        \prod_{j=t+1}^{t+\Delta-1}\left[\frac{1(A_j=0)}{P(A_j=0\mid H_j)}\right]^{1(\max_{1\leq s\leq j-t}R_{t+s}(\bar{A}_{t-1},a,\bar{0})=0)}=\prod_{j=t+1}^{t+u-1}\frac{1(A_j=0)}{P(A_j=0\mid H_j)}.
    \end{align*}
    So right hand side of \eqref{eq:proof3} equals
    \begin{equation}
        E\left\{ \prod_{j=t+1}^{t+u-1}\frac{1(A_j=0)}{P(A_j=0\mid H_j)}\times Y_{t,\Delta}(\bar{A}_{t-1},a,\bar{0})\mid A_t=a,H_t,I_t=1,C_{tu}\right\} .\label{eq:proof4}
    \end{equation}
    Display \eqref{eq:proof4} equals $E\left\{ Y_{t,\Delta}(\bar{A}_{t-1},a,\bar{0})\mid A_t=a,H_t,I_t=1,C_{tu}\right\} $
    by setting $k=u$ in Lemma \ref{lem:id2}. This completes the proof.
\end{proof}
\begin{lem}\label{lem:id4}
    Suppose Assumptions \ref{assumption1}-\ref{assumption4} hold. We have
    \begin{align}
         & E\left\{ \prod_{j=t+1}^{t+\Delta-1}\left[\frac{\mathbbm{1}(A_j=0)}{P(A_j=0\mid H_j)}\right]^{\mathbbm{1}(\max_{1\leq s\leq j-t}R_{t+s}(\bar{A}_{t-1},a,\bar{0})=0)}Y_{t,\Delta}(\bar{A}_{t-1},a,\bar{0})\mid A_t=a,H_t,I_t=1\right\} \nonumber \\
        = & E\left\{ \prod_{j=t+1}^{t+\Delta-1}\left[\frac{\mathbbm{1}(A_j=0)}{P(A_j=0\mid H_j)}\right]^{\mathbbm{1}(\max_{1\leq s\leq j-t}R_{t+s}=0)}Y_{t,\Delta}\mid A_t=a,H_t,I_t=1\right\} .\label{eq:proof10}
    \end{align}
\end{lem}

\begin{proof}[Proof of Lemma~\ref{lem:id4}]
    We prove Lemma \ref{lem:id4} by showing that \eqref{eq:proof10} holds when conditioning both sides on $C_{tu}$ for each $1\leq u\leq\Delta$. Lemma \ref{lem:id4} then follows because $\{C_{tu}: 1 \leq u \leq \Delta\}$ forms a partition of the whole sample space.

    The left hand side of \eqref{eq:proof10} conditional on $C_{t\Delta}$ is
    \begin{align}
         & E\left\{ \prod_{j=t+1}^{t+\Delta-1}\left[\frac{\mathbbm{1}(A_j=0)}{P(A_j=0\mid H_j)}\right]^{\mathbbm{1}(\max_{1\leq s\leq j-t}R_{t+s}(\bar{A}_{t-1},a,\bar{0})=0)}Y_{t,\Delta}(\bar{A}_{t-1},a,\bar{0})\mid A_t=a,H_t,I_t=1,C_{t\Delta}\right\} \nonumber \\
        = & E\left\{ \prod_{j=t+1}^{t+\Delta-1}\left[\frac{\mathbbm{1}(A_j=0)}{P(A_j=0\mid H_j)}\right]Y_{t,\Delta}(\bar{A}_{t-1},a,\bar{0}_{\Delta-1})\mid A_t=a,H_t,I_t=1,C_{t\Delta}\right\} \label{eq:proof11}\\
        = & E\left\{ \prod_{j=t+1}^{t+\Delta-1}\left[\frac{\mathbbm{1}(A_j=0)}{P(A_j=0\mid H_j)}\right]Y_{t,\Delta}(\bar{A}_{t-1},a,\bar{A}_{t+1:t+\Delta-1})\mid A_t=a,H_t,I_t=1,C_{t\Delta}\right\} \label{eq:proof12}\\
        = & E\left\{ \prod_{j=t+1}^{t+\Delta-1}\left[\frac{\mathbbm{1}(A_j=0)}{P(A_j=0\mid H_j)}\right]Y_{t,\Delta}\mid A_t=a,H_t,I_t=1,C_{t\Delta}\right\} ,\label{eq:proof13}
    \end{align}
    which is the right hand side of \eqref{eq:proof10} conditional on $C_{t\Delta}$.
    Here, \eqref{eq:proof11} follows from the definition of $C_{t\Delta}$;
    \eqref{eq:proof12} follows from $\prod_{j=t+1}^{t+\Delta-1}\mathbbm{1}(A_j=0) Y_{t,\Delta}(\bar{A}_{t-1}, a, \bar{0}_{\Delta-1}) = \prod_{j=t+1}^{t+\Delta-1}\mathbbm{1}(A_j=0) Y_{t,\Delta}(\bar{A}_{t-1}, a, \bar{A}_{t+1:t+\Delta-1})$, a consequence of Assumption \ref{assumption1};
    and \eqref{eq:proof13} follows from Assumption \ref{assumption1}.

    The left hand side of \eqref{eq:proof10} conditional on $C_{tu}$ for some $1\leq u\leq\Delta-1$ is
    \begin{align}
         & E\left\{ \prod_{j=t+1}^{t+\Delta-1}\left[\frac{\mathbbm{1}(A_j=0)}{P(A_j=0\mid H_j)}\right]^{\mathbbm{1}(\max_{1\leq s\leq j-t}R_{t+s}(\bar{A}_{t-1},a,\bar{0})=0)}Y_{t,\Delta}(\bar{A}_{t-1},a,\bar{0})\mid A_t=a,H_t,I_t=1,C_{tu}\right\} \nonumber \\
        = & E\left\{ \prod_{j=t+1}^{t+u-1}\left[\frac{\mathbbm{1}(A_j=0)}{P(A_j=0\mid H_j)}\right]Y_{t,\Delta}(\bar{A}_{t-1},a,\bar{0})\mid A_t=a,H_t,I_t=1,C_{tu}\right\} \label{eq:proof14}\\
        = & E\left\{ \prod_{j=t+1}^{t+u-1}\left[\frac{\mathbbm{1}(A_j=0)}{P(A_j=0\mid H_j)}\right]Y_{t,\Delta}(\bar{A}_{t-1},a,\bar{0}_{u-1},\bar{A}_{t+u:t+\Delta-1})\mid A_t=a,H_t,I_t=1,C_{tu}\right\} \label{eq:proof15}\\
        = & E\left\{ \prod_{j=t+1}^{t+u-1}\left[\frac{\mathbbm{1}(A_j=0)}{P(A_j=0\mid H_j)}\right]Y_{t,\Delta}(\bar{A}_{t+\Delta-1})\mid A_t=a,H_t,I_t=1,C_{tu}\right\} \label{eq:proof16}\\
        = & E\left\{ \prod_{j=t+1}^{t+u-1}\left[\frac{\mathbbm{1}(A_j=0)}{P(A_j=0\mid H_j)}\right]Y_{t,\Delta}\mid A_t=a,H_t,I_t=1,C_{tu}\right\}\label{eq:proof17}\\
        = & E\left\{ \prod_{j=t+1}^{t+\Delta-1}\left[\frac{\mathbbm{1}(A_j=0)}{P(A_j=0\mid H_j)}\right]^{\mathbbm{1}(\max_{1\leq s\leq j-t}R_{t+s}=0)}Y_{t,\Delta}\mid A_t=a,H_t,I_t=1,C_{tu}\right\} ,\label{eq:proof18}
    \end{align}
    which is the right hand side of \eqref{eq:proof10} conditional on $C_{tu}$.
    Here, \eqref{eq:proof14} follows from the definition of $C_{tu}$;
    \eqref{eq:proof15} follows from Lemma \ref{lem:id1}; \eqref{eq:proof16}
    follows from $\prod_{j=t+1}^{t+u-1}\mathbbm{1}(A_j=0) Y_{t,\Delta}(\bar{A}_{t-1},a,\bar{0}_{u-1},\bar{A}_{t+u:t+\Delta-1}) = \prod_{j=t+1}^{t+u-1}\mathbbm{1}(A_j=0) Y_{t,\Delta}(\bar{A}_{t-1},a,\bar{A}_{t+1:t+\Delta-1})$, a consequence of Assumption \ref{assumption1}; \eqref{eq:proof17}
    follows from Assumption \ref{assumption1}; and \eqref{eq:proof18} follows from the definition of $C_{tu}$. This completes the proof.
\end{proof}

\section{Proof of Theorem 1}\label{appen:sec:theorem1-proof}
To establish Theorem 1, we assume the following regularity conditions.

%\textit{Assumption 1}
\begin{asu}
    \label{appen:thm1-assumption1}
    Suppose $(\alpha,\beta) \in \Theta$, where $\Theta$ is a compact subset of a Euclidean space. Suppose there exists unique $(\alpha',\beta') \in \Theta$ such that $E[U_i(\alpha',\beta')] = 0$.
\end{asu}

%\textit{Assumption 2}. 
\begin{asu}
    \label{appen:thm1-assumption2}
    Suppose $S_t$, $\exp(S_t)$, $g(H_t)$ and $\exp \{g(H_t)\}$ all have finite fourth moment.
\end{asu}

\begin{asu}
    \label{appen:thm1-assumption3}
    Suppose $U_i(\alpha,\beta)$ and $U_i(\alpha,\beta)U_i(\alpha,\beta)^T$ are both bounded by integrable functions.
\end{asu}

\begin{lem}\label{lem:consistency}
    Suppose Equation \eqref{eq:linear-model} and Assumptions \ref{assumption1}-\ref{assumption4} in the main paper hold. $\beta^*$ is the value of $\beta$ corresponding to the data generating distribution, $P_0$. For an arbitrary $g(H_t)^T$ and an arbitrary $\alpha$, we have
    \begin{align}
        E\left[I_t e^{-A_t S_t^{T} \beta^{*}}\left\{Y_{t, \Delta}-e^{g\left(H_t\right)^{T} \alpha+A_t S_t^{T} \beta^{*}}\right\} M_t W_t \left\{A_t-\tilde{p}_t\left(S_t\right)\right\} S_t\right]=0. \nonumber
    \end{align}
\end{lem}

\begin{proof}[Proof of Lemma~\ref{lem:consistency}]
    We have
    \begin{adjustwidth}{-3em}{-3em}
    \begin{align}
        &E\left[I_t e^{-A_t S_t^{T} \beta^{*}}\left\{Y_{t, \Delta}-e^{g\left(H_t\right)^{T} \alpha+A_t S_t^{T} \beta^{*}}\right\} M_t W_t\left\{A_t-\tilde{p}_t\left(S_t\right)\right\} S_t\right] \nonumber \\
        =& E\left(E\left[I_t e^{-A_t S_t^{T} \beta^{*}}\left\{Y_{t, \Delta}-e^{g\left(H_t\right)^{T} \alpha+A_t S_t^{T} \beta^{*}}\right\} M_t W_t\left\{A_t-\tilde{p}_t\left(S_t\right)\right\} S_t \mid H_t\right]\right) \label{eq:B1}\\
        =& E\left(E\left[e^{-A_t S_t^{T} \beta^{*}}\left\{Y_{t, \Delta}-e^{g\left(H_t\right)^{T} \alpha+A_t S_t^{T} \beta^{*}}\right\} M_t W_t\left\{A_t-\tilde{p}_t\left(S_t\right)\right\} \mid H_t, I_t=1\right] I_t S_t\right) \label{eq:B2}\\
        =& E\left(E\bigg[e^{-S_t^{T} \beta^{*}}\left\{Y_{t, \Delta}-e^{g\left(H_t\right)^{T} \alpha+S_t^{T} \beta^{*}}\right\}\left\{1-\tilde{p}_t\left(S_t\right)\right\} \prod_{j=t+1}^{t+\Delta-1} \big[\frac{1\left(A_j=0\right)}{1-p_j\left(H_j\right)}\big]
        ^{\mathbbm{1} (\max\limits_{1 \leq s \leq j-t} R_{t+s} = 0)}  \mid H_t, I_t=1, A_t=1 \bigg] I_t \tilde{p}_t\left(S_t\right) S_t\right)  \nonumber \\ 
        &-E\left(E\bigg[\big\{Y_{t, \Delta}-e^{g\left(H_t\right)^{T} \alpha}\big\} \tilde{p}_t\left(S_t\right) \prod_{j=t+1}^{t+\Delta-1} \big[\frac{1 \left(A_j=0\right)}{1-p_j\left(H_j\right)}\big]^{\mathbbm{1} (\max\limits_{1 \leq s \leq j-t} R_{t+s} = 0)} \mid H_t, I_t=1, A_t=0\bigg] I_t\left\{1-\tilde{p}_t\left(S_t\right)\right\} S_t\right) \label{eq:B3}\\
        =& E\left[\left\{e^{-S_t^{T} \beta^{*}} E\Bigg(\prod_{j=t+1}^{t+\Delta-1} \bigg[\frac{\mathbbm{1}\left(A_j=0\right)}{1-p_j\left(H_j\right)}\bigg]^{\mathbbm{1} (\max\limits_{1 \leq s \leq j-t} R_{t+s} = 0)} Y_{t, \Delta} \mid H_t, I_t=1, A_t=1\Bigg)\right.\right. \nonumber\\ 
        &\left. \left.-E\Bigg(\prod_{j=t+1}^{t+\Delta-1} \bigg[\frac{\mathbbm{1}\left(A_j=0\right)}{1-p_j\left(H_j\right)}\bigg]^{\mathbbm{1} (\max\limits_{1 \leq s \leq j-t} R_{t+s} = 0)} Y_{t, \Delta} \mid H_t, I_t=1, A_t=0\Bigg)\right\} I_t \tilde{p}_t\left(S_t\right)\left\{1-\tilde{p}_t\left(S_t\right)\right\} S_t\right] \label{eq:B4}\\
        =& 0. \label{eq:B5}
    \end{align}
    \end{adjustwidth}
    Here, \eqref{eq:B1} follows from the law of iterated expectation; \eqref{eq:B2} and \eqref{eq:B3} follows from expanding the expectation on $I_t$ and $A_t$ respectively; \eqref{eq:B4} follows from collecting terms; and \eqref{eq:B5} follows from the identifiability result (4) and the linear model (5) in the main paper, and the fact that $\beta^*$ is the true parameter value. This completes the proof.
\end{proof}

\begin{proof}[Proof of Theorem \ref{thm:CAN}.]
    Assumption \ref{appen:thm1-assumption1} implies that $(\hat{\alpha}, \hat{\beta}) \stackrel{p}{\to} \left(\alpha', \beta'\right)$ by Theorem $5.9$ of \citet{van2000asymptotic}. Assumption \ref{appen:thm1-assumption1} (uniqueness of $\beta'$) and Lemma \ref{lem:consistency} imply that $\beta^{*}=\beta'$. Therefore, $(\hat{\alpha}, \hat{\beta}) \stackrel{p}{\to} \left(\alpha', \beta^*\right)$, and in particular $\hat\beta$ is consistent for $\beta^*$.

    Because $U_i(\alpha, \beta)$ is continuously differentiable and hence Lipschitz continuous, Theorem $5.21$ of \citet{van2000asymptotic} implies that $\sqrt{n} \big\{(\hat{\alpha}, \hat{\beta})-\left(\alpha', \beta^*\right)\big\}$ is asymptotically normal with mean zero and covariance matrix $\Sigma = \big\{E[\partial_\beta U_i(\alpha',\beta^*)] \big\}^{-1} E\left[ U_i(\alpha', \beta^*)U_i(\alpha', \beta^*)^T \right] \big\{E[ \partial_\beta U_i(\alpha', \beta^*)] \big\}^{-1,T}$. 

    \readd{
    Lastly, we show that this covariance matrix can be consistently estimated by 
    \begin{align*}
          \left\{\frac{1}{n}\sum_{i=1}^n \partial_\beta U_i(\hat{\alpha},\hat{\beta}) \right\}^{-1} \left\{\frac{1}{n}\sum_{i=1}^n U_i(\hat{\alpha}, \hat{\beta}) U_i(\hat{\alpha}, \hat{\beta})^T\right\} \left\{\frac{1}{n}\sum_{i=1}^n \partial_\beta U_i(\hat{\alpha}, \hat{\beta})\right\}^{-1,T}.
    \end{align*}
    This is achieved by proving
    \begin{align}
        \frac{1}{n}\sum_{i=1}^n \partial_\beta U_i(\hat{\alpha},\hat{\beta}) \stackrel{p}{\to} E[\partial_\beta U_i(\alpha',\beta^*)] \label{appen:eq:thm1-proof1}
    \end{align}
    and 
    \begin{align}
        \frac{1}{n}\sum_{i=1}^n U_i(\hat{\alpha}, \hat{\beta}) U_i(\hat{\alpha}, \hat{\beta})^T \stackrel{p}{\to} E\left[ U_i(\alpha', \beta^*)U_i(\alpha', \beta^*)^T \right], \label{appen:eq:thm1-proof2}
    \end{align}
    combined with Slutsky's theorem.

    We now prove \eqref{appen:eq:thm1-proof1} and \eqref{appen:eq:thm1-proof2}. Consider a function $f_i(\alpha,\beta)$, which we will later set to be either $\partial_\beta U_i(\alpha,\beta)$ or $U_i(\alpha,\beta)U_i(\alpha,\beta)^T$. We define some standard empirical process notation. Let $\mathbbm{P}_n$ denote the empirical average over $i=1,2,\ldots,n$, and let $\mathbbm{P}$ denote the expectation over a new observation, i.e., $\mathbbm{P} f_i(\hat\alpha, \hat\beta) = \int f_i(\hat\alpha, \hat\beta) d\mathbbm{P}$ where in the integral $\hat\alpha$ and $\hat\beta$ are considered fixed and the equation integrates over the distribution of a new observation. Then $E\{f_i(\alpha,\beta)\} = \mathbbm{P}\{f_i(\alpha,\beta)\}$. We have
    \begin{align}
        & ~~ \mathbbm{P}_n\{f_i(\hat\alpha,\hat\beta)\} - \mathbbm{P}\{f_i(\alpha,\beta)\} \nonumber \\
        & = [\mathbbm{P}_n\{f_i(\hat\alpha,\hat\beta)\} - \mathbbm{P}\{f_i(\hat\alpha,\hat\beta)\}] + [\mathbbm{P}\{f_i(\hat\alpha,\hat\beta)\} - \mathbbm{P}\{f_i(\alpha,\beta)\}]. \label{appen:eq:thm1-proof3}
    \end{align}
    The first term in \eqref{appen:eq:thm1-proof3} can be controlled by Glivenko-Cantelli Theorem:
    \begin{align*}
        |\mathbbm{P}_n\{f_i(\hat\alpha,\hat\beta)\}| - \mathbbm{P}\{f_i(\hat\alpha,\hat\beta)\} \leq \sup_{\alpha, \beta}|(\mathbbm{P}_n - \mathbbm{P}) f_i(\alpha,\beta)| = o_p(1),
    \end{align*}
    because whichever of $\partial_\beta U_i(\alpha,\beta)$ or $U_i(\alpha,\beta)U_i(\alpha,\beta)^T$ that we set $f_i$ to, it takes value in a Glivenko-Cantelli class as $(\alpha,\beta)$ is finite-dimensional and takes value in a compact set (Assumption \ref{appen:thm1-assumption1}). The second term in \eqref{appen:eq:thm1-proof3}, $\mathbbm{P}\{f_i(\hat\alpha,\hat\beta)\} - \mathbbm{P}\{f_i(\alpha,\beta)\}$, is $o_p(1)$ by the dominated convergence theorem because of $(\hat{\alpha}, \hat{\beta}) \stackrel{p}{\to} \left(\alpha', \beta^*\right)$ and Assumption \ref{appen:thm1-assumption3}. Therefore, \eqref{appen:eq:thm1-proof3} is $o_p(1)$ for when $f_i = \partial_\beta U_i$ and when $f_i = U_i U_i^T$. Thus \eqref{appen:eq:thm1-proof1} and \eqref{appen:eq:thm1-proof2} are proven. This completes the proof of Theorem \ref{thm:CAN}.
    }

\end{proof}

\section{More Details on the Efficiency Gain by pd-EMEE}

\label{appen:sec:explain-efficiency-gain}

\subsection{Elaboration on the Intuition Behind the Efficiency Gain of pd-EMEE}

This subsection elaborates the discussion of the efficiency gain of pd-EMEE over EMEE in Section \ref{sec:method}. Let us examine how the EMEE weighs the $(i,t)$-decision point by $W_{it}'$ and later contrast with our estimator. Case 1: suppose for some $t+1 \leq j \leq t+\Delta-1$, $A_{ij} = 1$ (Table \ref{tab:weight-comparison}, Examples 1a and 1b). Because this is different from the treatment trajectory for the excursion, the $(i,t)$-decision point is assigned weight $W_{it}' = 0$ by EMEE. Case 2: suppose $A_{i,t+1}=\cdots=A_{i,t+\Delta-1} = 0$ (Table \ref{tab:weight-comparison}, Examples 2a and 2b). In this case, $W_{it}' > 1$ and the $(i,t)$-decision point is upweighted to account for itself and similar decision points that fall into Case 1 and thus have weight 0.

The key innovation in our proposed estimator is $W_{it}$, which improves efficiency over EMEE. We call $W_{it}$ the \textit{per-decision inverse probability weight}, and $\hat\beta$ the pd-EMEE method, because the inclusion of each of the $(\Delta-1)$ IPW terms is determined by each decision point: for each $j$ from $t+1$ to $t+\Delta-1$, if $\max_{1\leq s\leq j-t}R_{i,t+s}=0$ then $\frac{\mathbbm{1}(A_{ij}=0)}{1-p_{j}\left(H_{ij}\right)}$ is included in $W_{it}$; otherwise it is not included. Therefore, $W_{it}$ includes in its product the IPW terms from $j=t+1$ to right before the first occurrence of some $R_{j^*} = 1$. By not including the IPW terms after this $R_{j^*} = 1$, the per-decision IPW does not account for the treatment assignments after $j^*$. This is appropriate because the maximum property \eqref{eq:maximum-property} implies that the treatment assignments after $j^*$ don't affect the value of $Y_{t,\Delta}$ (which has to equal 1). To see how our estimator improves efficiency over the EMEE, let us now examine how our estimator weighs the $(i,t)$-decision point by $W_{it}$. For Case 1 described above, if some $R_j = 1$ occurs before the first $A_k = 1$ (Table \ref{tab:weight-comparison}, Example 1b), the $(i,t)$-decision point will have a positive weight. For Case 2 described above, if the first $R_{j^*} = 1$ occurs for $j^* < t+\Delta-1$ (Table \ref{tab:weight-comparison}, Example 2b), then $W_{it}$ will include fewer terms than $W_{it}'$ and thus be smaller, because $ W_{it} = \prod_{j=t+1}^{j^*}\frac{1}{1-p_{j}\left(H_{ij}\right)} < \prod_{j=t+1}^{t+\Delta-1}\frac{1}{1-p_{j}\left(H_{ij}\right)} = W_{it}'$. When the proximal outcome window is shorter than the distance between two decision points so that $\Delta = 1$, $\hat\beta$ is equivalent to EMEE.

To summarize, by capitalizing on the maximum property \eqref{eq:maximum-property}, the per-decision IP weight $W_{it}$ makes use of some decision points that would be otherwise discarded because $W_{it}' = 0$ (such as Example 1b) and reduces the weight of some other decision points (such as Example 2b), and this leads to the efficiency improvement. The efficiency improvement of $\hat\beta$ over EMEE is larger when a larger proportion of the decision points belong to the above two types, which happens when the proximal outcome window length $\Delta$ is large, when the randomization probability is large, or when the baseline success probability (i.e., the success probability of the proximal outcome under no treatment) is large. 

\begin{table}
  \caption{Comparison between per-decision IP weight $W_{it}$ in our estimator and standard IP weight $W_{it}'$ in EMEE in a few examples for the case of $\Delta = 3$. Subscript $i$ for all variables are omitted.}\label{tab:weight-comparison}
  \centering
  % \small
  % \begin{tabular}{lcccccccll}
  \begin{tabular}{@{}cccccccccc@{}}
  \toprule
   & $R_{t}$ & $A_{t+1}$ & $R_{t+1}$ & $A_{t+2}$ & $R_{t+2}$ & $Y_{t,\Delta}$ & & $W_t$ & $W_t'$ \\
  \midrule
  Example 1a & 0 & 0 & 0 & 1 & 0 & 0 & & 0 & 0 \\
  Example 1b & 0 & 0 & 1 & 1 & 0 & 1 & & $\frac{1}{1-p_{t+1}\left(H_{t+1}\right)}$ & 0 \\
  Example 2a & 0 & 0 & 0 & 0 & 0 & 0 & & $\frac{1}{1-p_{t+1}\left(H_{t+1}\right)} \cdot \frac{1}{1-p_{t+2}\left(H_{t+2}\right)}$ & $\frac{1}{1-p_{t+1}\left(H_{t+1}\right)} \cdot \frac{1}{1-p_{t+2}\left(H_{t+2}\right)}$ \\
  Example 2b & 0 & 0 & 1 & 0 & 0 & 1 & & $\frac{1}{1-p_{t+1}\left(H_{t+1}\right)}$ & $\frac{1}{1-p_{t+1}\left(H_{t+1}\right)} \cdot \frac{1}{1-p_{t+2}\left(H_{t+2}\right)}$  \\
  % \lastline
  \bottomrule
  \end{tabular}

  % \begin{tabular}{lcccccc}
  % \toprule
  % & & $j = t+1$ & $j = t+2$ & & $W_t$ & $W_t'$ \\
  % \midrule
  % & $A_{ij}$ & 0 & 1 \\
  % \multirow{-2}{*}{\centering Case 1(a)} & $R_{ij}$ & 0 & 0 & & \multirow{-2}{*}{\centering 0} & \multirow{-2}{*}{\centering 0} \\
  % \\
  % & $A_{ij}$ & 0 & 1 \\
  % \multirow{-2}{*}{\centering Case 1(b)} & $R_{ij}$ & 1 & 0 & & \multirow{-2}{*}{\centering 0} & \multirow{-2}{*}{\centering 0} \\
  % \\
  % & $A_{ij}$ & 0 & 0 \\
  % \multirow{-2}{*}{\centering Case 2(a)} & $R_{ij}$ & 0 & 0 & & \multirow{-2}{*}{\centering 0} & \multirow{-2}{*}{\centering 0} \\
  % \\
  % & $A_{ij}$ & 0 & 0 \\
  % \multirow{-2}{*}{\centering Case 2(b)} & $R_{ij}$ & 1 & 0 & & \multirow{-2}{*}{\centering 0} & \multirow{-2}{*}{\centering 0} \\
  % \bottomrule
  % \end{tabular}
\end{table}

\subsection{A Theoretical Analysis of the Relative Efficiency}
\label{appen:subsec:re}

To understand the drivers of the efficiency gain, here we analytically derive the relative efficiency between the proposed method (we call it ``pd-EMEE'' where ``pd'' stands for per-decision) and the EMEE method \citep{qian2021estimating} under a set of simplifying assumptions. In particular, we assume that individuals are always available ($I_{it} \equiv 1$), the randomization probability $p_{t}(H_{it})$ is a constant $p$, the moderator $S_{it} = 1$ (i.e., interested in the fully marginal effect), and $\tilde{p}_{t}(S_{it})$ is chosen to be equal to $p$. Furthermore, we consider the case $T=1$ with $\Delta \geq 1$ so for each individual there is only one proximal outcome $Y$ but $\Delta$ treatment occasions and $\Delta$ sub-outcomes $R_2,\ldots,R_{\Delta+1}$. We assume that all the sub-outcomes are independent with all the treatments, the $\Delta$ sub-outcomes for the same individual are independent with each other, and the sub-outcomes are identically distributed: $R_t$ follows Bernoulli($q$) distribution for some $q$.

We consider a simplified variation of both the pd-EMEE and the EMEE methods. In both estimators, we assume that the nuisance parameter $\alpha$ is not estimated and is set to 0. In other words, no control variables are used in either estimator. This helps to further simplify the theoretical derivation while preserving the essence of this efficiency comparison, because both estimators are consistent for any control variables.

Under these simplifying assumptions, the estimating function for pd-EMEE ($m_1$) and the estimating function for EMEE ($m_2$) for the fully marginal causal excursion effect $\beta$ are as follows:
\begin{align*}
    m_1(\beta) & = e^{-A_1 \beta} Y \prod_{j=2}^{\Delta} \frac{\mathbbm{1}(A_j = 0)}{1-p}^{\mathbbm{1} (\max_{2 \leq s \leq j}R_j = 0)} (A_1-p), \\
    m_2(\beta) & = e^{-A_1 \beta} Y \prod_{j=2}^{\Delta} \frac{\mathbbm{1}(A_j = 0)}{1-p} (A_1-p),
\end{align*}
and $\prod_{j=2}^{\Delta}(\cdot) := 1$ if $\Delta = 1$.

We derive in Section \ref{appen:subsubsec:re} that under these simplifying assumptions, the asymptotic relative efficiency between the two estimators is (AVar denoting asymptotic variance)
\begin{align}
    \text{Rel. Eff.} = \frac{\text{AVar}(\hat\beta^{\text{EMEE}})}{\text{AVar}(\hat\beta^{\text{pd-EMEE}})} =
    \begin{cases}
        \frac{1 - (1-q)^\Delta}{(1-p)^\Delta - (1-q)^\Delta} \times \frac{q-p}{q} & \text{if } p \neq q, \\
        \frac{1 - (1 - q)^\Delta}{q\Delta (1-p)^{\Delta-1}} & \text{if } p = q.
    \end{cases} \label{appen:eq:rel-eff}
\end{align}
The relative efficiency \eqref{appen:eq:rel-eff} equals 1 if $\Delta = 1$ and increases with $\Delta$. When $\Delta > 1$, the relative efficiency increases with $p$ and $q$ (Figure \ref{fig:app-re}). This confirms our intuition: When $p$ is larger, additional treatments are more likely to occur in the proximal outcome window; when $q$ is larger, sub-outcomes are more likely to be equal to 1 in the proximal outcome window; when $\Delta$ is larger, additional treatments are more likely to occur in the proximal outcome window and sub-outcomes are more likely to be equal to 1 in the proximal outcome window. In all these situations, pd-EMEE can achieve a greater reduction on the terms in the IP weight product, thus improving efficiency.

\begin{figure}[htbp]
    \centering
    \begin{subfigure}{.32\textwidth}
        \centering
        \includegraphics[width=1\linewidth]{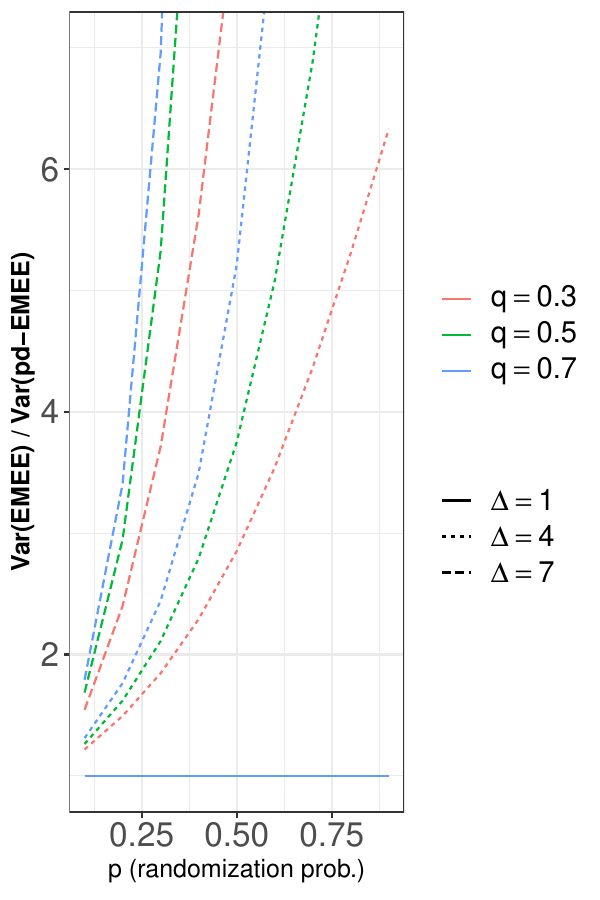}
        % \caption{$n$ vs. $\ate$}
        % \label{subfig:app-1-n-v-ate}
    \end{subfigure}
    \begin{subfigure}{.32\textwidth}
        \centering
        \includegraphics[width=1\linewidth]{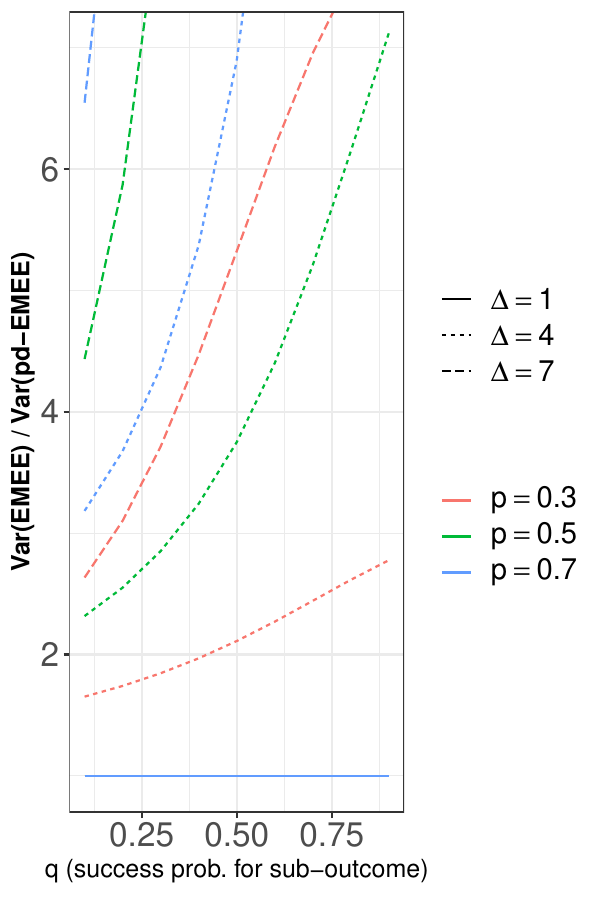}
        % \caption{$n$ vs. $\aspn$}
        % \label{subfig:app-1-n-v-aspn}
    \end{subfigure}
    \begin{subfigure}{.32\textwidth}
        \centering
        \includegraphics[width=1\linewidth]{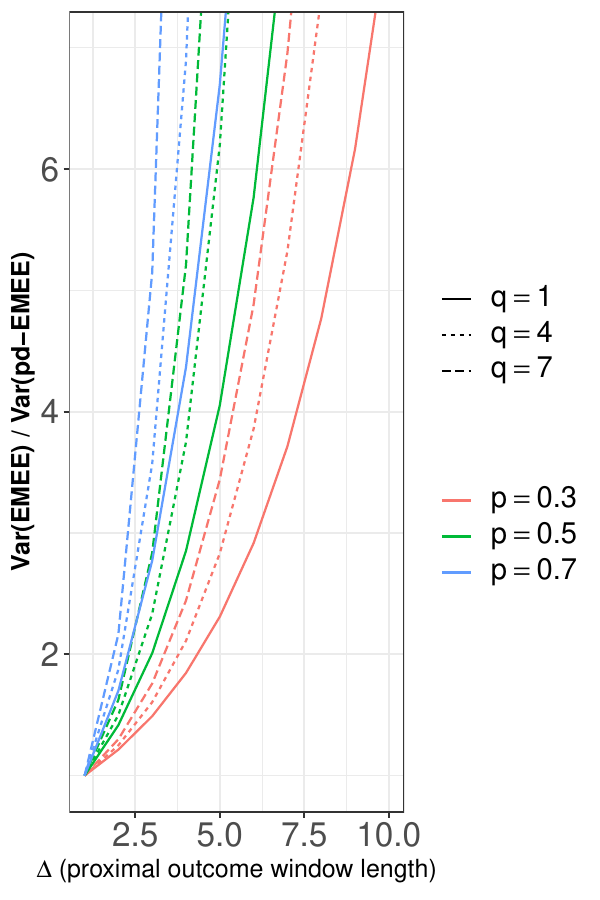}
        % \caption{Contour plot for $n$}
        % \label{subfig:app-1-contour}
    \end{subfigure}
    \caption{Relative efficiency \eqref{appen:eq:rel-eff} against $p$, $q$, and $\Delta$. In each of the first two plots, the three lines for $\Delta=1$ are all flat and overlap with each other.}
    \label{fig:app-re}
\end{figure}

\subsubsection{Proof of Equation \eqref{appen:eq:rel-eff}}

\label{appen:subsubsec:re}

Suppose all the simplifying assumptions stated in Section \ref{appen:subsec:re} hold.

\textbf{Step 1:} We derive the asymptotic variance of $\hat\beta^{\text{pd-EMEE}}$ that solves the following estimating equation:
\begin{align*}
    m_1(\beta) & = e^{-A_1 \beta} Y \prod_{j=2}^{\Delta} \left\{\frac{\mathbbm{1}(A_j = 0)}{1-p}\right\}^{\mathbbm{1} (\max_{2 \leq s \leq j}R_j = 0)} (A_1-p).
\end{align*}
The asymptotic variance equals
\begin{align*}
    \text{AVar}(\hat\beta^{\text{pd-EMEE}}) = \{E(\dot{m}_1)\}^{-1} E(m_1 m_1^T) \{E(\dot{m}_1)\}^{-1,T}.
\end{align*}
We derive each of the terms in the product.

\textbf{Step 1.1:} Deriving $E(\dot{m}_1)$. We have
\begin{align*}
    \dot{m}_1(\beta) = -e^{-A_1\beta} Y \prod_{j=2}^{\Delta}\left\{ \frac{\mathbbm{1}(A_j = 0)}{1-p}\right\}^{\mathbbm{1} (\max_{2 \leq s \leq j}R_j = 0)} A_1 (A_1-p).
\end{align*}
For $2 \leq u \leq \Delta$, let $B_u$ denote the event $R_2 = \cdots = R_u = 0$ and $R_{u+1} = 1$; let $B_1$ denote the event $R_2 = 1$; and let $B_0$ denote the event that $R_2 = \cdots = R_{\Delta+1} = 0$. Note that $\{B_t\}_{0 \leq t \leq \Delta}$ forms a partition of the whole sample space, and that $Y = 0$ if and only if $B_0$ holds. We have
\begin{align*}
    E(\dot{m}_1) & = - \sum_{t=0}^\Delta E\left[ e^{-A_1\beta} Y \prod_{j=2}^{\Delta}\left\{ \frac{\mathbbm{1}(A_j = 0)}{1-p}\right\}^{\mathbbm{1} (\max_{2 \leq s \leq j}R_j = 0)} A_1 (A_1-p) \mid B_t\right] P(B_t) \\
    & = - \sum_{t=1}^\Delta E\left[ e^{-A_1\beta} \prod_{j=2}^{t}\left\{ \frac{\mathbbm{1}(A_j = 0)}{1-p}\right\} A_1 (A_1-p)\right] P(B_t) \\
    & = - \sum_{t=1}^\Delta E\{ e^{-A_1\beta}  A_1 (A_1-p)\} P(B_t) \\
    & = - E\{ e^{-A_1\beta}  A_1 (A_1-p)\} \sum_{t=1}^\Delta P(B_t) \\
    & = - E\{ e^{-A_1\beta}  A_1 (A_1-p)\} \{1-P(B_0)\} \\
    & = - E\{ e^{-A_1\beta}  A_1 (A_1-p)\} \{1-(1-q)^\Delta\},
\end{align*}
where the first equality follows from the fact that $\{B_t\}_{0 \leq t \leq \Delta}$ forms a partition of the whole sample space, the second equality follows from the fact that $Y = 0$ if and only if $B_0$ holds, the third equality follows from the law of iterated expectation and the fact that $p$ is the randomization probability of $A_j$, and the last equality follows from the assumption of i.i.d. $R_j$.

\textbf{Step 1.2:} Deriving $E(m_1 m_1^T)$. Because $S_{it} = 1$, both $\beta$ and $m_1$ are scalars and thus $E(m_1 m_1^T) = E(m_1^2)$. We have
\begin{align*}
    E(m_1 m_1^T) & = E \left[ e^{-2A_1 \beta} Y \prod_{j=2}^{\Delta}\left\{ \frac{\mathbbm{1}(A_j = 0)}{(1-p)^2}\right\}^{\mathbbm{1} (\max_{2 \leq s \leq j}R_j = 0)} (A_1-p)^2 \right] \\
    & = \sum_{t=0}^\Delta E\left[ e^{-2A_1\beta} Y \prod_{j=2}^{\Delta}\left\{ \frac{\mathbbm{1}(A_j = 0)}{(1-p)^2}\right\}^{\mathbbm{1} (\max_{2 \leq s \leq j}R_j = 0)} (A_1-p)^2 \mid B_t \right] P(B_t) \\
    & = \sum_{t=1}^\Delta E\left[ e^{-2A_1\beta} \prod_{j=2}^{t}\left\{ \frac{\mathbbm{1}(A_j = 0)}{(1-p)^2}\right\} (A_1-p)^2\right] P(B_t) \\
    & = \sum_{t=1}^\Delta E\left[ e^{-2A_1\beta} (A_1-p)^2\right] P(B_t) / (1-p)^{t-1} \\
    & = \sum_{t=1}^\Delta E\{ e^{-2A_1\beta} (A_1-p)^2\} \frac{(1-q)^{t-1}q}{(1-p)^{t-1}} \\
    & = 
    \begin{cases}
        E\{ e^{-2A_1\beta} (A_1-p)^2\} q \frac{1 - \left( \frac{1-q}{1-p} \right)^\Delta}{1 - \left( \frac{1-q}{1-p} \right)} & \text{if } p\neq q \\
        E\{ e^{-2A_1\beta} (A_1-p)^2\} q \Delta  & \text{if } p = q,
    \end{cases}
\end{align*}
where the second equality follows from the fact that $\{B_t\}_{0 \leq t \leq \Delta}$ forms a partition of the whole sample space, the third equality follows from the fact that $Y = 0$ if and only if $B_0$ holds, the fourth equality follows from the law of iterated expectation and the fact that $p$ is the randomization probability of $A_j$, and the fifth equality follows from the assumption of i.i.d. $R_j$.

\textbf{Step 1.3:} Putting together the above two steps and we have 
\begin{align}
    \text{AVar}(\hat\beta^{\text{pd-EMEE}}) & = \{E(\dot{m}_1)\}^{-1} E(m_1 m_1^T) \{E(\dot{m}_1)\}^{-1,T} \nonumber \\
    & = (1 - (1-q)^\Delta)^{-2} \cdot [E\{ -e^{-A_1\beta} A_1 (A_1-p)\}]^{-2} \cdot E \{ e^{-2A_1 \beta} (A_1-p)^2 \} \cdot L, \label{eq:appen:avar-pdemee}
\end{align}
where
\begin{align*}
    L = 
    \begin{cases}
        q \frac{1 - \left( \frac{1-q}{1-p} \right)^\Delta}{1 - \left( \frac{1-q}{1-p} \right)} & \text{if } p\neq q \\
        q \Delta  & \text{if } p = q.
    \end{cases}
\end{align*}

\textbf{Step 2:} We derive the asymptotic variance of $\hat\beta^{\text{EMEE}}$ that solves the following estimating equation:
\begin{align*}
    m_2(\beta) & = e^{-A_1 \beta} Y \prod_{j=2}^{\Delta} \left\{\frac{\mathbbm{1}(A_j = 0)}{1-p}\right\} (A_1-p).
\end{align*}
The asymptotic variance equals
\begin{align*}
    \text{AVar}(\hat\beta^{\text{EMEE}}) = \{E(\dot{m}_2)\}^{-1} E(m_2 m_2^T) \{E(\dot{m}_2)\}^{-1,T}.
\end{align*}
We derive each of the terms in the product.

\textbf{Step 2.1:} Deriving $E(\dot{m}_2)$. We have
\begin{align*}
    \dot{m}_2(\beta) = -e^{-A_1\beta} Y \prod_{j=2}^{\Delta}\left\{ \frac{\mathbbm{1}(A_j = 0)}{1-p}\right\} A_1 (A_1-p).
\end{align*}
We have
\begin{align*}
    E(\dot{m}_2) & = -E(Y) E\left[ -e^{-A_1\beta} \prod_{j=2}^{\Delta}\left\{ \frac{\mathbbm{1}(A_j = 0)}{1-p}\right\} A_1 (A_1-p)\right] \\
    & = -E(Y) E\{ -e^{-A_1\beta} A_1 (A_1-p)\} \\
    & = -(1 - (1-q)^\Delta) E\{ -e^{-A_1\beta} A_1 (A_1-p)\},
\end{align*}
where the first equality follows from the assumption that all the sub-outcomes are independent with all the treatments and thus $Y$ is independent of all the $A_j$'s, the second equality follows from the law of iterated expectation and the fact that $p$ is the randomization probability of $A_j$, and the last equality follows from the assumption of i.i.d. $R_j$.

\textbf{Step 2.2:} Deriving $E(m_2 m_2^T)$. Because $S_{it} = 1$, both $\beta$ and $m_2$ are scalars and thus $E(m_2 m_2^T) = E(m_2^2)$. We have
\begin{align*}
    E(m_2 m_2^T) & = E \left[ e^{-2A_1 \beta} Y \prod_{j=2}^{\Delta}\left\{ \frac{\mathbbm{1}(A_j = 0)}{(1-p)^2}\right\} (A_1-p)^2 \right] \\
    & = E(Y) E \left[ e^{-2A_1 \beta} \prod_{j=2}^{\Delta}\left\{ \frac{\mathbbm{1}(A_j = 0)}{(1-p)^2}\right\} (A_1-p)^2 \right] \\
    & = E(Y) E \{ e^{-2A_1 \beta} (A_1-p)^2 \}  / (1-p)^{\Delta-1} \\
    & = (1 - (1-q)^\Delta) E \{ e^{-2A_1 \beta} (A_1-p)^2 \} / (1-p)^{\Delta-1},
\end{align*}
where the second equality follows from the assumption that all the sub-outcomes are independent with all the treatments and thus $Y$ is independent of all the $A_j$'s, the third equality follows from the law of iterated expectation and the fact that $p$ is the randomization probability of $A_j$, and the last equality follows from the assumption of i.i.d. $R_j$.

\textbf{Step 2.3:} Putting together the above two steps and we have 
\begin{align}
    \text{AVar}(\hat\beta^{\text{EMEE}}) & = \{E(\dot{m}_2)\}^{-1} E(m_2 m_2^T) \{E(\dot{m}_2)\}^{-1,T} \nonumber \\
    & = (1 - (1-q)^\Delta)^{-1} \cdot [E\{ -e^{-A_1\beta} A_1 (A_1-p)\}]^{-2} \cdot E \{ e^{-2A_1 \beta} (A_1-p)^2 \} / (1-p)^{\Delta-1}. \label{eq:appen:avar-emee}
\end{align}

Finally, combining \eqref{eq:appen:avar-pdemee} and \eqref{eq:appen:avar-emee} yields \eqref{appen:eq:rel-eff}. This completes the proof.

\section{Connection to the Literature}
\label{appen:sec:literature} 

\subsection{Connection between Per-Decision IPW (Our Approach) and the Per-Decision Importance Sampling}
\label{appen:subsec:IPW-literature-rl}

% Here we describe two related IPW techniques in the literature and how our per-decision IPW approach connects to and differs from each of them.

% The first IPW technique we discuss here is per-decision importance sampling (PDIS) in offline reinforcement learning \citet{precup2000eligibility}. 

We describe a related IPW technique in the literature, the per-decision importance sampling (PDIS) in offline reinforcement learning \citet{precup2000eligibility}, and discuss how our per-decision IPW approach connects to and differs from PDIS. To make connections to our approach, we describe a simple example of PDIS using a setting similar to MRT. See \citet[Section 2]{liu2020understanding} for the general form of PDIS. Suppose there are $\Delta$ decision points, and at each decision point a treatment is randomized with randomization probability $p_t$ and a short-term outcome $R_{t+1}$ is observed. Let a distal outcome (an outcome of interest) be $Y := \sum_{t=1}^\Delta R_{t+1}$. Suppose the goal is to estimate $E\{Y (a, \bar{0}_{\Delta-1})\}$. Then we have
\begin{align}
    E\{Y(a, \bar{0}_{\Delta-1})\} & = E\left[ \left\{\frac{\mathbbm{1}(A_1 = a)}{ap_1 + (1-a)(1-p_1)} \prod_{s=2}^\Delta \frac{\mathbbm{1}(A_s = 0)}{1 - p_s} \right\}\sum_{t=1}^\Delta R_{t+1}\right] \label{appen:eq:ipw-1}\\
    & = E\left[ \sum_{t=1}^\Delta  \left\{\frac{\mathbbm{1}(A_1 = a)}{ap_1 + (1-a)(1-p_1)} \prod_{s=2}^t \frac{\mathbbm{1}(A_s = 0)}{1 - p_s} \right\} R_{t+1}\right], \label{appen:eq:ipw-2}
\end{align}
where \eqref{appen:eq:ipw-1} is the identification result using the usual IPW and \eqref{appen:eq:ipw-2} is the identification result using PDIS. The intuition behind \eqref{appen:eq:ipw-2} is that earlier outcomes cannot depend on later treatments. The same intuition is used in our approach to derive the identification result for proximal outcomes satisfying the maximum property; i.e.,
\begin{align}
    & E\{Y_{t,\Delta}(\bar{A}_{t-1}, a, \bar{0}_{\Delta-1}) \mid S_t(\bar{A}_{t-1}), I_t(\bar{A}_{t-1}) = 1\} \nonumber \\ & = E\left(E\left\{\prod_{j=t+1}^{t+\Delta-1}\left[\frac{\mathbbm{1}(A_{j}=0)}{1-p_{j}\left(H_{j}\right)}\right]^{\mathbbm{1}(\max_{1\leq s\leq j-t}R_{t+s}=0)} Y_{t,\Delta}\mid A_{t}=a,H_{t},I_{t}=1\right\} \mid S_t, I_t = 1 \right). \label{appen:eq:ipw-3}
\end{align}
The difference between our approach and PDIS is that (i) the PDIS literature typically considers the outcome to be a sum of sub-outcomes whereas we considers a binary outcome that is the maximum of sub-outcomes, (ii) the PDIS literature considers a single outcome at the end whereas we considered the setting where the main outcome itself is longitudinal (i.e., we consider $Y_{t,\Delta}$ for $1 \leq t \leq T$), and (iii) the PDIS identification result \eqref{appen:eq:ipw-2} is solely based on IPW, whereas our identification result \eqref{appen:eq:ipw-3} is a combination of IPW and standardization.

\subsection{Connection between Our Estimating Equation and the Literature}
\label{appen:subsec:ee-literature}

Our estimating function is defined as 
\begin{equation}
    U_i(\alpha,\beta) = \sum_{t=1}^T I_{it} e^{- A_{it} S_{it}^T \beta } \epsilon_{it} (\alpha,\beta) M_{it} W_{it} \begin{bmatrix}
        g(H_{it}) \\           
        \{A_{it} - \tilde{p}_t(S_{it})\}S_{it} 
    \end{bmatrix}. \label{appen:eq:ee}
\end{equation}
The role of each term is discussed in Section \ref{sec:method}. Per the request of a reviewer, here we explicate the derivation of our estimating function.

Besides the per-decision IP weight, $W_{it}$, which was explained in the previous subsection, our estimating function $U_i(\alpha,\beta)$ resembles several existing estimating functions in the literature: the multiplicative structural nested mean model (MSNMM) \citep{robins1994correcting}, the estimator for causal excursion effect (EMEE) \citep{qian2021estimating}, and the weighted and centered least squares (WCLS) estimator \citep{boruvka2018assessing}.

In particular, the WCLS estimating function for the case with continuous proximal outcome and $\Delta=1$ is
\begin{equation*}
    \sum_{t=1}^T I_{it} \tilde\epsilon_{it} (\alpha,\beta) M_{it} \begin{bmatrix}
        g(H_{it}) \\           
        \{A_{it} - \tilde{p}_t(S_{it})\}S_{it} 
    \end{bmatrix},
\end{equation*}
with $\tilde\epsilon_{it}(\alpha,\beta) = Y_{it,1} - g(H_{it}^T\alpha - \{A_{it} - \tilde{p}_t(S_{it})\}S_{it}^T\beta$. Our estimating function \eqref{appen:eq:ee} adopts the weighting (by $M_{it}$) and the centering of the treatment (by $\tilde{p}_t(S_{it})$), which were used to (i) estimate a marginal causal effect instead of a fully conditional causal effect, and (ii) make the estimator robust to misspecified control variables (the $g_t(H_{it})$ part).

To extend the WCLS estimating function to binary outcome, the EMEE estimator \citep{qian2021estimating} used a blipping-down technique of MSNMM developed by \citet{robins1994correcting}, i.e., pre-multiplying the residual term $\epsilon_{it}(\alpha,\beta)$ by a negative causal effect (to remove the causal effect, i.e., ``blip'' the outcome down to its counterpart under no treatment). We borrowed this idea from EMEE. This is where the term $e^{- A_{it} S_{it}^T \beta}$ comes from. Note that the original blipping-down technique \citet{robins1994correcting} depends on an exponent representing the fully conditional causal effect (such as $e^{- A_{it} f(H_{it})^T \theta}$), but we used a blipping-down factor that represents the marginal causal effect because this is the causal effect of interest in our case.

The maximum property may resemble the structure of a discrete-time survival outcome, and may raise the question of whether the proximal outcome should be treated as a survival outcome instead of a binary one. We believe the choice of the model and outcome type should depend on the scientific question. For some interventions, the goal is so that a desirable event, such as self-monitoring in Drink Less and activity bout in HeartStep, occurs within a reasonable time window but not necessarily as soon as possible. For such settings, the causal effect on the binary outcome will contain relevant information for assessing and improving the intervention. For other settings, the goal might be to reduce the risk of a dangerous or harmful behavior as soon as possible, in which case the causal effect on the survival outcome might be more relevant. An example of the latter is problem anger prevention \citep{metcalf2022anger}, where the intervention is to reduce the risk of aggressive behavior once anger is detected.

% \subsection{Connection between Per-Decision IPW (Our Approach) and Discrete-Time Survival Analysis}
% \label{appen:sec:IPW-literature-survival}

% \begin{itemize}
%     \item Describe the identifiability result for cox-MSM (how each sample is weighted, and if things can be rewritten to connect to our weights.)
%     \item Describe the identifiability result for survival outcome (see my project with Walter). Make it a single time point case.
% \end{itemize}

% The second IPW technique we discuss here is the Cox marginal structural model (Cox-MSM) for estimating the causal effect of time-varying treatments on a survival outcome \citet{robins2000marginal}. Again, we describe a simple example of Cox-MSM to connnect to our approach. Let $Y$ be the survival outcome. Suppose there are $\Delta$ decision points at time $1$, $2$, \ldots, $\Delta$, and at each decision point a treatment is randomized with randomization probability $p_1$, $p_2$, \ldots, $p_\Delta$. Let $R_{s+1}$ denote whether an event occurs between $s$ and $s+1$. To fit a Cox proportional hazards model for estimating the causal effect of the treatment sequence $()$

\section{Technical Details for the pd-EMEE2 Estimator}
\label{appen:sec:projection}

In Section \ref{sec:projection-estimator} we proposed an estimator $\tilde\beta$ that leverages the projection idea in semiparametric efficiency theory to further improve the efficiency. We present the technical details for the estimator derivation and the asymptotic theory here.

\subsection{Derivation of $\tilde\beta$}

Here we derive the analytic form of $\sum_{u=1}^T [E\{\xi_i(\beta) \mid H_{iu}, A_{iu}\} - E\{\xi_i(\beta) \mid H_{iu}\}]$ mentioned in Section \ref{sec:projection-estimator}, which directly leads to the estimating equation for $\tilde\beta$.

The following working assumption facilitates a closed form projection term.

\begin{asu}[Absence of delayed effect]
    \label{appen:est2_asu}
    For any $(t,u)$ pair such that $1 \leq u < t \leq T$, $\xi_{it}(\beta)$ is independent of $A_{iu}$ given $H_{iu}$. 
\end{asu}

The following theorem gives the form of $\sum_{u=1}^T [E\{\xi_i(\beta) \mid H_{iu}, A_{iu}\} - E\{\xi_i(\beta) \mid H_{iu}\}]$.

\begin{thm}
    \label{appen:thm:projection-form}
    Under Assumption \ref{appen:est2_asu} and Assumptions \ref{assumption1}--\ref{assumption4}, we have
    \begin{align}
      & ~~~~ \xi_i(\beta) - \sum_{u=1}^T [E\{\xi_i(\beta) \mid H_{iu}, A_{iu}\} - E\{\xi_i(\beta) \mid H_{iu}\}] \nonumber \\
      & = \sum_{t=1}^T I_{it} e^{- A_{it} S_{it}^T \beta}  M_{it}\{A_{it} - \tilde{p}_t(S_{it})\} S_{it} \bigg[Y_{it, \Delta} W_{it} - E(Y_{it, \Delta}W_{it}|H_{it}, A_{it}) \nonumber \\
      & ~~~~ - \sum_{u = t+1}^{t+\Delta-1} \left\{ E(Y_{it, \Delta}W_{it}|H_{iu}, A_{iu}) - E(Y_{it, \Delta}W_{it}|H_{iu}) \right\} \bigg] \nonumber \\
        &+\sum_{t=1}^T I_{it} \tilde{p}_t(S_{it}) \{1 - \tilde{p}_t(S_{it})\} S_{it} \left\{e^{- S_{it}^T \beta} E(Y_{it, \Delta}W_{it}|H_{it}, A_{it} = 1) - E(Y_{it, \Delta}W_{it}|H_{it}, A_{it} = 0) \right\}. \label{appen:eq:eqs2_star}
    \end{align}
\end{thm}

\subsubsection{Proof of Theorem \ref{appen:thm:projection-form}}

We have
\begin{align}
    & ~~~~\xi_i(\beta) - \sum_{u=1}^T [E\{\xi_i(\beta) \mid H_{iu}, A_{iu}\} - E\{\xi_i(\beta) \mid H_{iu}\}] \nonumber \\
    & = \sum_{t=1}^T \left[ \xi_{it}(\beta)- \sum_{u = t}^{t+\Delta - 1} \left\{ E(\xi_{it}(\beta)|H_{iu}, A_{iu}) - E\{\xi_{it}(\beta)|H_{iu}\} \right\} \right], \label{appen:est2_proof1}
\end{align}
because when $u < t$, $E(Y_{it, \Delta}W_{it}|H_{iu}, A_{iu}) - E(Y_{it, \Delta}W_{it}|H_{iu}) = 0$ from Assumption \ref{appen:est2_asu}; when $u \geq t + \Delta$, $\xi_t(\beta)$ is deterministic from $H_u$. 

When $u = t$, $E\{\xi_{it}(\beta)|H_{iu}\}$ is calculated as following:
\begin{align}
    E\{\xi_{it}(\beta)|H_{iu}\} &= 
    I_{it} S_{it} E\left(e^{- A_{it} S_{it}^T \beta}  M_{it}\{A_{it} - \tilde{p}_t(S_{it})\} Y_{it, \Delta}W_{it}|H_{it}\right) \notag\\
    &=I_{it} e^{-S_{it}^T \beta}\frac{\tilde{p}_t(S_{it})}{p_t(H_{it})} \{1 - \tilde{p}_t(S_{it})\} S_{it} E(Y_{it, \Delta}W_{it}|H_{it}, A_{it} = 1)p_t(H_{it}) \notag\\
    &- I_{it} \frac{1-\tilde{p}_t(S_{it})}{1-p_t(H_{it})} \tilde{p}_t(S_{it}) S_{it}  E(Y_{it, \Delta}W_{it}|H_{it}, A_{it} = 0)(1- p_t(H_{it})) \label{appen:est2_proof2} \\
    & = I_{it}S_{it}  \tilde{p}_t(S_{it}) \{1 - \tilde{p}_t(S_{it})\} \left\{e^{- S_{it}^T \beta} E(Y_{it, \Delta}W_{it}|H_{it}, A_{it} = 1) - E(Y_{it, \Delta}W_{it}|H_{it}, A_{it} = 0)\right\}. \label{appen:est2_proof2.1}
\end{align}
\eqref{appen:est2_proof2} follows from the fact that $e^{- A_{it} S_{it}^T \beta}$, $M_{it}$ and $A_{it} - \tilde{p}_t(S_{it})$ are $(H_{it}, A_{it})$-measurable. \eqref{appen:est2_proof2.1} follows from rearranging items in \eqref{appen:est2_proof2}. Therefore, we also have
\begin{align}
    E(\xi_{it}(\beta)|H_{iu}, A_{iu}) = I_{it} e^{-A_{it}S_{it}^T \beta} M_{it} \{A_{it} - \tilde{p}_t(S_{it})\} S_{it} E(Y_{it, \Delta}W_{it}|H_{it}, A_{it}) \label{appen:est2_proof3}
\end{align}

Similarly, when $t <u \leq t+\Delta - 1$, $Y_{it, \Delta}$ and $W_{it}$ are not $(H_u, A_u)$-measurable, and $E(Y_{it, \Delta}W_{it}|H_{iu}, A_{iu}) - E(Y_{it, \Delta}W_{it}|H_{iu})$ can be written as
\begin{align}
    E(Y_{it, \Delta}W_{it}|H_{iu}, A_{iu})  &= I_{it} e^{-A_{it}S_{it}^T \beta} M_{it} \{A_{it} - \tilde{p}_t(S_{it})\} S_{it} E(Y_{it, \Delta}W_{it}|H_{iu}, A_{iu}) \label{appen:est2_proof4}\\
    &\text{and} \notag\\   
    E(Y_{it, \Delta}W_{it}|H_{iu}) &= I_{it} e^{-A_{it}S_{it}^T \beta} M_{it} \{A_{it} - \tilde{p}_t(S_{it})\} S_{it} E(Y_{it, \Delta}W_{it}|H_{iu}) \label{appen:est2_proof5}
\end{align}
By plugging \eqref{appen:est2_proof2.1}, \eqref{appen:est2_proof3}, \eqref{appen:est2_proof4} and 
\eqref{appen:est2_proof5} in \eqref{appen:est2_proof1} we have
\begin{align}
    & ~~~~ \xi_i(\beta) - \sum_{u=1}^T [E\{\xi_i(\beta) \mid H_{iu}, A_{iu}\} - E\{\xi_i(\beta) \mid H_{iu}\}] \nonumber \\
    = & \sum_{t=1}^T I_{it} e^{- A_{it} S_{it}^T \beta}  M_{it}\{A_{it} - \tilde{p}_t(S_{it})\} S_{it} \left[Y_{it, \Delta} W_{it} - \sum_{u = t+1}^{t+\Delta-1} \left\{ E(Y_{it, \Delta}W_{it}|H_{iu}, A_{iu}) - E(Y_{it, \Delta}W_{it}|H_{iu}) \right\} \right] \nonumber \\
    &- \sum_{t=1}^T I_{it} e^{- A_{it} S_{it}^T \beta}  M_{it}\{A_{it} - \tilde{p}_t(S_{it})\} S_{it} E(Y_{it, \Delta}W_{it}|H_{it}, A_{it}) \nonumber  \\
    &+\sum_{t=1}^T I_{it}S_{it}  \tilde{p}_t(S_{it}) \{1 - \tilde{p}_t(S_{it})\} \{e^{- S_{it}^T \beta} E(Y_{it, \Delta}W_{it}|H_{it}, A_{it} = 1) - E(Y_{it, \Delta}W_{it}|H_{it}, A_{it} = 0)\}. \nonumber 
\end{align}
Rearranging terms and we get \eqref{appen:eq:eqs2_star}.
%&=  \sum_{t=1}^T \left[ \xi_{it}(\beta)- \sum_{u = t}^{t+\Delta - 1} \left\{ E(\xi_{it}(\beta)|H_{iu}, A_{iu}) - E\{\xi_{it}(\beta)|H_{iu}\} \right\} \right] \notag\\
%&- \sum_{t=1}^T \left[\sum_{u = 1}^{t - 1} \left\{ E(\xi_{it}(\beta)|H_{iu}, A_{iu}) - E\{\xi_{it}(\beta)|H_{iu}\} \right\} - \sum_{u = t+ \Delta}^{T} \left\{ E(\xi_{it}(\beta)|H_{iu}, A_{iu}) + E\{\xi_{it}(\beta)|H_{iu}\} \right\}\right]\\

% \\ &=  \sum_{t=1}^T I_{it} e^{- A_{it} S_{it}^T \beta}  M_{it}\{A_{it} - \tilde{p}_t(S_{it})\} S_{it} \left[Y_{it, \Delta} W_{it} - \sum_{u = t}^{t+\Delta-1} \left\{ E(Y_{it, \Delta}W_{it}|H_{iu}, A_{iu}) - E(Y_{it, \Delta}W_{it}|H_{iu}) \right\} \right] \label{appen:est2_proof2}\\ & = \sum_{t=1}^T I_{it} e^{- A_{it} S_{it}^T \beta}  M_{it}\{A_{it} - \tilde{p}_t(S_{it})\} S_{it} \left[Y_{it, \Delta} W_{it} - \sum_{u = t+1}^{t+\Delta-1} \left\{ E(Y_{it, \Delta}W_{it}|H_{iu}, A_{iu}) - E(Y_{it, \Delta}W_{it}|H_{iu}) \right\} \right] \notag \\ & -  \sum_{t=1}^T I_{it} e^{- A_{it} S_{it}^T \beta}  M_{it}\{A_{it} - \tilde{p}_t(S_{it})\} S_{it} \left\{ E(Y_{it, \Delta}W_{it}|H_{it}, A_{it}) - E(Y_{it, \Delta}W_{it}|H_{it}) \right\}

\subsection{Proof of Theorem \ref{thm:CAN2}}
\label{subsec:appen:proof-ee2}

Due to Theorem 5.1 in \citet{cheng2023efficient}, it suffices to show that $\tilde{U}_i(\beta,\mu)$ is globally robust with respect to the nuisance parameter $\mu$; i.e., $E\{\tilde{U}_i(\beta^*,\mu)\} = 0$ for all $\mu$.

We have
\begin{align*}
    \tilde{U}_i(\beta,\mu) = I_1(\beta) + I_2(\beta, \mu) + I_3(\beta,\mu),
\end{align*}
where
\begin{align*}
    I_1(\beta) & = \sum_{t=1}^T I_{it} e^{- A_{it} S_{it}^T \beta}  M_{it}\{A_{it} - \tilde{p}_t(S_{it})\} S_{it} Y_{it, \Delta} W_{it}, \\
    I_2(\beta, \mu) & = - \sum_{t=1}^T I_{it} e^{- A_{it} S_{it}^T \beta}  M_{it}\{A_{it} - \tilde{p}_t(S_{it})\} S_{it} \\ 
    & ~~~~ \times \sum_{u = t+1}^{t+\Delta-1} \Big\{ \mu_{u-t}(H_{iu}, A_{iu}) - p_u(H_{iu})\mu_{u-t}(H_{iu}, 1) - (1-p_u(H_{iu}))\mu_{u-t}(H_{iu}, 0) \Big\}, \\
    I_3(\beta, \mu) & = \sum_{t=1}^T I_{it} S_{it} \bigg[ - e^{-A_{it} S_{it}^T\beta} M_{it} \{A_{it} - \tilde{p}_{it}(S_{it})\} \mu_0(H_{it}, A_{it}) + \tilde{p}_t(S_{it}) \{1 - \tilde{p}_t(S_{it})\} \Big\{ e^{- S_{it}^T \beta} \mu_0(H_{it}, 1) - \mu_0(H_{it}, 0) \Big\} \bigg].
\end{align*}
That $E\{I_1(\beta^*)\} = 0$ follows from Lemma \ref{lem:consistency}. That $E\{I_2(\beta^*,\mu)\} = 0$ for all $\mu$ follows from the law of iterated expectation and the fact that 
\begin{align*}
    E\{\mu_{u-t}(H_{iu}, A_{iu}) \mid H_{iu}\} = p_u(H_{iu})\mu_{u-t}(H_{iu}, 1) - (1-p_u(H_{iu}))\mu_{u-t}(H_{iu}, 0) \quad \text{for any } \mu.
\end{align*}
That $E\{I_3(\beta^*,\mu)\} = 0$ for all $\mu$ follows from the law of iterated expectation and the fact that
\begin{align*}
    & ~~~~ E[e^{-A_{it}S_{it}^T\beta} M_{it} \{A_{it} - \tilde{p}_{it}(S_{it})\} \mu_0(H_{it}, A_{it}) \mid H_{it}] \\
    & = E[e^{-A_{it}S_{it}^T\beta} M_{it} \{A_{it} - \tilde{p}_{it}(S_{it})\} \mu_0(H_{it}, A_{it}) \mid H_{it}, A_{it} = 1] p_t(H_{it}) \\
    & ~~~ + E[e^{-A_{it}S_{it}^T\beta} M_{it} \{A_{it} - \tilde{p}_{it}(S_{it})\} \mu_0(H_{it}, A_{it}) \mid H_{it}, A_{it} = 0] \{ 1 - p_t(H_{it})\} \\
    & = \tilde{p}_t(S_{it}) \{1 - \tilde{p}_t(S_{it})\} \Big\{ e^{- S_{it}^T \beta} \mu_0(H_{it}, 1) - \mu_0(H_{it}, 0) \Big\}.
\end{align*}

This completes the proof.

\section{Generative model} 
\label{appen:sec:dgm}

\subsection{Generative Model used in Simulation}

We set $I_t \equiv 1$ and $T = 100$. The generative model depends on two parameters, $\Delta$ (length of the proximal outcome window) and $p_a$ (constant randomization probability). $A_t \in \{0,1\}$ was generated from $\text{Bernoulli}(p_a)$. A scalar time-varying covariate, $Z_t \in \{0, 1, 2\}$, was generated independently of all previous variables (i.e., $Z_t \perp \{Z_s,A_s,R_{s+1}: 1 \leq s \leq t-1\}$) with
\begin{align}
    P(Z_t = 0) = \gammaplaceholder^{-1/2\Delta}/C, \quad
    P(Z_t = 1) = 1/C, \quad
    P(Z_t = 2) = \gammaplaceholder^{1/2\Delta}/C, \label{eq:dgm-Zt}
\end{align}
where $C = \gammaplaceholder^{-1/2\Delta} + \gammaplaceholder^{1/2\Delta} + 1$.
The sub-outcome $R_{t+1}$ given $A_t$ and $H_t$ was generated as
\begin{align}
    P(R_{t+1} = 0 \mid A_t = 0, H_t) & = \gammaplaceholder^{(1.5-0.5Z_t)/\Delta}, \label{eq:dgm-R-A0}\\
    P(R_{t+1} = 0 \mid A_t = 1, H_t) & = \frac{1- \{1-\gammaplaceholder^{(1.5-0.5Z_t)/\Delta}\cdot (3/C\cdot \gammaplaceholder^{1/\Delta})^{\Delta-1}\} \cdot e^{0.1+0.2Z_t}}{(3/C \cdot \gammaplaceholder^{1/\Delta})^{\Delta-1}}. \label{eq:dgm-R-A1}
\end{align}
We set $R_{t+1} = 0$ for $ T < t \leq T + \Delta$.
The proximal outcome was generated as $Y_{t,\Delta} = \max (R_{t+1}, \ldots, R_{t+\Delta})$. 

%appen:sec:dgm
We will show in the next subsection that for $1 \leq t \leq T - \Delta$, 
\begin{align}
    & E\{Y_{t,\Delta}(\bar{A}_{t-1}, a_t, \bar{0}_{\Delta-1}) \mid H_t(\bar{A}_{t-1}), I_t(\bar{A}_{t-1}) = 1\} = \{1 - \gammaplaceholder^{(\Delta+0.5-0.5Z_{t})/\Delta}\cdot(3/C)^{\Delta-1}\} \cdot e^{a_t(0.1+0.2Z_t)}. \label{eq:appen:dgm-Y}
\end{align}
Therefore, $\est(Z_t) = 0.1 + 0.2 Z_t$ for $1 \leq t \leq T - \Delta$ under the generative model, and this holds for all $\Delta \geq 1$ and $p_a \in (0,1)$. 
%\tq{7/12: Yihan, please check the correctness of my following statement.}\yb{should be correct!}
For $T - \Delta < t < T$, $\est(Z_t)$ is slightly different but its influence on the simulation results is negligible because in the simulations $\Delta \ll T$.

There is substantial variability in $R_{t+1}$ and $Y_{t,\Delta}$. When we vary $\Delta$ from 2 to 10 and vary $Z_t \in \{0,1,2\}$,
% and set $\gammaplaceholder = 0.5$,
\eqref{eq:dgm-R-A0} varies from 0.59 to 0.97, \eqref{eq:dgm-R-A1} varies from 0.38 to 0.80, and \eqref{eq:appen:dgm-Y} varies from 0.41 to 0.58 when setting $a_t = 0$ and from 0.57 to 0.80 when setting $a_t = 1$. 

We constructed this generative model based on three considerations.
First, in order to fairly assess how the efficiency improvement is affected by $\Delta$ and $p_a$, we needed $\est(Z_t)$ to not depend on $\Delta$ or $p_a$. 
%As discussed in Section \ref{sec:method}, a larger $\Delta$ or a smaller $p_a$ would lead to a larger efficiency improvement. To assess this relationship in the simulations, we need a parsimonious $\est(Z_t)$ that does not depend on $\Delta$ or $p_a$. This way, when we vary $\Delta$ or $p_a$, the resulting change in the efficiency improvement is not confounded by the potential change in $\est(Z_t)$. 
Thus, the distribution of the sub-outcome $R_t$ would depend on $\Delta$ because there are $\Delta$ sub-outcomes in $Y_{t,\Delta}$.
Second, we wanted to construct a complicated $E\{Y_{t,\Delta}(\bar{A}_{t-1}, 0, \bar{0}_{\Delta-1}) \mid H_t, I_t = 1\}$ to demonstrate the robustness of the estimator against misspecified working models.
Third, we wanted $Z_t$ to depend on $\Delta$ as little as possible. In a real application, $Z_t$ (a covariate) would not depend on $\Delta$ (a choice in the analysis). However, we found it difficult to achieve while keeping the form of $\est(Z_t)$ simple. The current generating distribution of $Z_t$ \eqref{eq:dgm-Zt} is the least dependent on $\Delta$ that we came up with, and the three probabilities in \eqref{eq:dgm-Zt} are all very close to 1/3 for $\Delta \geq 2$.

\readd{

\subsection{Proof of Equation \ref{eq:appen:dgm-Y}}
We show that
\begin{equation*}
    E(Y_{t,\Delta}(\bar{A}_{t-1}, a, \bar{0}_{\Delta-1}) \mid H_t(\bar{A}_{t-1}), I_t(\bar{A}_{t-1}) = 1) = \{1 -  \gamma^{(\Delta + 0.5(1-Z_t))/\Delta} \cdot (3/C)^{\Delta-1}\} \cdot e^{a(0.1+0.2Z_t)}.
\end{equation*}
\begin{proof}
    Define $\phi_0(Z_t)$ and $\phi_1(Z_t)$ as follows:
    \begin{align}
         \phi_0(Z_t) := \gamma^{(1.5-0.5Z_t)/\Delta} = P(R_{t+1} = 0 \mid A_t = 0, H_t) \nonumber
    \end{align}
    and
    \begin{align}
         \phi_1(Z_t) := \frac{1-\{1-\gamma^{(1.5-0.5Z_t)/\Delta}\cdot (3/C\cdot \gamma^{1/\Delta})^{\Delta-1}\} \cdot \exp\{0.1+0.2Z_t\}}{(3/C \cdot \gamma^{1/\Delta})^{\Delta-1}} = P(R_{t+1} = 0 \mid A_t = 1, H_t) \nonumber
    \end{align}
    where $C = \gamma^{-1/2\Delta} + \gamma^{1/2\Delta} + 1$. (In the generative model, the feasible value of $\gamma$ is subject to the constraint that $\phi_0(Z_t)$ and $\phi_1(Z_t)$ are bounded by $(0,1)$.)
    Then we have the following equation for $E(Y_{t,\Delta}(\bar{A}_{t-1}, 1, \bar{0})\mid H_t)$:
    \begin{align}
        & E(Y_{t,\Delta}(\bar{A}_{t-1}, 1, \bar{0})\mid H_t) = P(Y_{t,\Delta}(\bar{A}_{t-1}, 1, \bar{0}) = 1 \mid H_t) \nonumber\\
        & = 1 - P(R_{t+1}(\bar{A}_{t-1}, 1, \bar{0}) = 0, R_{t+2}(\bar{A}_{t-1}, 1, \bar{0}) = 0, \ldots, R_{t+\Delta}(\bar{A}_{t-1}, 1, \bar{0}) = 0 \mid H_t)
        \label{eq:dm1} \\
        & = 1 - P(R_{t+1}(\bar{A}_{t-1}, 1, \bar{0}) = 0 \mid H_t) \cdot P(R_{t+2}(\bar{A}_{t-1}, 1, \bar{0}) = 0 \mid H_t) \cdots P(R_{t+\Delta}(\bar{A}_{t-1}, 1, \bar{0}) = 0 \mid H_t), \label{eq:dm2}
    \end{align}
    where \eqref{eq:dm1} follows from Assumption \ref{assumption4}, and \eqref{eq:dm2} follows from the fact that $R_{t+1}, \ldots, R_{t+\Delta}$ are independent of each other in the generative model. Similarly, for $E(Y_{t,\Delta}(\bar{A}_{t-1}, 0, \bar{0})$ we have
    \begin{align}
        & E(Y_{t,\Delta}(\bar{A}_{t-1}, 0, \bar{0})\mid H_t) \nonumber\\
        & = 1 - P(R_{t+1}(\bar{A}_{t-1}, 0, \bar{0}) = 0 \mid H_t) \cdot P(R_{t+2}(\bar{A}_{t-1}, 0, \bar{0}) = 0 \mid H_t) \cdots P(R_{t+\Delta}(\bar{A}_{t-1}, 0, \bar{0}) = 0 \mid H_t). \label{eq:dm2.1}
    \end{align}

    For any $1 \leq j \leq \Delta - 1$ and for any $a \in \{0,1\}$ we have
    \begin{align}
        & P\left(R_{t+j+1}\left(A_{t-1}, a, \bar{0}\right)=0 \mid H_t\right)\nonumber \\
        \mbox{(Assumption \ref{assumption2})} & = P\left(R_{t+j+1}\left(A_{t-1}, a, \bar{0}\right)=0 \mid H_t, A_t = a\right)\nonumber \\
        \mbox{(law of iterated expectation)} & =E\left\{E\left(\mathbbm{1}\left(R_{t+j+1}\left(\bar{A}_{t-1},a, \overline{0}\right)=0\right) \mid H_{t+j}\right) \mid H_t, A_t = a\right\} \nonumber\\
        & =E\{E\left(\mathbbm{1}\left(R_{t+j+1}=0\right)\left|H_{t+j}, \left.A_{t+j}=0\right)\right| H_t, A_t = a\right\} \label{eq:dm2.2}\\
        \mbox{(definition of $\phi_0$)} & = E\{\phi_0(Z_{t+j}) \mid H_t , A_t = a\} \nonumber\\
        \mbox{($\phi_0(Z_{t+j})$ is independent of $H_t,A_t$)} & = E\{\phi_0(Z_{t+j})\} \\
        & = \gamma^{(1.5-0.5\cdot 0)/\Delta} \cdot \gamma^{-1/2\Delta}/C + \gamma^{(1.5-0.5\cdot 1)/\Delta} \cdot 1/C + \gamma^{(1.5-0.5\cdot 2)/\Delta} \cdot \gamma^{1/2\Delta}/C \nonumber\\
        & = 3/C \cdot  \gamma^{1/\Delta}, \label{eq:dm2.5}
    \end{align}
    where \eqref{eq:dm2.2} follows from Assumption \ref{assumption1} and the fact that $R_{t+j+1}$ only depends on $(Z_{t+j},A_{t+j})$ in the generative model.
    Plugging \eqref{eq:dm2.5} into \eqref{eq:dm2} and we have
    \begin{align}
        & E(Y_{t,\Delta}(\bar{A}_{t-1}, 1, \bar{0})\mid H_t) \nonumber\\
        % & = 1 - \phi_1(Z_t) \prod_{j=1}^{\Delta-1} E\{\phi_0(Z_{t+j})\} \nonumber\\
        & = 1 - \phi_1(Z_t) \cdot \{3/C \cdot \gamma^{1/\Delta}\}^{\Delta-1}\nonumber \\
        & = \{1- \gamma^{(1.5-0.5Z_t)/\Delta}\cdot (3/C\cdot \gamma^{1/\Delta})^{\Delta-1}\} \cdot \exp\{0.1+0.2Z_t\} \nonumber\\
        & = \{1 - \gamma^{(\Delta+0.5-0.5Z_{t})/\Delta}\cdot(3/C)^{\Delta-1}\} \cdot \exp\{0.1+0.2Z_t\} \label{eq:dm3}
    \end{align}
    Plugging \eqref{eq:dm2.5} into \eqref{eq:dm2.1} and we have
    \begin{align}
        & E(Y_{t,\Delta}(\bar{A}_{t-1}, 0, \bar{0})\mid H_t) \nonumber\\
        % & = 1 - \phi_0(Z_t) \prod_{j=1}^{\Delta-1} E\{\phi_0(Z_{t+j})\} \nonumber\\
        & = 1 - \phi_0(Z_t) \cdot \{3/C \cdot 0.5^{1/\Delta}\}^{\Delta-1}\nonumber \\
        & = 1 -  \gamma^{(1.5-0.5Z_t)/\Delta} \cdot \{3/C \cdot \gamma^{1/\Delta}\}^{\Delta-1} \nonumber\\
        & = 1 - \gamma^{(\Delta+0.5-0.5Z_{t})/\Delta}\cdot(3/C)^{\Delta-1} \label{eq:dm4}
    \end{align}
    Based on \eqref{eq:dm3} and \eqref{eq:dm4}, we have:
    \begin{align*}
        E(Y_{t,\Delta}(\bar{A}_{t-1}, A_t, \bar{0})\mid H_t, A_t) = \{1 - \gamma^{(\Delta+0.5-0.5Z_{t})/\Delta}\cdot(3/C)^{\Delta-1}\} \cdot \exp\{A_t(0.1+0.2Z_t)\}
    \end{align*}
\end{proof}

\section{Generalized Estimating Equations (GEE) Used in Simulation}
\label{appen:sec:gee}

In the simulations, we considered two estimators based on generalized estimating equations (GEE): GEE with independent correlation structure (GEE.ind), and GEE with exchangeable correlation structure (GEE.exch). Here for concreteness we write down the mean and the variance models for both GEEs, and write down the estimating equation for GEE.ind.

Because we are considering causal effects on the relative risk scale, the GEE assumes the following mean model using a log link: $E(Y_{t,\Delta}\mid H_{t},A_{t})=e^{g(H_{t})^{T}\alpha+A_{t}S_{t}^{T}\beta}=:\mu_{t}(\alpha,\beta)$. Because $Y_{t,\Delta}$ is a binary variable, its variance is associated with the mean model: $\text{var}(Y_{t,\Delta}\mid H_{t},A_{t}) = \mu_{t}(\alpha,\beta)\{1 - \mu_{t}(\alpha,\beta)\}$. The estimating equation can be constructed following the original GEE proposal in \citet{liang1986longitudinal}. Below we write down the estimating equation for GEE.ind:
\begin{align*}
    \mathbbm{P}_{n} \sum_t^T \frac{\partial\mu_{t}(\alpha,\beta)}{\partial(\alpha,\beta)^T} \text{var} (Y_{t,\Delta}\mid H_{t},A_{t})^{-1}\{Y_{t,\Delta}-\mu_{t}(\alpha,\beta)\} = 0,
\end{align*}
which, by the form of $\mu_{t}(\alpha,\beta)$ and $\text{var}(Y_{t,\Delta}\mid H_{t},A_{t})$, simplifies to
\begin{align*}
    \mathbbm{P}_{n}\sum_{t} \left\{1-e^{g(H_{t})^{T}\alpha+A_{t}S_{t}^{T}\beta} \right\}^{-1} \left\{Y_{t,\Delta}-e^{g(H_{t})^{T}\alpha+A_{t}S_{t}^{T}\beta}\right\}\begin{bmatrix}g(H_{t})\\
        A_{t}S_{t}
        \end{bmatrix} = 0.
\end{align*}
Here we used $\mathbbm{P}_{n}$ to denote the empirical average over all individuals. In the simulation, the particular mean model used is $E(Y_{t,\Delta}\mid H_{t},A_{t}) = e^{\alpha_0 + \alpha_1 Z_t + \beta A_t}$.

}

\readd{

\section{Generalization to Other Reference Regimes}
\label{appen:sec:reference-regime}

In the main paper, we considered the causal excursion effect contrasting treatment trajectories $(\bar{A}_{t-1}, 1, \bar{0}_{\Delta-1})$ and $(\bar{A}_{t-1}, 0, \bar{0}_{\Delta-1})$, where the future reference regime ($A_{t+1},\ldots,A_{t+\Delta-1}$) in the two trajectories are all set to 0. That is, we consider excursions starting from decision point $t$ that receives treatment (or no treatment) at $t$, and then receive no treatment for the next $\Delta-1$ decision points.

As pointed out by a reviewer, the choice of the reference treatment regime depends on the scientific question, and in some applications other reference regimes may be of interest. For example, in \citet{boruvka2018assessing} they considered causal excursion effect contrasting treatment trajectories $(\bar{A}_{t-1}, 1, \bar{A}_{t+1:t+\Delta-1})$ and $(\bar{A}_{t-1}, 1, \bar{A}_{t+1:t+\Delta-1})$, where $\bar{A}_{t+1:t+\Delta-1}:= (A_{t+1}, A_{t+2}, \ldots, A_{t+\Delta-1})$. That is, they considered excursions starting from decision point $t$ that receive treatment (or no treatment) at $t$, and then switch back to follow the MRT policy for the next $\Delta-1$ decision points.

We generalize our method to reference regimes that are in between the reference regime we considered in the main paper and the reference regime considered in \citet{boruvka2018assessing}. That is, excursions starting from decision point $t$ that receive treatment (or no treatment) at $t$, then receive no treatment for the next $K$ decision points (with some $0\leq K \leq \Delta-1$), then follow the MRT policy for the next $\Delta - 1 - K$ decision points.

Formally, consider the following causal excursion effect for $0 \leq K \leq \Delta - 1$:
\begin{equation}
    \est^{(K)} \left\{S_t(\bar{A}_{t-1}) \right\}=\log \frac{E\left\{Y_{t, \Delta}(\bar{A}_{t-1}, 1, \bar{0}_K, \bar{A}_{t+K+1:t+\Delta - 1}) \mid S_t(\bar{A}_{t-1}), I_t(\bar{A}_{t-1})=1\right\}}{E\left\{Y_{t, \Delta}(\bar{A}_{t-1}, 0, \bar{0}_K, \bar{A}_{t+K+1:t+\Delta - 1}) \mid S_t(\bar{A}_{t-1}), I_t(\bar{A}_{t-1})=1\right\}}.
    \label{appen:eq:CEE-other-reference-regime}
\end{equation}
Under the same set of causal assumptions as in Section 2.4, we have the identification result:
\begin{align}
    \est^{(K)}\left\{S_t(\bar{A}_{t-1})\right\} 
    = \log \frac{E \bigg[ E\left\{\prod_{j=t+1}^{t+K}\left[\frac{\mathbbm{1}(A_{j}=0)}{1-p_{j}\left(H_{j}\right)}\right]^{\mathbbm{1}(\max_{1\leq s\leq j-t}R_{t+s}=0)} Y_{t,\Delta}\mid A_{t}=1,H_{t},I_{t}=1\right\}  \mid S_t, I_t = 1 \bigg]} {E \bigg[ E\left\{ \prod_{j=t+1}^{t+K}\left[\frac{\mathbbm{1}(A_{j}=0)}{1-p_{j}\left(H_{j}\right)}\right]^{\mathbbm{1}(\max_{1\leq s\leq j-t}R_{t+s}=0)} Y_{t,\Delta}\mid A_{t}=0,H_{t},I_{t}=1\right\}  \mid S_t, I_t = 1 \bigg]}.
    \label{appen:eq:identification-other-reference-regime}
\end{align}
The difference between the two identification results \eqref{appen:eq:identification-other-reference-regime} and (4) is that in \eqref{appen:eq:identification-other-reference-regime} the product of weights is only up to decision point $t+K$ instead of up to $t+\Delta-1$, because after decision point $t+K$ the treatment trajectory comes back to the MRT protocol and thus no additional weight is needed.

Suppose we impose a linear model analogous to (5): $\est^{(K)}(S_t) = S_t^T \beta^{(K)}$. To estimate $\beta^{(K)}$, an estimating function analogous to (6) can be used by revising the weights:
\begin{equation}
    U_i^{(K)}(\alpha,\beta^{(K)}) = \sum_{t=1}^T I_{it} e^{- A_{it} S_{it}^T \beta^{(K)} } \epsilon_{it} (\alpha,\beta^{(K)}) M_{it} W_{it}^{(K)} \begin{bmatrix}
        g(H_{it}) \\           
        \{A_{it} - \tilde{p}_t(S_{it})\}S_{it} 
    \end{bmatrix},
    \label{appen:eq:estimatingeq-other-reference-regime}
\end{equation}
with
\begin{equation*}
    W_{it}^{(K)} = \prod_{j=t+1}^{t+K} \left\{\frac{\mathbbm{1}(A_{ij}=0)}{1-p_{j}(H_{ij})}\right\}^{\mathbbm{1}(\max_{1\leq s\leq j-t}R_{i,t+s}=0)}.
\end{equation*}
Using the same proof for Theorem 1, it is straightforward to show that consistency and asymptotic normality for $\hat\beta^{(K)}$ holds with the revised estimating function.

Below we present simulation results to compare this generalized version of our method (generalized pd-EMEE) and a generalized version of EMEE by \citet{qian2021estimating}, and see how the relative efficiency depends on the choice of reference regime. Here the EMEE by \citet{qian2021estimating} needs to be generalized to handle the reference regimes we consider here. In particular, for the generalized version of EMEE, the estimating function would be the same as \ref{appen:eq:estimatingeq-other-reference-regime} except that $W_{it}^{(K)}$ would be replaced by
\begin{equation*}
    \tilde{W}_{it}^{(K)} = \prod_{j=t+1}^{t+K} \left\{\frac{\mathbbm{1}(A_{ij}=0)}{1-p_{j}(H_{ij})}\right\}.
\end{equation*}
We used the same generative model as in Section \ref{sec:simulation} and set $\Delta = 10$. When $K$ varies from 0 to 9, the relative efficiency between the generalized pd-EMEE and the generalized EMEE increases from 1 to over 1.4 (Figure \ref{fig:appen-simulation-ref-regime}). This confirms the intuition that the farther the reference regime is from the MRT policy, the more efficiency gain there is when we use pd-EMEE over EMEE.

\begin{figure}[htbp]
    \begin{center}
    \caption{Relative efficiency between the generalized pd-EMEE and the generalized EMEE with various $K$, where $K$ parameterizes the reference regime. A relative efficiency of greater 1 means that the generalized pd-EMEE is more efficient.}
    \includegraphics[width = 0.5\textwidth]{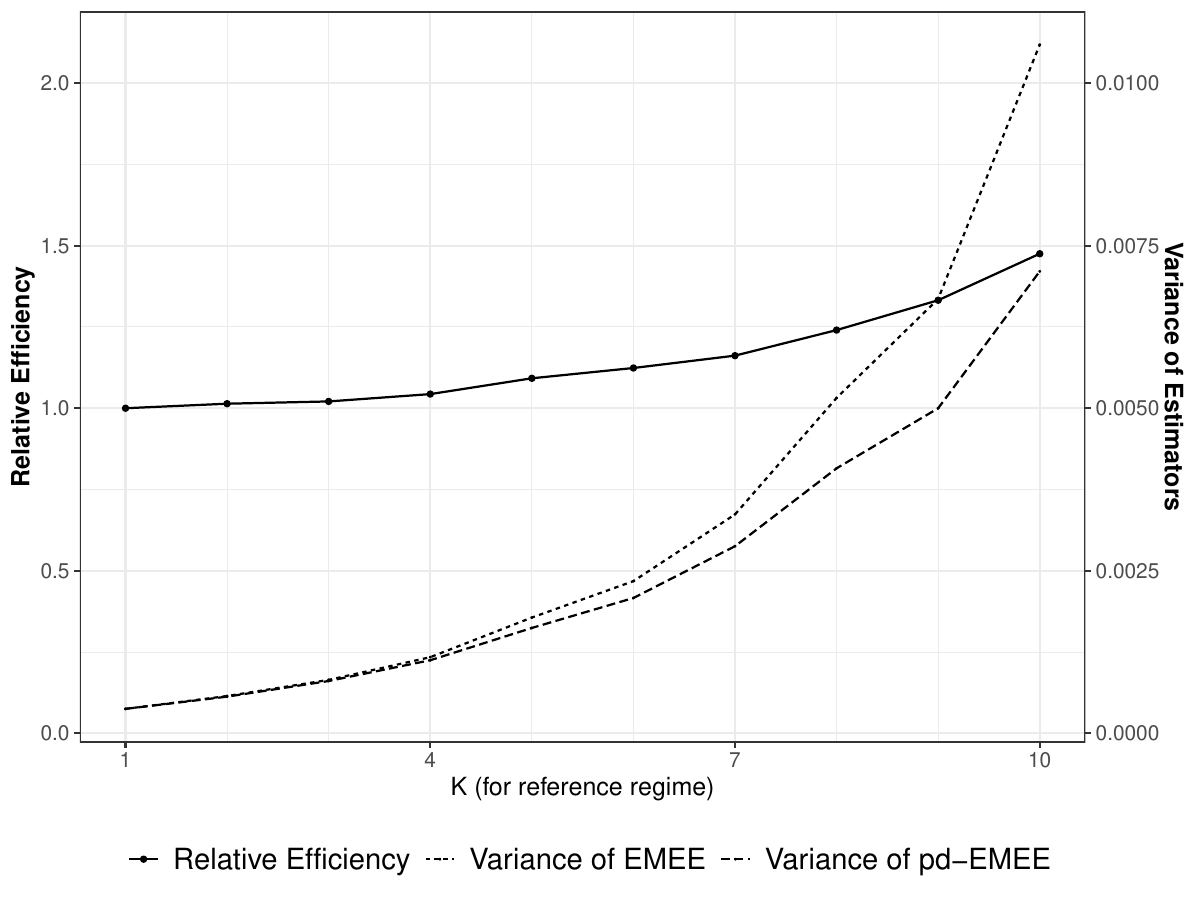}
    \label{fig:appen-simulation-ref-regime}
    \end{center}
\end{figure}
}

% \newpage
% \bibliographystyle{biorefs}
% \bibliography{refs}
% %\bibliography{pd-IPW-ref}

\end{document}